\documentclass{jcgt}

\usepackage{amssymb,latexsym,amsmath,amsfonts,amsthm}
\usepackage[linesnumbered,lined,boxed,commentsnumbered]{algorithm2e}
\usepackage{mathtools}
\DeclarePairedDelimiter{\ceil}{\lceil}{\rceil}
\DeclarePairedDelimiter{\floor}{\lfloor}{\rfloor}

\theoremstyle{plain}
\newtheorem{theorem}{Theorem}[section]

\setciteauthor{Alexandros D. Keros, Divakaran Divakaran, Kartic Subr}
\setcitetitle{Jittering Samples using a kd-Tree Stratification}

\begin{document}

\title{Jittering Samples using a kd-Tree Stratification}

\author{Alexandros D. Keros\\University of Edinburgh \and Divakaran Divakaran\\University of Edinburgh \and Kartic Subr \\University of Edinburgh}

%\teaser{
%  \includegraphics[width=2in]{teapot.png}
%  \label{fig:teaser}
%  \caption{A teapot teaser figure. A teaser figure is \textbf{mandatory} for JCGT.}
%}

\maketitle
%\thispagestyle{firstpagestyle}
%-------------------------------------------------------------------------
\begin{abstract}
\small
Monte Carlo sampling techniques are used to estimate high-dimensional integrals that model the physics of light transport in virtual scenes for computer graphics applications. These methods rely on the law of large numbers to estimate expectations via simulation, typically resulting in slow convergence. Their errors usually manifest as undesirable grain in the pictures generated by image synthesis algorithms. It is well known that these errors diminish when the samples are chosen appropriately. A well known technique for reducing error operates by subdividing the integration domain, estimating integrals in each \emph{stratum} and aggregating these values into a stratified sampling estimate. Na\"{i}ve methods for stratification, based on a lattice (grid) are known to improve the convergence rate of Monte Carlo, but require samples that grow exponentially with the dimensionality of the domain.

We propose a simple stratification scheme for $d$ dimensional hypercubes using the kd-tree data structure. Our scheme enables the generation of an arbitrary number of equal volume partitions of the rectangular domain, and $n$ samples can be generated in $O(n)$ time. Since we do not always need to explicitly build a kd-tree, we provide a simple procedure that allows the sample set to be drawn fully in parallel without any precomputation or storage,  speeding up sampling to $O(\log n)$ time per sample when executed on $n$ cores. If the tree is implicitly precomputed ($O(n)$ storage) the parallelised run time reduces to $O(1)$ on $n$ cores. In addition to these benefits, we provide an upper bound on the worst case star-discrepancy for $n$ samples matching that of lattice-based sampling strategies, which occur as a special case of our proposed method. We use a number of quantitative and qualitative tests to compare our method against state of the art samplers for image synthesis.
\end{abstract}

%-------------------------------------------------------------------------
\section{Introduction}

Photo-realistic visuals of virtual environments are generated by simulating the physics of light and its interaction within the virtual environments. The simulation of \textit{light transport} requires estimation of integrals over high-dimensional spaces, for which Monte Carlo integration is the method of choice. The defining characteristic of a Monte Carlo estimator is its choice of locations where the function to be integrated is evaluated. Given a fixed computational budget, the quality of the pictures rendered by this technique depends heavily on the choice of these \textit{sample} locations. Despite its many benefits, one of the problems of Monte Carlo integration is its relatively slow convergence of $O(1/\sqrt{n})$, given $n$ samples. 

 The primary expectation of a sampling algorithm is that it results in low error. Many theories such as equi-distribution~\cite{zaremba1968some} and spectral signatures~\cite{Subr16Fourier} underpin sampling choices. A simple way to improve equi-distribution is to partition the domain into homogeneous strata, and to aggregate the estimates within each of the strata. One popular variant of this type of \textit{stratified} sampling is \textit{jittered} sampling -- where the hypercube of random numbers is partitioned as a grid and a single sample is drawn from each cell. While stratification results in improved convergence, a more dramatic improvement is obtained when the samples optimise a measure of equi-distribution called \textit{discrepancy}~\cite{kuipers1975uniform,shirley91Discrepancy,owenmcbook}. 
 
 A second consideration is the time taken to generate samples. Random numbers and deterministic sequences~\cite{niederreiter1987point} are popular choices since they can be generated fast and in parallel. Deterministic samples provide the added advantage of repeatability, which is desirable during development and debugging. Finally, the impact of  the dimensionality of the integration domain is an important factor. Many sampling algorithms either suffer from disproportionately greater error or poor performance for high-dimensional domains.  Some algorithms  require sample sizes that are exponentially dependent on the dimensionality.~e.g.~jittered sampling requires that $n=k^d, k\in \mathbb{Z}$. To avoid such problems, modern renderers sacrifice stratification in the high-dimensional space by interpreting it as a chain of outer products of low-dimensional spaces (e.g. one or two),  which can be sampled independently, a technique known as ``padding". It has recently been shown that for samplers such as Halton~\cite{halton1964algorithm} or Sobol~\cite{sobol1967}, which are able to multiply-stratify  high-dimensional spaces, that they project undesirably to lower dimensional subspaces~\cite{jarosz19orthogonal}.  
 
 In this paper, we stratify the $d$-dimensional hypercube using a kd-tree. Our algorithm is simple, efficient and scales well (with $n$, $d$ and parallelisation). For any $n$, we design a kd-tree so that all leaves of the tree result to domain partitions that occupy the same volume. Then, we draw a single random value from each of the leaves. This can be viewed as a generalisation of jittered sampling. When $n=2^{kd}$ for some integer $k$, the cells of the kd-tree align perfectly with those of the regular grid. We derive the worst-case star-discrepancy of samples generated via  kd-tree stratification. Finally, we perform empirical quantitative and qualitative experiments comparing the performance of kd-tree stratification with  state of the art sampling methods.

\paragraph{Contributions}
In this paper, we propose a kd-tree stratification scheme with the following properties:
\begin{enumerate} 
	\item it generalises jittered sampling for arbirary sample counts; 
% 	\item propose $k$D padded versions, each consecutive sequence of dimensions $(\dots, d_i, d_{i+1}, \dots , d_{i+k}, \dots)$ satisfies the kd-tree stratification property
	\item the $i^\mathrm{th}$ of $n$ samples can be calculated independently;
	\item an upper bound on the star-discrepancy of our samples as  $2^{d-1} d n^{-\frac{1}{d}}$;
	\item error that is empirically at least as good as jittered sampling but similar to Halton and Sobol in many situations; and 
	\item it can easily be parallelised.
\end{enumerate}

\section{Related Work}
Although Monte Carlo integration is the \textit{de facto} statistical technique for estimating high dimensional integrals pertaining to light transport, it can be applied in a few different ways.~e.g.~path tracing, bidirectional path tracing, Markov Chains~\cite{Veach1998mclighttransport}, etc. The computer graphics literature is rich with sampling algorithms to reduce the variance of the estimates~\cite{christensen2016path} and a discussion of these techniques is beyond the scope of this paper. Here, we address a few closely related classes of works that are relevant to our proposed scheme for stratification.   

\paragraph{Assessing sample sets} Empirical comparisons of samplers on specific scenes is the most common method for assessment. Of the several image comparison metrics few are suited to high dynamic range images~\cite{mantiuk2007high}. Of those, many focus on assessing the quality of tone-mapping operators rather than their suitability for assessing noisy rendering. The \textit{de facto} choice of metric for assessing samplers in rendering is an adaptation of numerical mean squared error (MSE). Theoretical considerations such as equidistribution~\cite{zaremba1968some} and spectral properties~\cite{Durand11Afrequency,Subr16Fourier} provide a more general assessment of sample quality. A well known  measure for equidistribution of a point set in a domain is the maximum \textit{discrepancy}~\cite{kuipers1975uniform,shirley91Discrepancy,Dick10Digital} between the proportion of points falling into an arbitrary subset of the domain and the measure of that subset. Since this is not tractable across all subsets, a restricted version considers all sub-hyperrectangular boxes with one vertex at the origin. This so-called star-discrepancy can be used to bound error introduced by the point set when used in numerical integration of functions with certain properties~\cite{koksma1942een,aistleitner2014functions}. Although it is desirable to know the discrepancy of a sampling strategy, so that this bound may be known, it is generally non-trivial to derive. We derive an upper bound on the discrepancy of points resulting from our proposed sampling technique.  For the remainder of the paper, unless otherwise clarified, we will use discrepancy to refer to $L_{\infty}$ star discrepancy.

\paragraph{Stratified sampling} A powerful way to reduce the variance of sampling-based estimators is to partition the domain into strata with mutually disparate values, estimate integrals within each of the strata and then carefully aggregate them into a collective estimator~\cite{cochran1977sampling}. Unfortunately, \textit{stratified sampling} is challenging when there is insufficient information \textit{a priori} about the integrand. Typically, the domain is partitioned anyway, into strata of known measures, and proportional allocation is used to draw samples within them. In the simplest case, the $d$-dimensional hypercube is partitioned using a regular lattice (grid) and one random sample is drawn from each cell~\cite{haber1967modified}.  This is known as \textit{jittered sampling}~\cite{Cook:1984:DRT:964965.808590} in the graphics literature and has been shown to improve convergence. Although jittered sample is simple and parallelisable, it suffers from the curse of dimensionality~~\cite[Chapter 7.3, \S 3]{Pharr2010PBRT} ,~i.e.~it is only effective when the number of samples required is a perfect $d^{\mathrm{th}}$ power. This does not allow fine-grained control over the computational budget for large $d$. Various sampling strategies are based on such equal-measure stratification with proportional allocation, since it can potentially improve convergence and is never worse than not stratifying. Another example is n-rooks sampling~\cite{shirley91Discrepancy} or latin hypercube sampling~\cite{mckay1979comparison}, where stratification is performed  along multiple axes. Multi-jittered sampling~\cite{Chiu94Multi,tang1993orthogonal} combines jittered grids with latin hypercube sampling. Several tiling-based approaches~\cite{kopf2006recursive,ostromoukhov2004fast} have been proposed for generating sample distributions with desirable blue-noise characteristics. Although they produced blue-noise patterns that are useful in halftoning, stippling and image anti-aliasing, it is unclear -- based on recent theoretical connections between blue-noise and error and convergence rates~\cite{pilleboue2015variance} -- whether those methods are useful in building useful estimators. For the benefits of stratification to be realised in a practical setting, for example when the hypercube is mapped to arbitrary manifolds, the mapping needs to be constrained.~e.g.~it needs to satisfy area-preservation~\cite{arvo2001stratified}.

\paragraph{Low-discrepancy sequences} Infinite sequences in the unit hypercube whose star discrepancy is $O(log(n)^d/n)$ are known as \textit{low-discrepancy sequences}~\cite{niederreiter1987point}. One such sequence, the Van de Corput sequence~\cite{van1936verteilungsfunktionen} forms the core of state-of-the-art methods such as Halton~\cite{halton1964algorithm} and Sobol~\cite{sobol1967} samplers. These \textit{quasi-random} sequences~\cite{niederreiter1992random} can be used to produce deterministic samples that result in rapidly converging quasi Monte Carlo (QMC) estimators~\cite{Niederreiter92Quasi}. QMC has been shown to significantly speed up high-quality, offline rendering~\cite{keller1995quasi,Keller12Advanced}. QMC samplers, and their randomized variants~\cite{randomlyPermuterOwen2011} based on the more general concept of \emph{elementary intervals}, can be viewed as the ultimate form of equal-measure stratification, since they strive to stratify across all origin-anchored hyperrectangles in the domain simultaneously. In addition to low-discrepancy, infinite sequences have the additional desirable property that any prefix set of samples is well distributed. This is an active area of research and recent work includes methods for progressive multijitter~\cite{christensen2018progressive} and low-discrepancy samples with blue noise properties~\cite{Ahmed16Low}. There are two hurdles to using QMC samplers: the first, a technical issue, is that their effectiveness in high-dimensional problems is limited; the second, a legal issue, is that their use for rendering is patented~\cite{kellerpatent}. We propose a simple, parallelisable alternative that, in the worst case matches, and often surpasses the performance of alternatives in terms of MSE. 

\paragraph{kd-trees} Space partitioning data structures such as quadtrees, octrees and kd-trees~\cite{Friedman1977optimizedkdtree} are well known in computational geometry and computer graphics. These data structures are typically used to optimise location queries such as nearest-neighbour queries. In computer graphics, kd-trees have been used to speed up ray intersections~\cite{wald2006building}, optimise multiresolution frameworks~\cite{Goswami2013} and its variants have been used to perform high-dimensional filtering~\cite{adams2009gaussian}. They have been used for sampling in a variety of ways including importance sampling via domain warping~\cite{McCool1997probabilitykdtrees,clarberg2005wavelet}, progressive refinement for antialiasing~\cite{Painter1989antialisedraytracing}, efficient multiscale sampling from products of gaussian mixtures~\cite{ihler2004multiscalesampling} and optimisation of the sampling of mean free paths~\cite{yue2010unbiasedsamplingrendering}.

In this paper, we propose a new kd-tree stratification scheme that results in comparable error to state of the art method, with the added advantages of being simple to construct and easily parallelisable. We derive an upper bound for the star-discrepancy of sample sets produced using this stratification scheme.

\section{Jittered kd-tree Stratification}\label{sec:jkdt}

The central idea of our construction is to use a kd-tree to partition the $d$-dimensional hypercube into $n$ equal-volume strata. One sample is then drawn from each of these cells.

%___________________________________________________
\newcommand{\dom}{\ensuremath{\mathcal{V}}}
\newcommand{\domi}[1]{\ensuremath{\mathcal{V}_{#1}}}
\newcommand{\hplane}{\ensuremath{\mathcal{H}}}
\newcommand{\lb}{\ensuremath{\mathbf{l}}}
\newcommand{\Lb}{\ensuremath{\mathbf{L}}}
\newcommand{\lbm}[1]{\ensuremath{\mathbf{l}^i_{#1}}}
\newcommand{\Lbm}[1]{\ensuremath{\mathbf{L}[#1]}}
\newcommand{\ub}{\ensuremath{\mathbf{u}}}
\newcommand{\Ub}{\ensuremath{\mathbf{U}}}
\newcommand{\ubm}[1]{\ensuremath{\mathbf{u}^i_{#1}}}
\newcommand{\Ubm}[1]{\ensuremath{\mathbf{U}[#1]}}

%%%% Recursive algorithm
\begin{algorithm}[h]
\SetAlgoLined
\SetKwData{Left}{left}\SetKwData{This}{this}\SetKwData{Up}{up}
\SetKwFunction{Union}{Union}\SetKwFunction{FindCompress}{FindCompress}
\SetKwInOut{Input}{input}\SetKwInOut{Output}{output}

\Input{number of strata $n$ \\ dimension $d$}
\Output{Lower bounds array $\Lb$ with elements $\lbm{m}$ \\ Upper bounds array $\Ub$ with elements $\ubm{m}$}

$N_{\text{rem}} \leftarrow n$ \tcp*[r]{remaining partitions}
$\lb{} \leftarrow (0,\dots,0)$ \tcp*[r]{lower partition bound}
$\ub{} \leftarrow (1,\dots,1)$ \tcp*[r]{upper partition bound}

\eIf{$N_{\text{rem}} >1 $ }
    {
        $m \leftarrow 0$  \tcp*[r]{dimension to partition}
        $c \leftarrow \lb_m +\frac{\ceil{\frac{N_{\text{rem}}}{2}}}{N_{\text{rem}}} (\ub_m-\lb_m)$ \;
        
        $\lb_{\text{right}} \leftarrow \lb $ \tcp*[r]{right subtree lower bounds}
        $\lb_{\text{right}_m} \leftarrow c$ \;
        
        $\text{RightSubTree}\left((m+1)\%d,\ \lb_{\text{right}},\ \ub,\ N_{\text{rem}}- \ceil{\frac{N_{\text{rem}}}{2}}\right)$ \;

        $\ub_{\text{left}} \leftarrow \ub $ \tcp*[r]{left subtree upper bounds}
        $\ub_{\text{left}_m} \leftarrow c$ \;
        
        $\text{LeftSubTree}\left((m+1)\%d,\ \lb,\ \ub_{\text{left}},\ \ceil{\frac{N_{\text{rem}}}{2}}\right)$ \;
    }
    {
        $\Lb.\text{push}(\lb)$ \;
        $\Ub.\text{push}(\ub)$ \;
    }

\caption{CalculateBoundsRecursive\label{alg:rec_bounds}}
\end{algorithm}

%% subtree algo
\begin{algorithm}[h]
\SetAlgoLined
\SetKwData{Left}{left}\SetKwData{This}{this}\SetKwData{Up}{up}
\SetKwFunction{Union}{Union}\SetKwFunction{FindCompress}{FindCompress}
\SetKwInOut{Input}{input}\SetKwInOut{Output}{output}

\Input{dimension $m$ \\ lower bound $\lb$ \\ upper bound $\ub$ \\ remaining partitions $N_{\text{rem}}$}

\eIf{$N_{\text{rem}} >1 $ }
{
$c \leftarrow \lb_m +\frac{\floor{\frac{N_{\text{rem}}}{2}}}{N_{\text{rem}}} (\ub_m-\lb_m)$ \;
        
        $\lb_{\text{right}} \leftarrow \lb $ \tcp*[r]{right subtree lower bounds}
        $\lb_{\text{right}_m} \leftarrow c$ \;
        
        $\text{RightSubTree}\left((m+1)\%d,\ \lb_{\text{right}},\ \ub, \ceil{\frac{N_{\text{rem}}}{2}}\right)$ \tcp*[r]{right subtree of right subtree}

        $\ub_{\text{left}} \leftarrow \ub $ \tcp*[r]{left subtree upper bounds}
        $\ub_{\text{left}_m} \leftarrow c$ \;
        
        $\text{RightSubTree}\left((m+1)\%d,\ \lb,\ \ub_{\text{left}},\ \floor{\frac{N_{\text{rem}}}{2}}\right)$ \tcp*[r]{left subtree of right subtree}
    }
    {
        $\Lb.\text{push}(\lb)$ \;
        $\Ub.\text{push}(\ub)$ \;
    }

\caption{RightSubTree\label{alg:rightsubtree}}
\end{algorithm}

\begin{algorithm}[h]
\SetAlgoLined
\SetKwData{Left}{left}\SetKwData{This}{this}\SetKwData{Up}{up}
\SetKwFunction{Union}{Union}\SetKwFunction{FindCompress}{FindCompress}
\SetKwInOut{Input}{input}\SetKwInOut{Output}{output}

\Input{dimension $m$ \\ lower bound $\lb$ \\ upper bound $\ub$ \\ remaining partitions $N_{\text{rem}}$}

\eIf{$N_{\text{rem}} >1 $ }
{
$c \leftarrow \lb_m +\frac{\ceil{\frac{N_{\text{rem}}}{2}}}{N_{\text{rem}}} (\ub_m-\lb_m)$ \;
        
        $\lb_{\text{right}} \leftarrow \lb $ \tcp*[r]{right subtree lower bounds}
        $\lb_{\text{right}_m} \leftarrow c$ \;
        
        $\text{LeftSubTree}\left((m+1)\%d,\ \lb_{\text{right}},\ \ub, \floor{\frac{N_{\text{rem}}}{2}}\right)$ \tcp*[r]{right subtree of left subtree}

        $\ub_{\text{left}} \leftarrow \ub $ \tcp*[r]{left subtree upper bounds}
        $\ub_{\text{left}_m} \leftarrow c$ \;
        
        $\text{LeftSubTree}\left((m+1)\%d,\ \lb,\ \ub_{\text{left}},\ \ceil{\frac{N_{\text{rem}}}{2}}\right)$ \tcp*[r]{left subtree of left subtree}
    }
    {
        $\Lb.\text{push}(\lb)$ \;
        $\Ub.\text{push}(\ub)$ \;
    }

\caption{LeftSubTree\label{alg:leftsubtree}}
\end{algorithm}

\subsection{Sample generation}

We illustrate our method using a simple (linear time) recursive procedure to generate $n$ deterministic strata in $d$ dimensions. Then, we explain how this construction can be used to obtain stochastic samples and derive a formula for determining the axis-aligned boundaries of the $i^{\mathrm{th}}$ cell (leaf of the kd-tree) out of $n$ samples. Samples are obtained by drawing a random sample within each of these cells.

\paragraph{Constructing the kd-tree} 
To construct $n$ strata with equal volumes in a $d$-dimensional hypercube \dom, we first partition \dom\ into two strata \domi 0\ and \domi 1\ by a splitting plane whose normal is parallel to an arbitrary axis  $m, 0\leq m \leq d-1$. If $n$ is even, the plane is located mid-way along the $l^\mathrm{th}$ axis in \dom. If $n$ is odd, the plane is positioned so that the volumes of  \domi 0\ and \domi 1\  are proportional to $\ceil{n/2}$ and $n-\ceil{n/2}$ respectively.~e.g.~if $n=5$, $d=3$ and $m=0$, the first split is performed by placing a plane parallel to the $YZ$ plane at $X=3/5$. The splitting procedure is recursively applied to \domi 0\ and \domi 1\ using $(m+1) \mod d$ as the splitting axis and numbers $n_0=\ceil{n/2}$ and $n_1=n-\ceil{n/2}$ respectively. A binary digit is prefixed to the subscript at each split -- a zero for the "lower" stratum (left branch) and a one to indicate the "upper" stratum (right branch). At the first recursion level (second split), the resulting partitions are  \domi {00},  \domi {01}\,  \domi {10}\ and  \domi {11} respectively. The recursion bottoms out when the number of stratifications required within a sub-domain is one. Algorithm~\ref{alg:rec_bounds}, along with the accompanying functions of Algorithms~\ref{alg:rightsubtree} and~\ref{alg:leftsubtree}, implement the afforementioned recursive procedure in $O(n)$ time, for the complete partitioning. Figure~\ref{fig:kdtreevis} visualises the cells of the tree and samples drawn within them for $n=16$, $n=59$ and $n=152$. 

\paragraph{Formula for jittered sampling}
The above procedure induces a tree whose $n$ leaves are axis-aligned hypercubes with equal volume. Although we could use the recursive procedure to generate a random sample in each stratum (every time the recursion bottoms out), this would induce unwanted computational overhead when a single sample is required. Instead, we derive a direct procedure (see Alg.~\ref{alg:bounds}) to calculate the lower and upper bounds for each cell with a run time complexity of $O(\ceil{\log_2 n})$ per sample. If all $n$ samples are to be generated at once, the recursive procedure is more efficient since it is $O(n)$ instead of $O(n\ceil{\log_2 n})$. However, the former is not easily parallelisable and requires pregenerated samples to be stored. The latter is fully parallelisable and pregenerated samples are notional and may be independently drawn when required.

%%%% New algorithm
\begin{algorithm}[h]
\SetAlgoLined
\SetKwData{Left}{left}\SetKwData{This}{this}\SetKwData{Up}{up}
\SetKwFunction{Union}{Union}\SetKwFunction{FindCompress}{FindCompress}
\SetKwInOut{Input}{input}\SetKwInOut{Output}{output}

\Input{number of strata $n$ \\ dimension $d$ \\ sample number $i \in \{0,1,\cdots,n-1\}$}
\Output{$d$-dimensional bounds \lbm {}\ and \ubm {}\ }

$m \leftarrow 0$  \tcp*[r]{dimension to partition}
$N_{\text{rem}} \leftarrow n$ \tcp*[r]{remaining points in partition}
$s \leftarrow \ceil{\log_2 n}$ \tcp*[r]{num. bits}
$B=[b^i_{s-1} b^i_{s-2}\ \dots\ b^i_{0}]$ \tcp*[r]{big-endian binary representation of point index $i$}
$\lbm{} \leftarrow (0,\dots,0)$ \tcp*[r]{lower partition bound}
$\ubm{} \leftarrow (1,\dots,1)$ \tcp*[r]{upper partition bound}
 
 \While{$B$ not empty and $N_{\text{rem}}>1$}
 {
    $b \leftarrow \text{pop last element of $B$}$ \;
    $r\leftarrow \ceil{\frac{N_{\text{rem}}}{2}}$ \;
    \eIf( \tcp*[h]{update upper bound}){$b$ is zero}
    {
        $N_{\text{curr}} \leftarrow \ceil{\frac{N_{\text{rem}}}{2}}$ \;
        $\ubm{m} \leftarrow (\ubm{m}-\lbm{m})\frac{r}{{N_{\text{rem}}}}+\lbm{m}$ \;
    }
    (\tcp*[h]{update lower bound})
    {
        $N_{\text{curr}} \leftarrow \floor{\frac{N_{\text{rem}}}{2}}$ \;

        $\lbm{m} \leftarrow (\ubm{m}-\lbm{m})\frac{r}{{N_{\text{rem}}}}+\lbm{m}$ \;

    }
    $N_{\text{rem}}\leftarrow N_{\text{curr}}$ \;
    $m \leftarrow (m+1) \% d$ \;
}

\caption{CalculateBounds\label{alg:bounds}}
\end{algorithm}

\begin{figure*}[h]
    \centering
    \includegraphics[width=\linewidth]{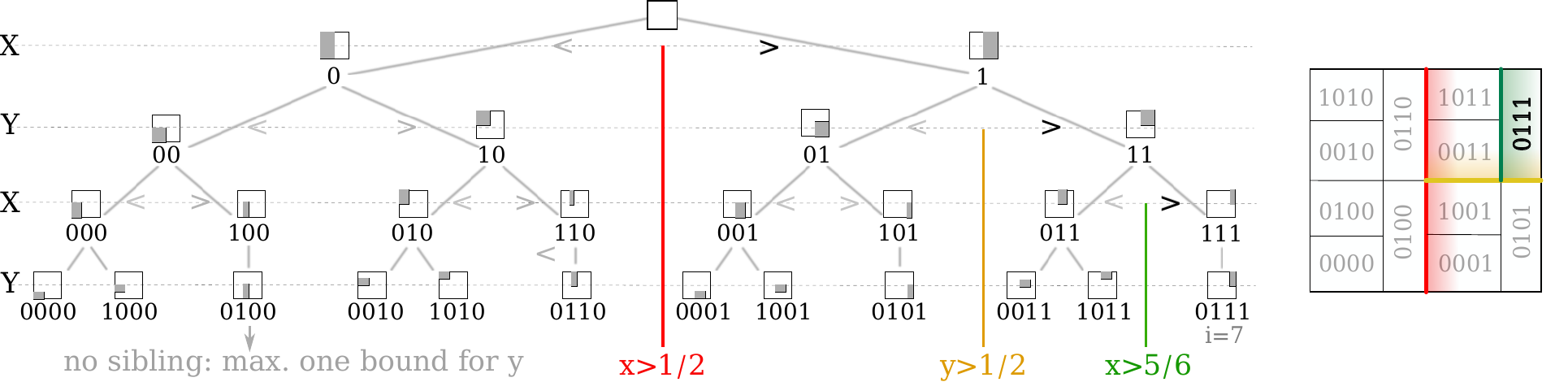}
    \caption{We stratify the d-dimensional hypercube into axis-aligned hyperrectangles of equal volume using a kd-tree. The figure illustrates an example with $12$ samples in $2$D. Each leaf node in the tree represents a cell with area $1/12$. In practice, we do not need to build the tree explicitly. When a sample is needed from a particular cell, we calculate the bounds the cell using Algorithm~\ref{alg:bounds} and then draw a random sample within it. The coloured portions of the figure illustrate the example described in the text to obtain the bounds of the $8^{\mathrm{th}}$ sample.~i.e.~$i=7$.}
    \label{fig:kdTree}
\end{figure*}

\paragraph{Example} To find the bounds for the $8^{th}$ sample ($i=7$) out of $n=12$ samples in $d=2$D, we first express $i$ as the bit-string $0111$. Starting from the right end of this string, we process one digit at a time while progressively halving the number of points $n$. At each step, we update the bounds appropriately. Alternate digits correspond to bounds for X and Y respectively and we adopt the convention that the first digit (rightmost) corresponds to a split in X. Since the first digit is a $1$, it corresponds to a lower bound on $X$.~i.e.~$x>\lbm 0$. Since $N_0=12$ is divisible by $2$, $N_1/N_0=1/2$ leading to $x>1/2$. The second least-significant digit is $1$, which again leads to a lower bound but this time on $Y$.  Since $N_1=6$ is even, the bound is $y>1/2$. The third digit is $1$ as well, and imposes a lower bound on X. However, since $N_2=3$ is odd, $N_3=2$. The new bound for $X$ is $x>1/2 + (1-1/2)*2/3$. i.e. $x>5/6$. Finally, the most significant bit is $0$, so it corresponds to an upper bound.  However, because the current node $0111$ does not have a sibling ($1111$ corresponds to $15$ which is greater than $12$ which is the number of samples), the upper bound is not updated. Combining all this information, we obtain $x\in[5/6,1]$ and $y\in[1/2,1]$. See Fig.~\ref{fig:kdTree} for an illustration of this procedure. 

\begin{figure}[htbp!]
    \centering
    \begin{tabular}{@{}c@{}c@{}c@{}}
        \includegraphics[width=0.3\linewidth]{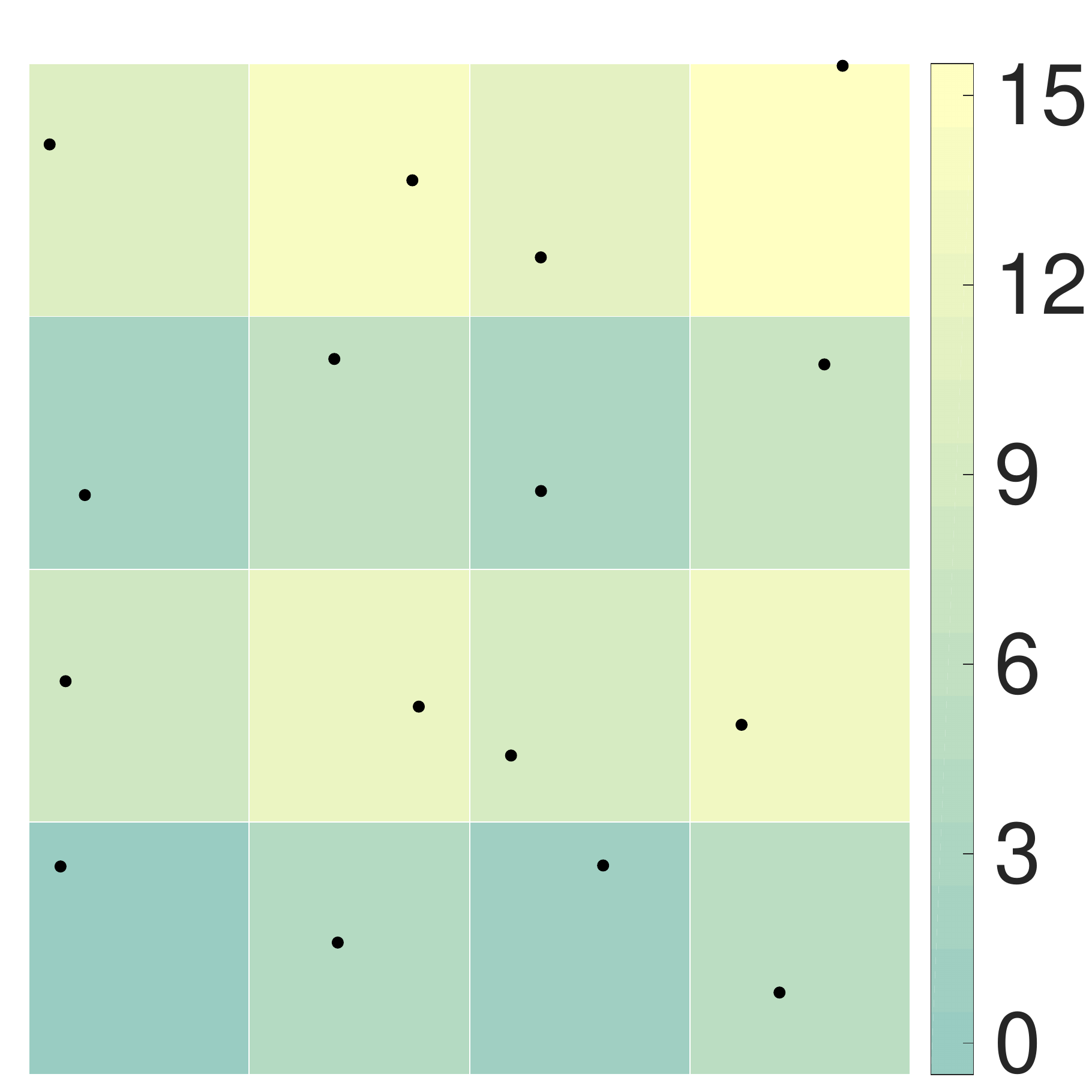}&
        \includegraphics[width=0.3\linewidth]{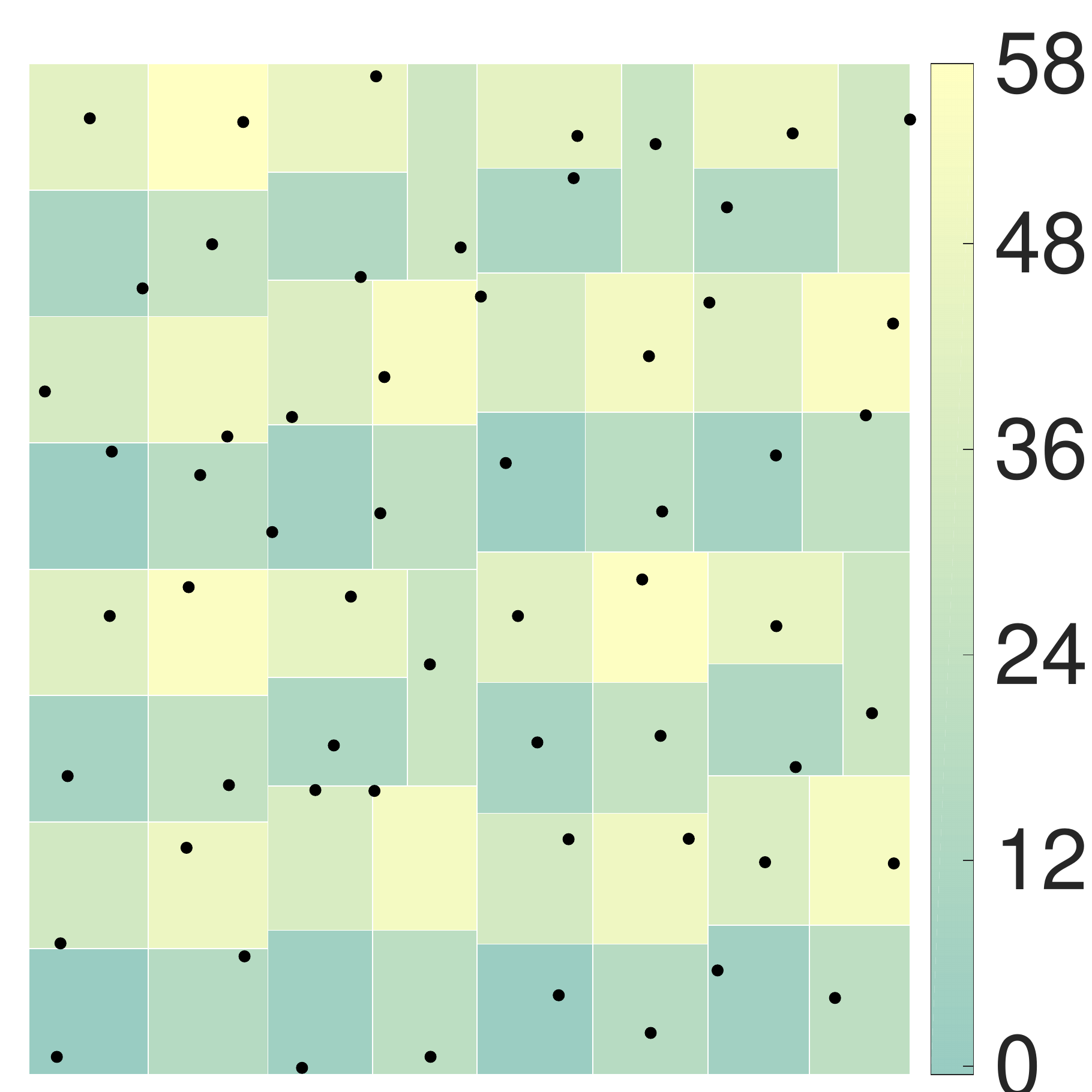}&
        \includegraphics[width=0.32\linewidth]{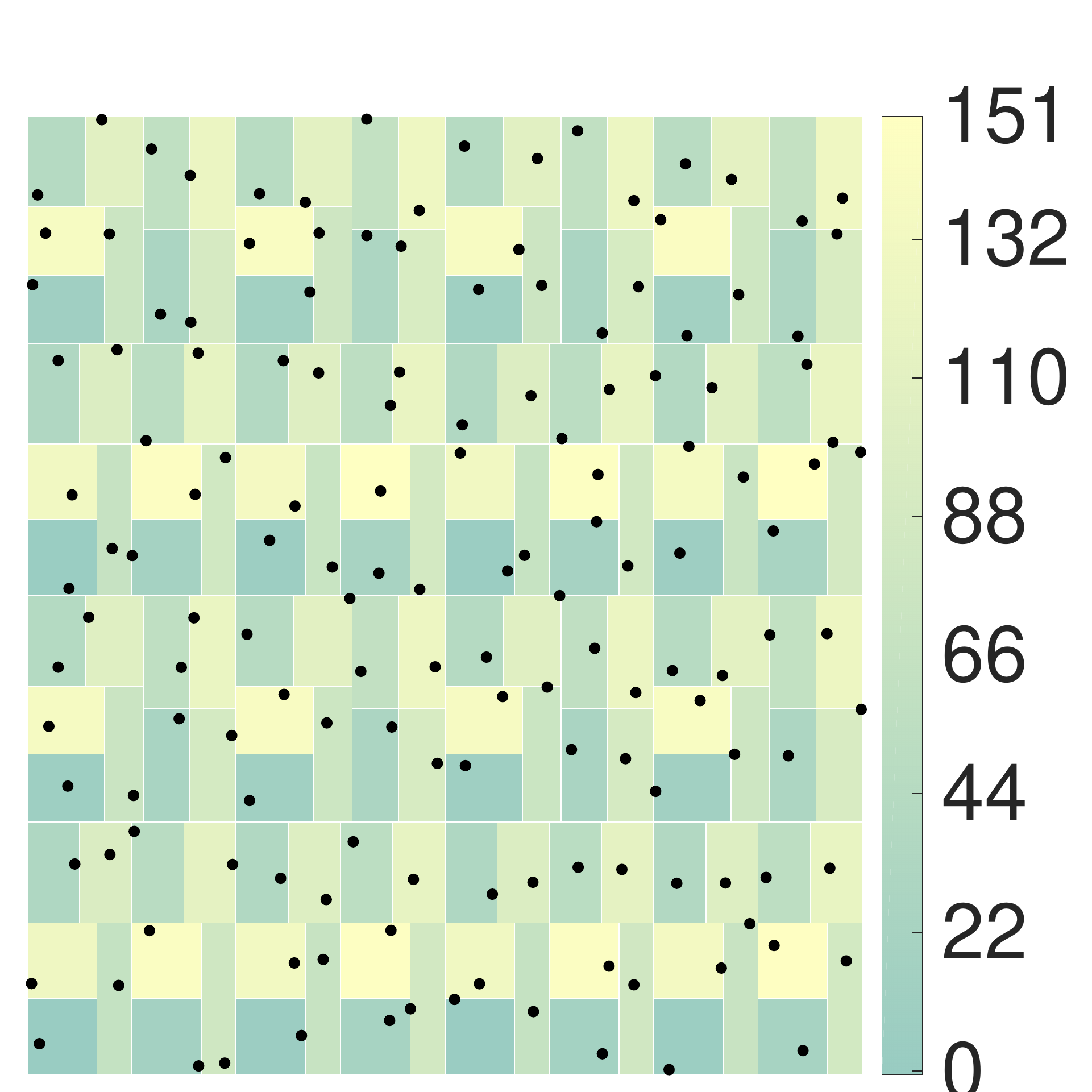} \\
        16 samples & 59 samples & 152 samples
    \end{tabular}
    \caption{\label{fig:kdtreevis} A visualisation of our kd-tree stratification and samples. The cells are coloured by the sample number.}
\end{figure}

\subsection{Discrepancy of kd-tree samples}\label{subsec:discr_proof}

Samples on a regular grid have a discrepancy of $O\left({1}/{\sqrt[d]{n}}\right)$. Jittered sampling~\cite{pausinger2016discrepancy} and Latin hypercube sampling~\cite{doerr2018probabilistic} exploit the correlation structure imposed by the grid, and, by randomly drawing samples within each stratum, improve the expected discrepancy to $O\left(\frac{\log n^{\frac{1}{2}}}{n^{\frac{1}{2}+\frac{1}{2d}}}\right)$ (under sufficiently dense sampling assumption) and $O\left(\sqrt{\frac{d}{n}}\right)$, respectively. 

Kd-tree stratification, as a generalization of jittered sampling to arbitrary sample counts, results to the exact same domain partition as jittered sampling for $n=2^{kd},\ k\in \mathbb{Z}$. Thus, for these special cases, our method satisfies the same expected discrepancy bounds as jittered sampling~\cite{pausinger2016discrepancy}, namely $O\left(\frac{\log n^{\frac{1}{2}}}{n^{\frac{1}{2}+\frac{1}{2d}}}\right)$. Theorem~\ref{thm:discr} derives a worst-case upper bound for the star-discrepancy of the general case of our jittered kd-tree stratification method.

\begin{theorem}\label{thm:discr}
   Given a set $P$ of $n$ samples,  $D^*(P) \leq 2^{d-1} d n^{-\frac{1}{d}}$.

\end{theorem}
\begin{proof}
If $P$ is an $n$ point set and $X$ is a set, then let $D(P,X)$ denote the discrepancy of $P$ in $X$.  Observation 1.3 in \cite{matousek2009geometric} states that if $A$ and $B$ are disjoint sets, then $|D(P,A\cup B)| = |D(P,A)+ D(P,B)| \leq |D(P,A)|+|D(P,B)|$. We use this observation inductively to compute discrepancy. Any axis-parallel rectangle with a vertex anchored at the origin, as used for the star-discrepancy computation, contains partial and complete cells from our kd-tree. By construction, complete cells in our kd-tree have zero discrepancy.  The discrepancy of incomplete cells is less than or equal to $1/{n}$.  This value, multiplied by a bound on the number of partial cells in such a rectangle can provide an upper bound for discrepancy.  
%As star-discrepancy --  -- over axis-parallel rectangles with one vertex at the origin -- is less than or equal to discrepancy, we have an upper bound for star discrepancy.  
The number of cells a splitting plane (perpendicular to an axis) intersects increases as $N$ increases.  The maximum number of intersections occurs when the cells are identical to cells in a regular grid.~i.e.~ when $n = (2^k)^d, k\in\mathbb{Z}$.  Thus, each splitting plane intersects at most 
% $\sqrt[d]{(2^d)^{\ceil{\log_{2^d}(n)}}}$ cells.  But, 
\begin{eqnarray}
\left(2^{d\ceil{\log_{2^d}(n)}}\right)^{\frac{d-1}{d}}   \leq  
\left(2^{d\log_{2^d}(n)}  2^d\right)^{\frac{d-1}{d}}  = 
2^{d-1} n^{\frac{d-1}{d}}
\end{eqnarray}
cells. Thus, an axis-parallel rectangle intersects at most $ 2^{d-1} d n^{\frac{d-1}{d}}$ cells, and the star-discrepancy of samples obtained using kd-tree stratification has an upper bound of $2^{d-1} d n^{-\frac{1}{d}}$.

\end{proof}

Empirical computation of star-discrepancy in Figure~\ref{fig:emp_discr} suggests the looseness of the bound derived in Theorem~\ref{thm:discr}, and the fact that the expected star-discrepancy of our jittered kd-tree stratification method should satisfy the bounds of jittered sampling~\cite{pausinger2016discrepancy} for general sample counts $n\in \mathbb{Z}$ as well.

%-------------------------------------------------------------------------
%-------------------------------------------------------------------------
\section{Empirical evaluation}
 We performed quantitative and qualitative experiments to assess the proposed stratification scheme. First, we plotted (Fig.~\ref{fig:emp_discr}) the expected $L_2$-star discrepancy computed by Warnock's formula~\cite{warnock1972computational} versus number of samples (up to $4500,0$) for various samplers in 2D and 4D. Our method exhibits discrepancy similar to jittered sampling, which, as expected, is inferior to the discrepancy of QMC sample sequences. Then, we tested the errors resulting from jittered kd-tree sampling when integrating analytical functions (sec.~\ref{subsec:analytic_ints}) of variable complexity in low-dimensions. Finally, we tested a padded version of jittered kd-tree stratification using rendering scenarios consisting of high-dimensional light paths (sec.~\ref{subsec:renderings}). For all  empirical results, we compare against popular baseline samplers such as random and jittered sampling, and state-of-the art Quasi-Monte Carlo methods Halton and Sobol.

%___________________________________________________
% \subsection{Discrepancy empirical analysis}

\begin{figure*}[h]
\centering
\begin{tabular}{ccc}
2 Dimensions & 4 Dimensions & 7 Dimensions \\
\includegraphics[clip,trim=1.35cm 8.05cm 2.45cm 8.75cm, scale=0.24]{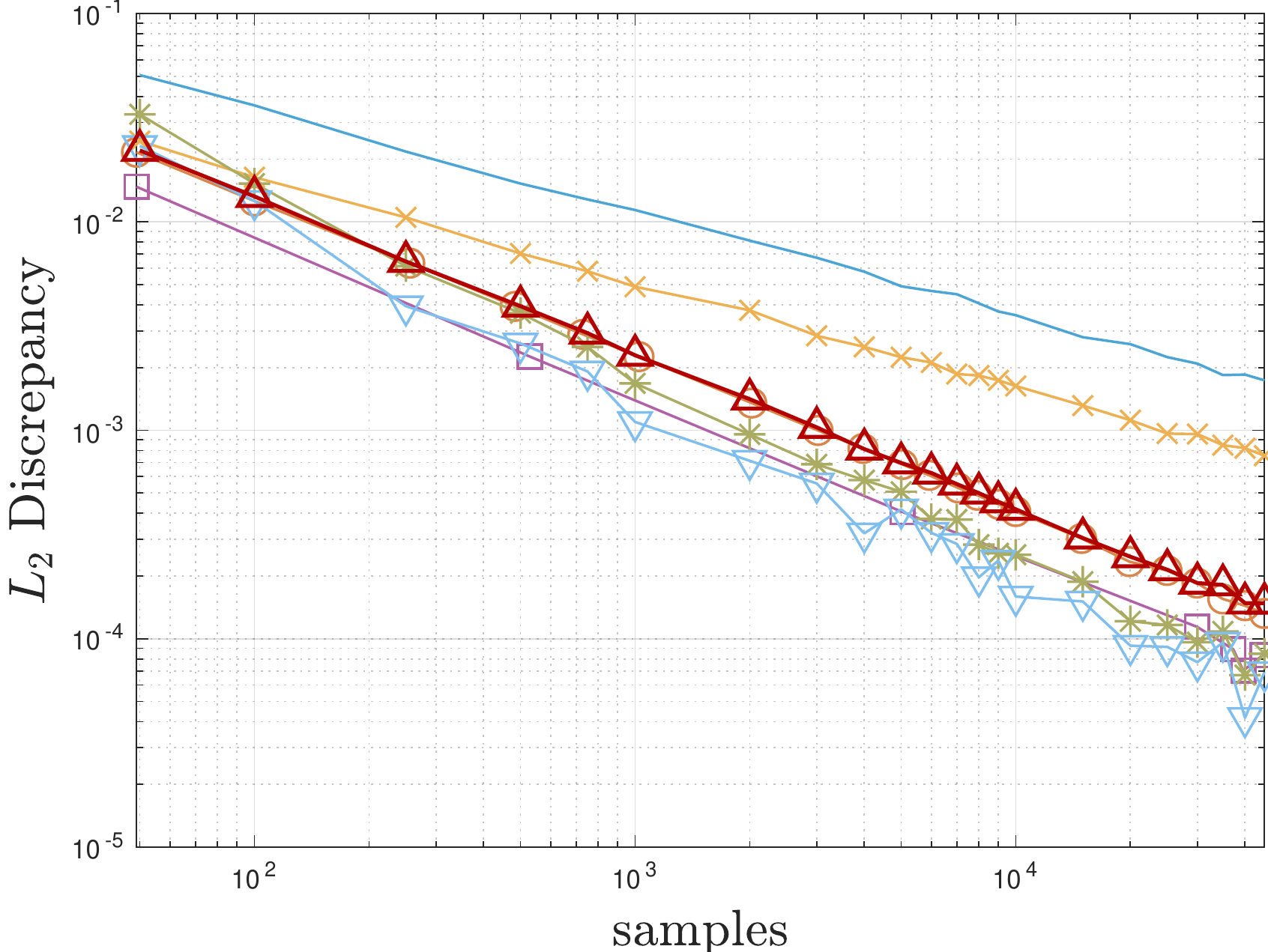} &
\includegraphics[clip,trim=1.35cm 8.05cm 2.45cm 8.75cm, scale=0.24]{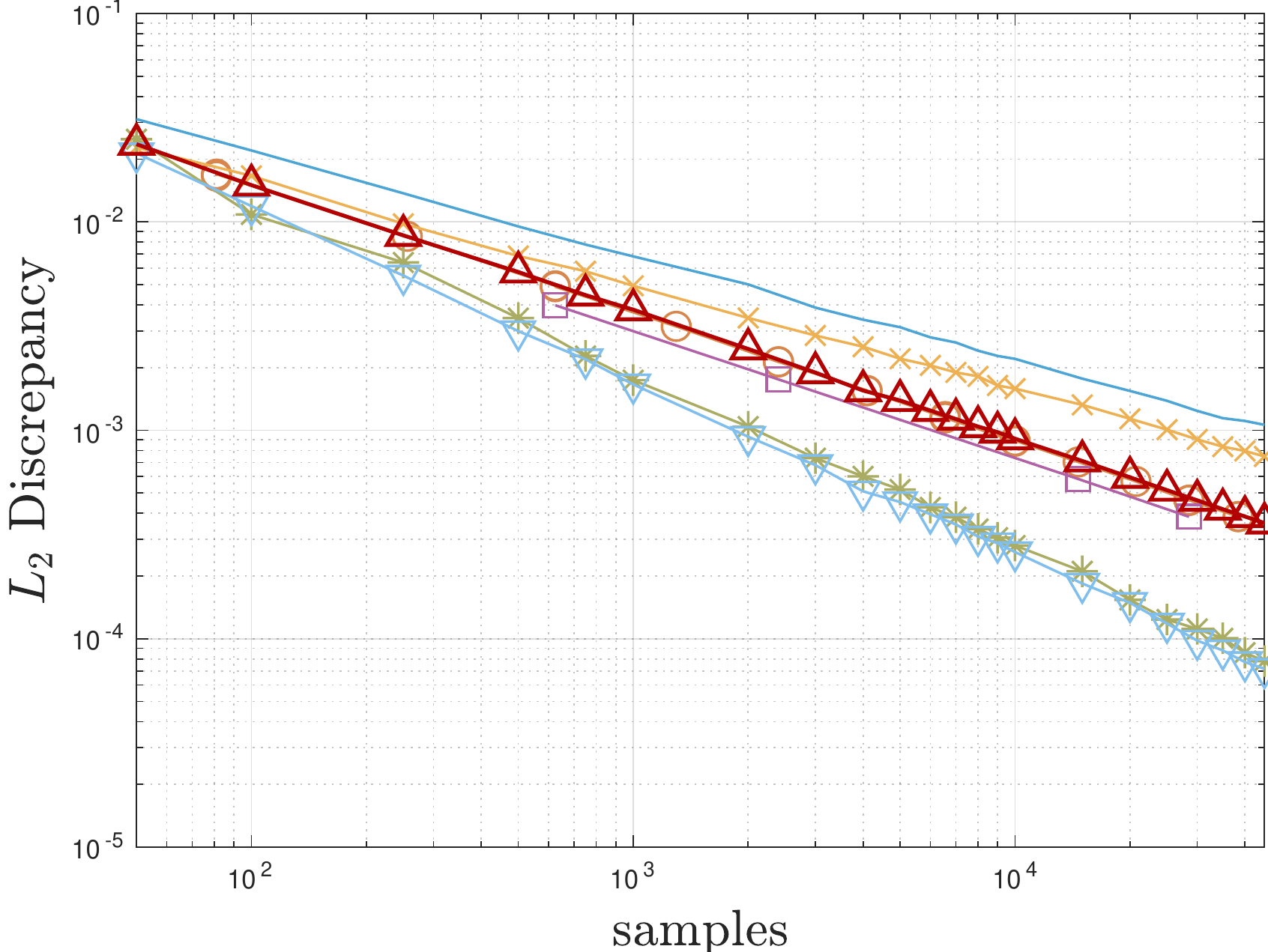} &
\includegraphics[clip,trim=1.35cm 8.05cm 2.45cm 8.75cm, scale=0.24]{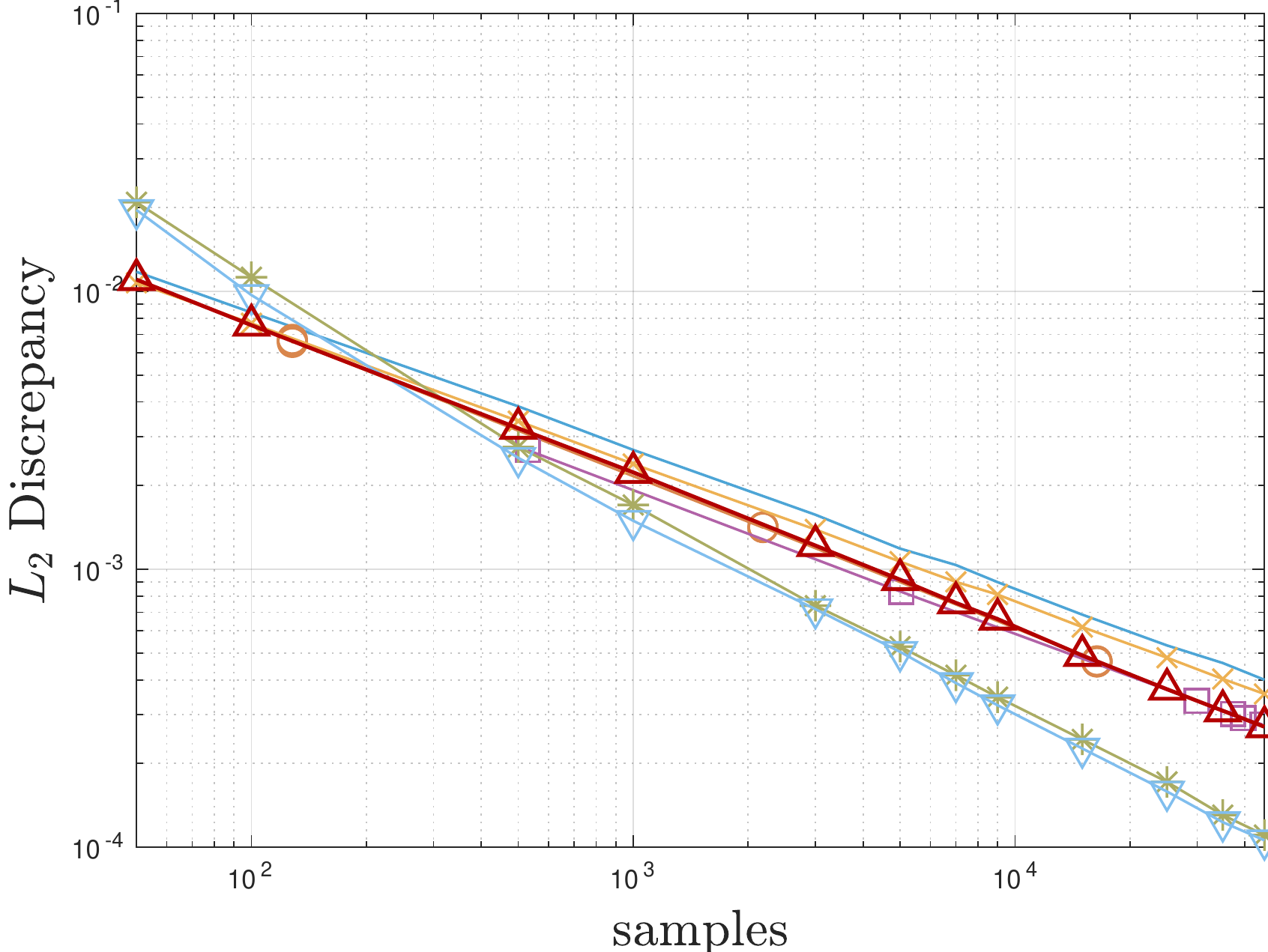}
\end{tabular}
\centering
\includegraphics[scale=0.3]{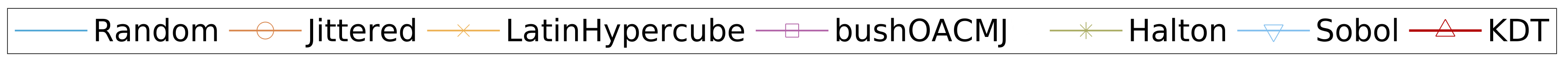}
\caption{The empirical expected $L_2$-star discrepancy (averaged over 100 sample realizations) for our method matches that of jittered sampling. It outperforms all examined samplers except the low-discrepancy QMC methods.}
\label{fig:emp_discr}
\end{figure*}

%___________________________________________________
\subsection{Evaluation on analytic integrands}\label{subsec:analytic_ints}

We compare the performance of our jittered kd-tree stratification method (KDT) against popular baseline methods (Jittered and Random) and  state-of-the art samplers (Halton and Sobol). We also compare our sampler against a new stratification method based on  Bush's Orthogonal Arrays~\cite{jarosz19orthogonal} (bushOACMJ), using fully factorial strengths where possible.

We plotted  (Figs.~\ref{fig:analytical_cont},~\ref{fig:analytical_discont}) mean-squared error vs number of samples, averaging over 100 sample realizations per method and per sample count, for various samplers used in estimating integrals of  known (analytical) functions  in 2 and 4 dimensions (rows respectively). We estimated integrals with up to $10^6$ samples of smooth and discontinuous integrands of increasing variability. As the smooth function, we used a normalised Gaussian mixture model  with $k$ randomly centered, randomly weighted modes ($GMM_k$). The standard deviation of each component is set to one third of the minimum distance between centres. We performed four experiments on this integrand: $k=3$ and $k=20$ each in 2 and 4 dimensions. As the discontinuous test integrand, we used a piecewise constant function by triangulating the domain using $k$ randomly selected points ($PWConst_k$). The faces created were weighted randomly and normalised so that the function integrates to unity. As with the Gaussian mixture, we performed 4 sets of experiments with $k=3$ and $k=20$ each in 2D and 4D.

In two dimensions the performance of our sampler consistently matches the performance of state-of-the-art QMC methods Halton and Sobol. For discontinuous integrands, $PWConst_3$  and $PWConst_{20}$, our sampler exhibits lower ``oscillations" than  QMC methods. Jittered sampling performs similar to our kd-tree stratification method in all cases except $PWConst_3$. Our sampler exhibits lower error than jittered for this case. 
In four dimensions QMC methods clearly outperform all other samplers. Our method again exhibits convergence similar to jittered sampling. Nevertheless, when the complexity of the integrand increases, in terms of non-linearities or modes, the distinction between different sampling strategies becomes increasingly difficult.

\begin{figure*}[t]
\centering
\begin{tabular}{ccc}
   &  $GMM_3$ 
   &  $GMM_{20}$\\
  \rotatebox{90}{\quad\quad\quad 2 Dimensions} & \hspace*{-0.2cm}\includegraphics[clip,trim=1.28cm 8.05cm 2.45cm 8.75cm, scale=0.35]{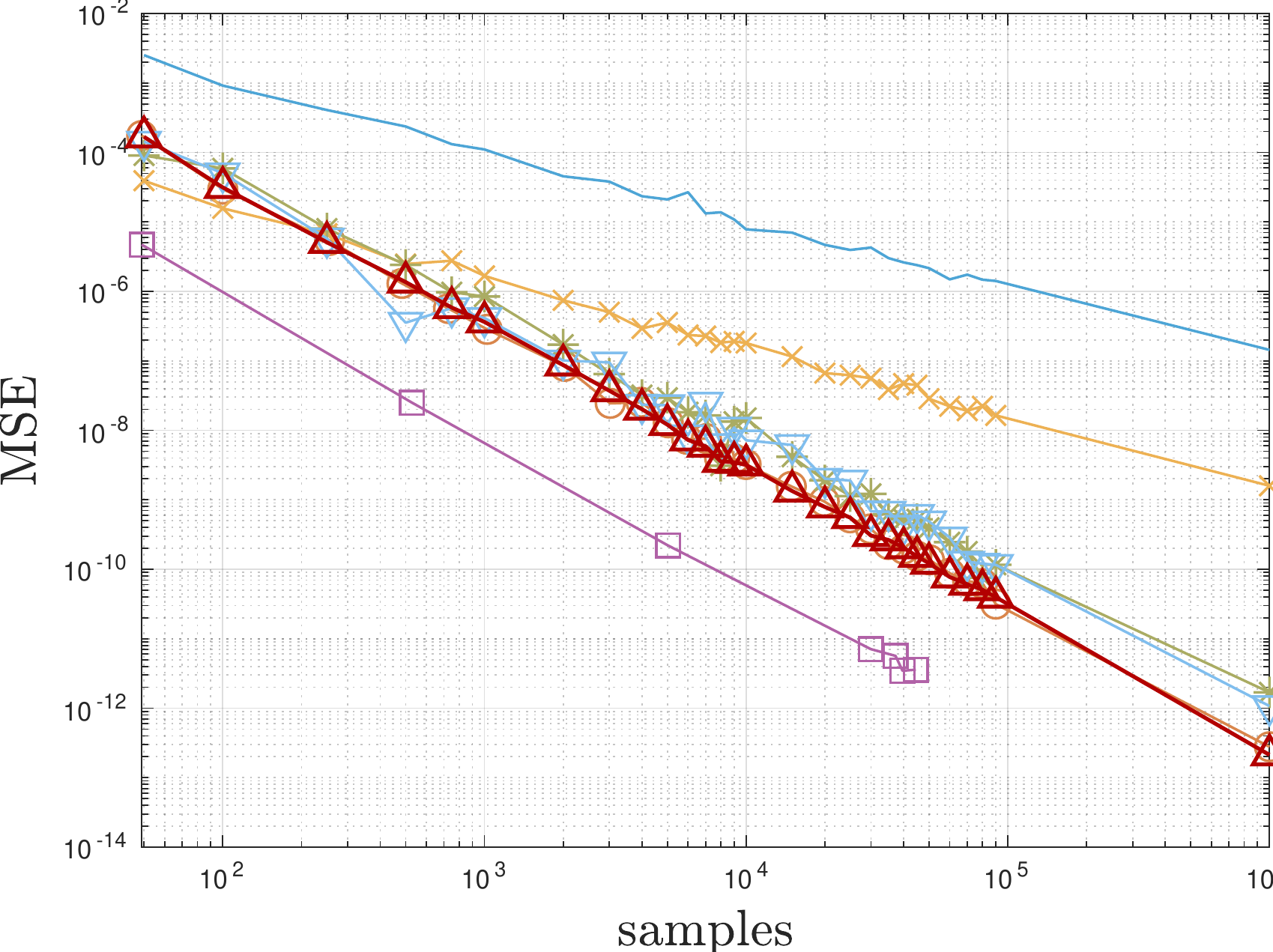}
   & \hspace*{-0.2cm}\includegraphics[clip,trim=1.28cm 8.05cm 2.45cm 8.75cm, scale=0.35]{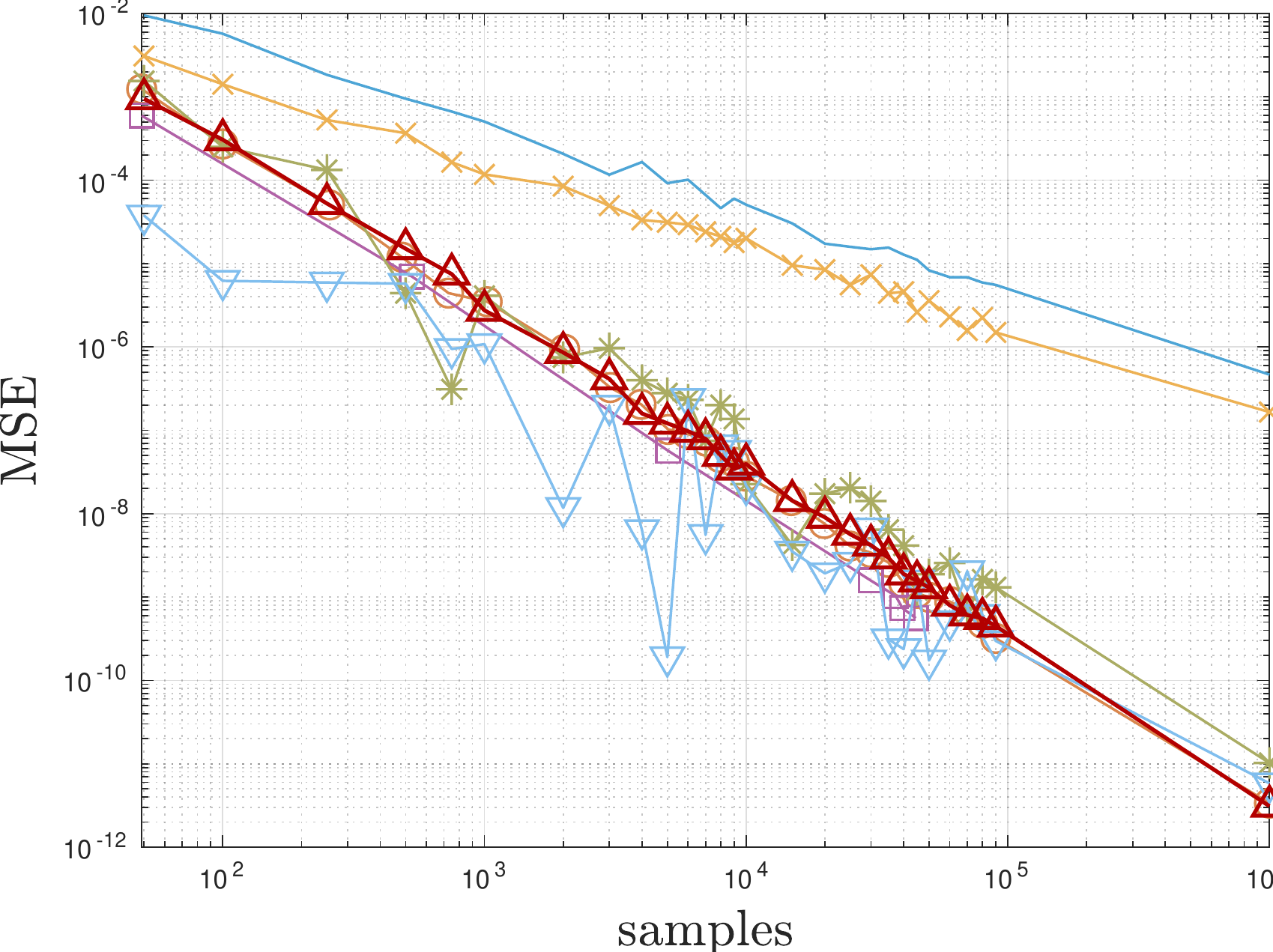}\\
     \rotatebox{90}{\quad\quad\quad 4 Dimensions} & \hspace*{-0.2cm}\includegraphics[clip,trim=1.28cm 8.05cm 2.45cm 8.75cm, scale=0.35]{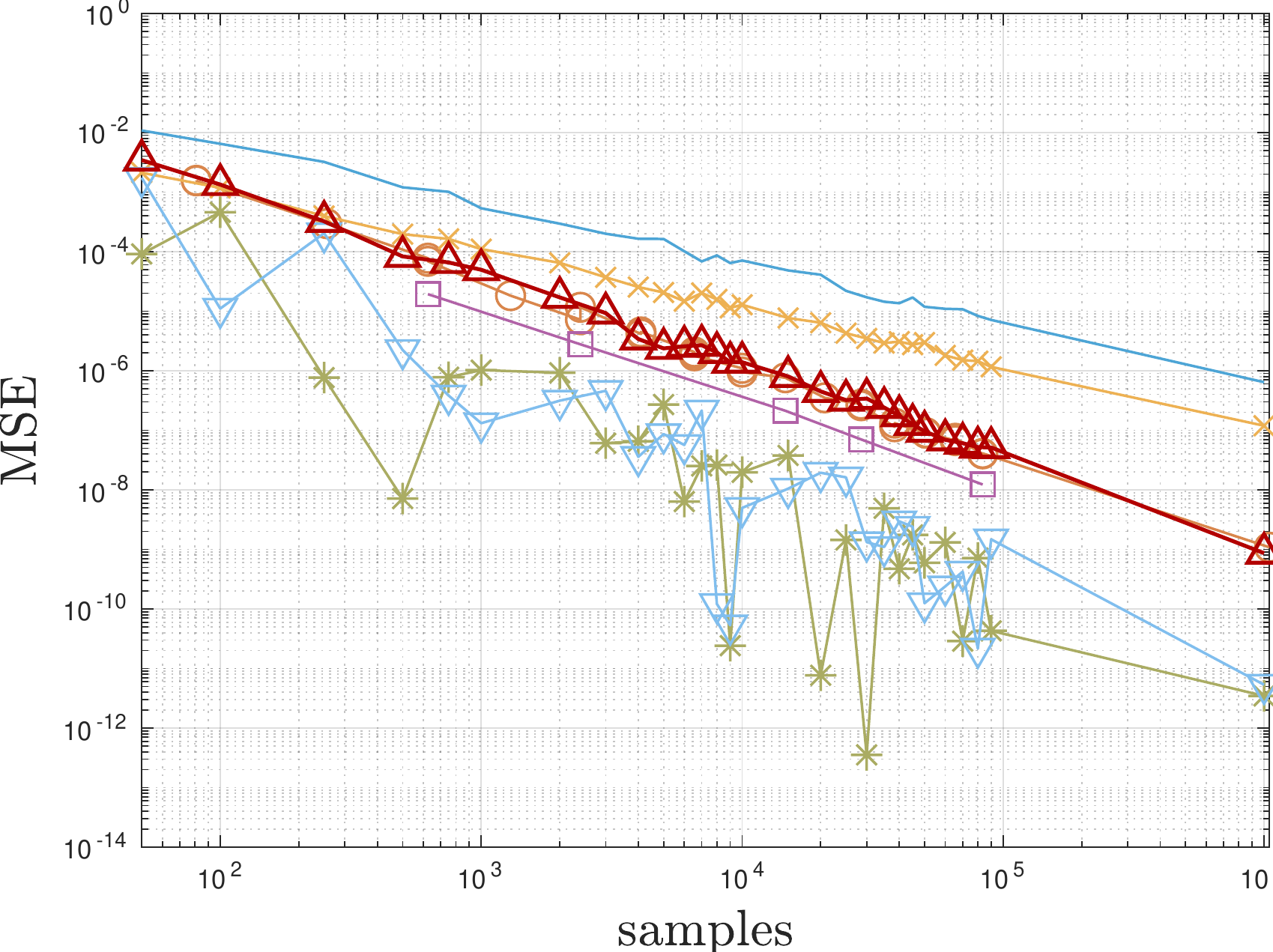}
   & \hspace*{-0.2cm}\includegraphics[clip,trim=1.28cm 8.05cm 2.45cm 8.75cm, scale=0.35]{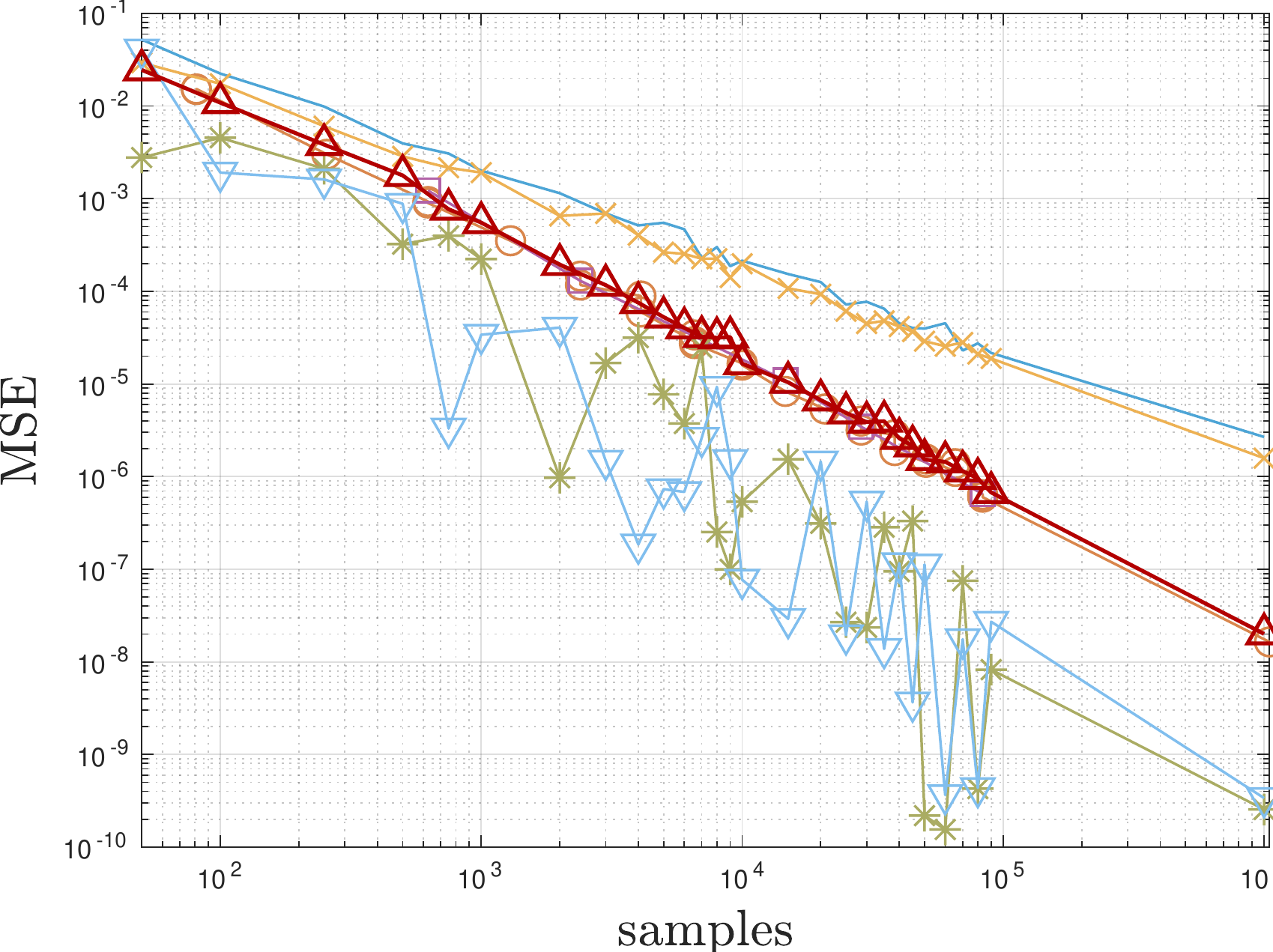}
\end{tabular}
\centering
\includegraphics[scale=0.3]{images/analytical/legend/legend.pdf}
\caption{\label{fig:analytical_cont}Error convergence of 7 samplers on 2D and 4D continuous analytic integrands of varied complexity. Our method (KDT) exhibits similar performance to jittered sampling, which is comparable to state-of-the-art QMC samplers in 2D, without limiting the allowed number of samples. In 4D the convergence of KDT again matches that of jittered sampling, both of which appear inferior to popular QMC methods.}
\end{figure*}

\begin{figure*}[t]
\centering
\begin{tabular}{ccc}
   & $PWConst_3$
   & $PWConst_{20}$\\
  \rotatebox{90}{\quad\quad\quad 2 Dimensions} 
   & \hspace*{-0.2cm}\includegraphics[clip,trim=1.28cm 8.05cm 2.45cm 8.75cm, scale=0.35]{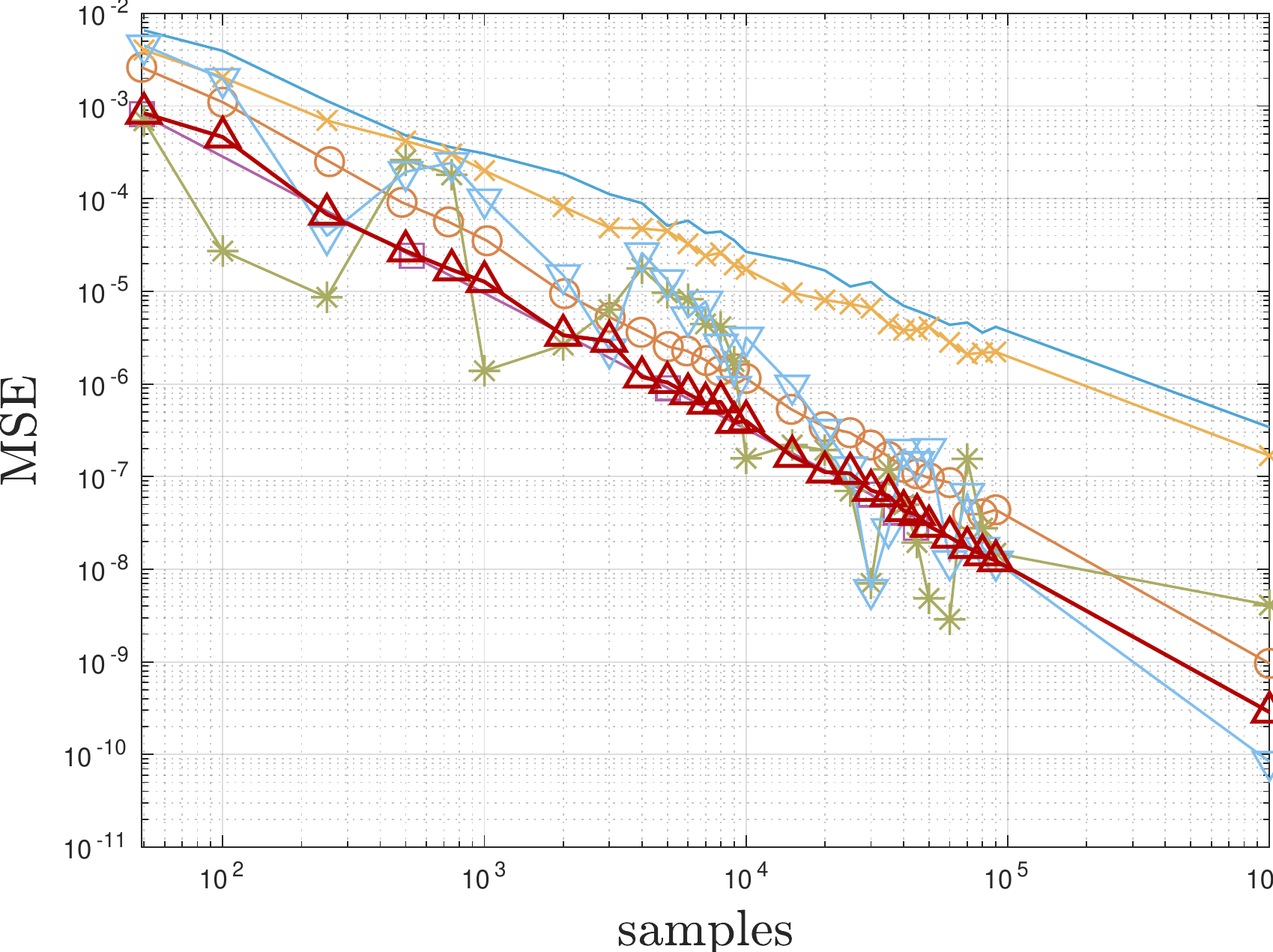}
   & \hspace*{-0.2cm}\includegraphics[clip,trim=1.28cm 8.05cm 2.45cm 8.75cm, scale=0.35]{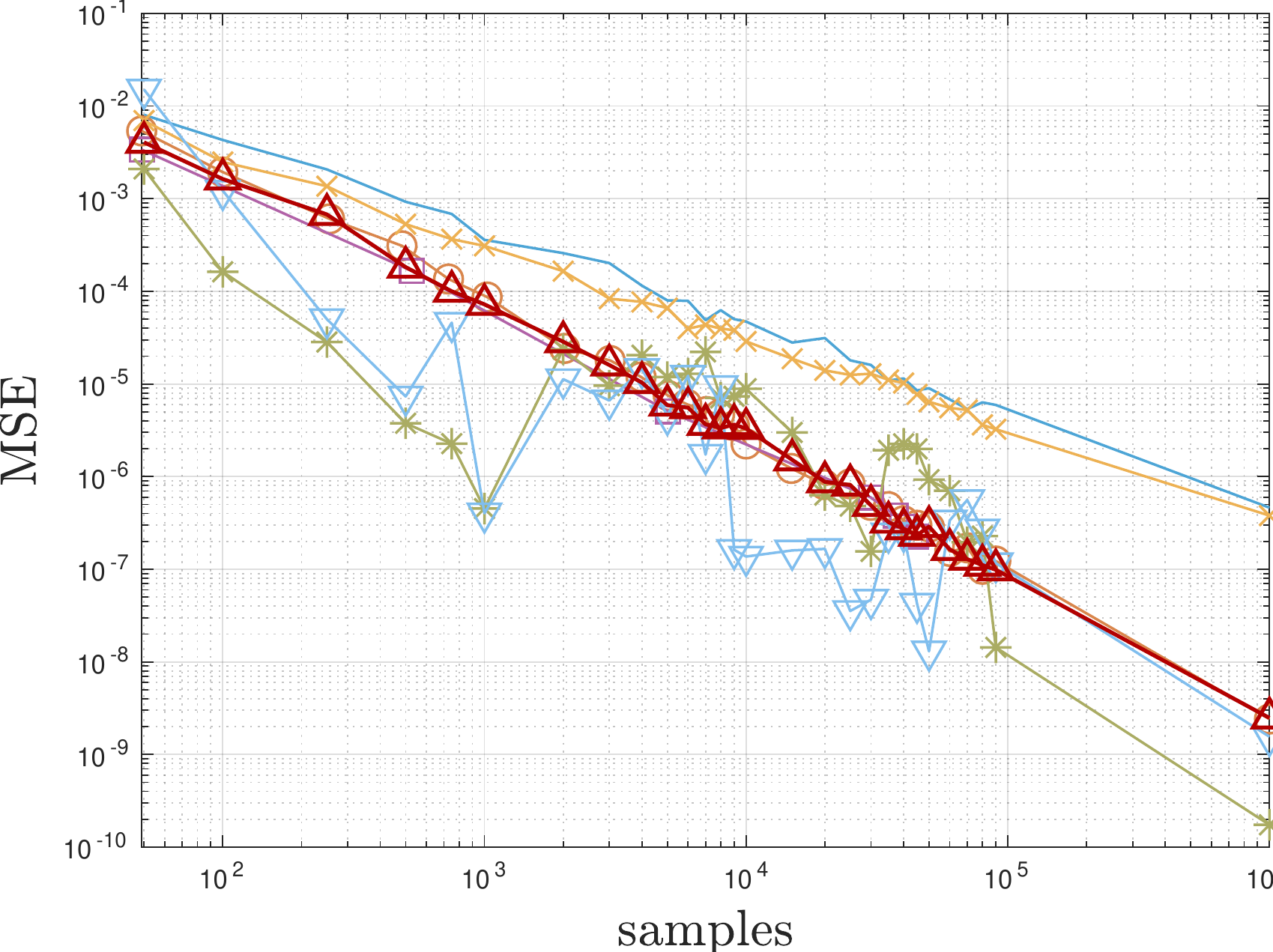}\\
     \rotatebox{90}{\quad\quad\quad 4 Dimensions} 
   & \hspace*{-0.2cm}\includegraphics[clip,trim=1.28cm 8.05cm 2.45cm 8.75cm, scale=0.35]{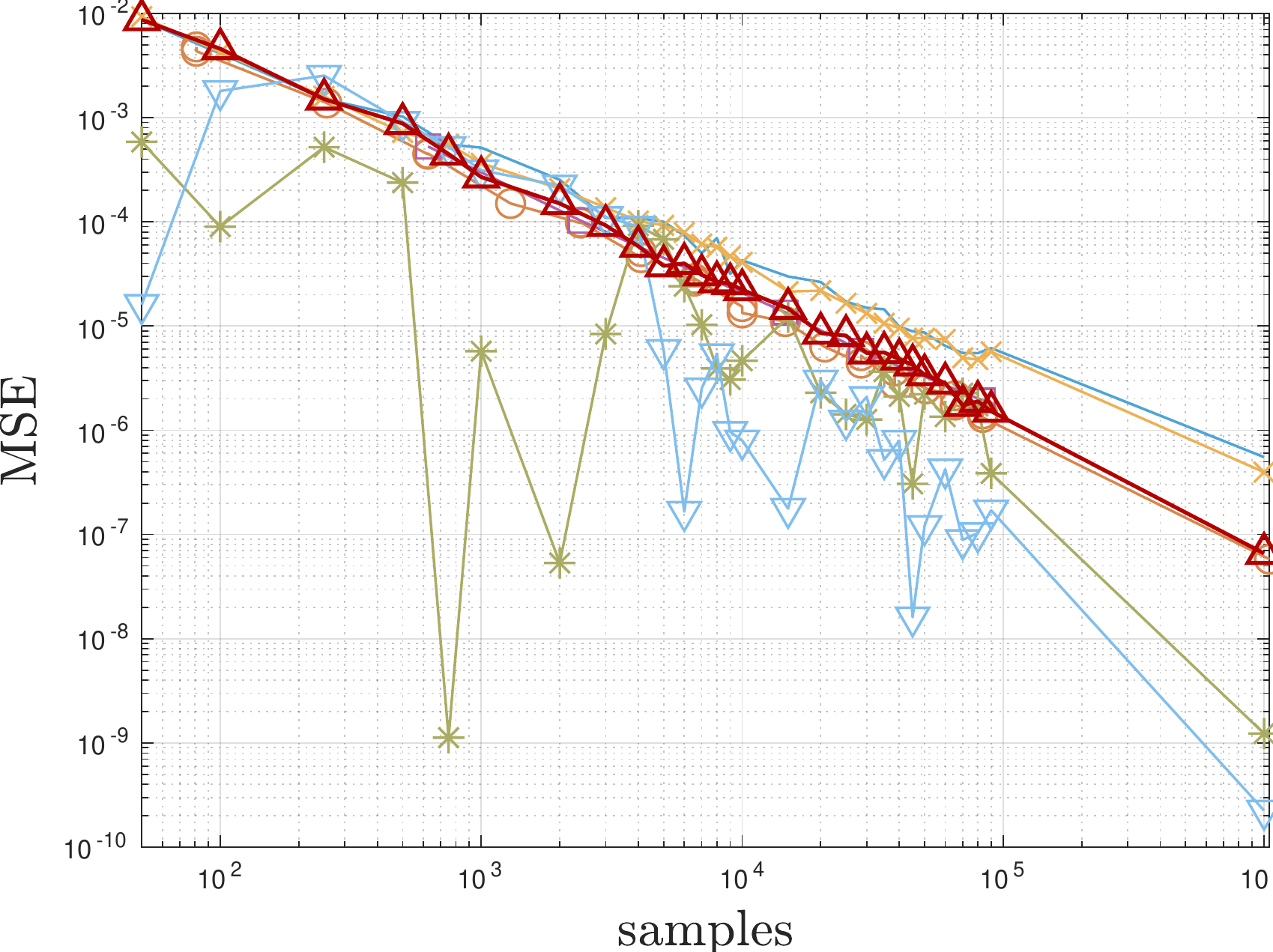}
   & \hspace*{-0.2cm}\includegraphics[clip,trim=1.28cm 8.05cm 2.45cm 8.75cm, scale=0.35]{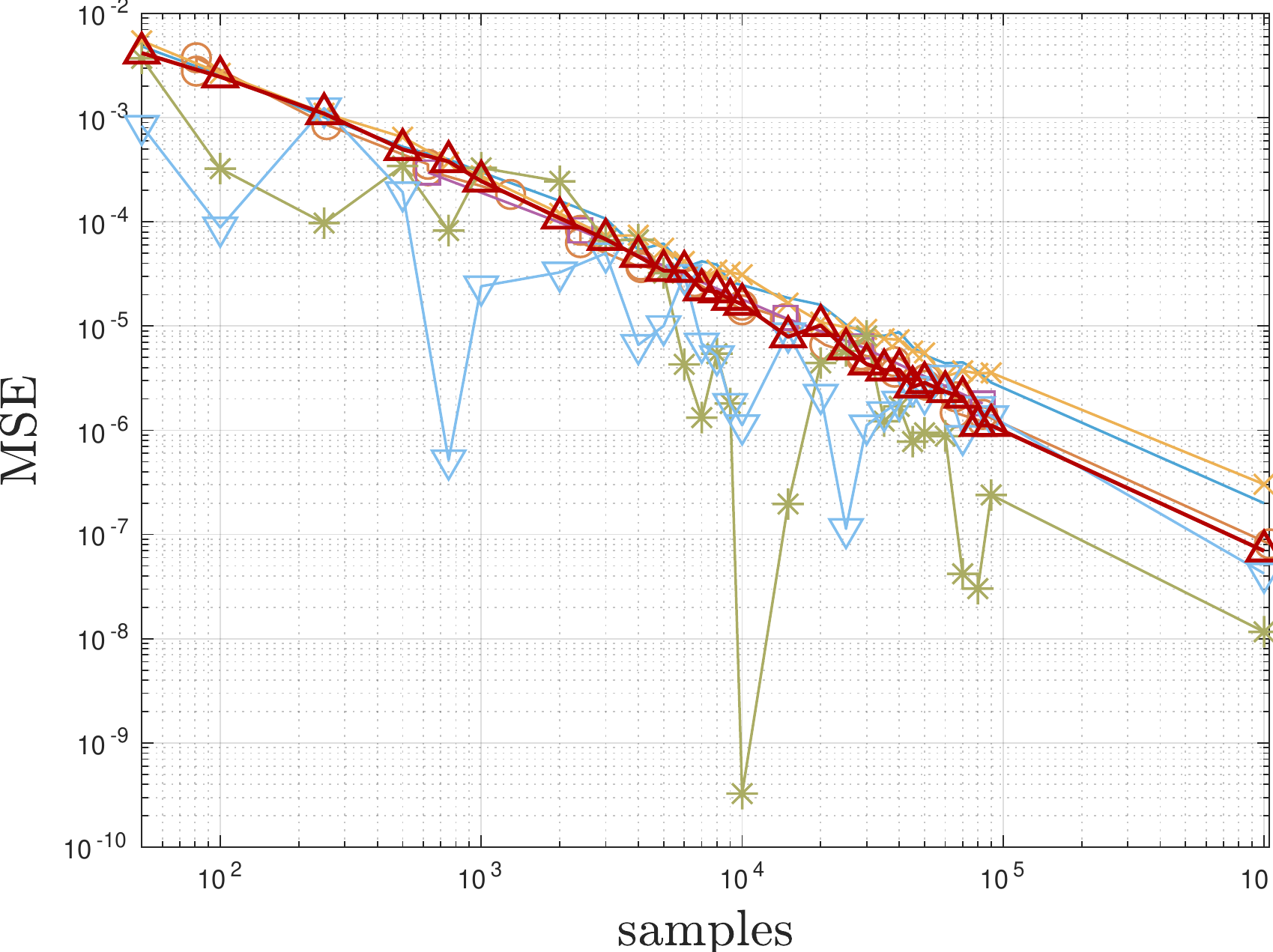}
\end{tabular}
\centering
\includegraphics[scale=0.3]{images/analytical/legend/legend.pdf}
\caption{\label{fig:analytical_discont}Error convergence of 7 samplers on 2D and 4D discontinuous analytic integrands of varied complexity. Our method (KDT) exhibits similar performance to jittered sampling, which is comparable to state-of-the-art QMC samplers in 2D, without limiting the allowed number of samples. In 4D the convergence of KDT again matches that of jittered sampling, both of which appear inferior to popular QMC methods.}
\end{figure*}

%___________________________________________________
\subsection{Evaluation on rendered images} \label{subsec:renderings}

The experiments on the photorealistically rendered images are performed using PBRT version 3~\cite{Pharr2010PBRT}. We use  the \emph{Empirical Error Analysis} toolbox~\cite{Subr16Fourier} to calculate errors due to various samplers. We implemented jittered kd-tree stratification within PBRT as a \emph{pixel sampler}.~i.e.~a jittered kd-tree stratified sample set is generated for each pixel independently. Since our sampler approaches that of random sampling for $n\ll 2^d$, we also implemented a padded version for 2D subspaces, which we hereby refer to as KDT2Dpad.

We compare KDT2Dpad against other pixel samplers found within PBRT, such as stratified (Jitter2Dpad) and Random, as well as QMC methods Halton and Sobol. The QMC methods are  implemented as \emph{global samplers}, which generate samples for the whole image plane and associate pixel coordinates to the indices of samples from the appropriate sequences. This scheme ensures that each pixel is assigned the correct number of samples. We also compare our method against Bush's Orthogonal Arrays stratification~\cite{jarosz19orthogonal} with strength 2 (OAbushMJ2) implemented as a pixel sampler.

We evaluated our method using two metrics, RGB-MSE (RGBMSE) and Log-Luminance-MSE (LLMSE), both computed on the high-dynamic range images output by the renderer. The images shown in the paper are tonemapped for visualisation, but we include the original images along with an html browser as supplementary material. The former is computed as the squared norm of the differences in RGB space:
\begin{equation*}
|R_{ref}(p)-R_{test}(p)|^2+|G_{ref}(p)-G_{test}(p)|^2+|B_{ref}(p)-B_{test}(p)|^2,
\end{equation*}
for a pixel $p$, where $R_{ref}, G_{ref}, B_{ref}$ are the linear RGB channels of the reference image and $R_{test}, G_{test}, B_{test}$ are those of the test image respectively. The Log-Luminance-MSE metric measures the error in the perceived luminance,
\begin{equation*}
|\log (L_{ref}(p)) - \log(L_{test} (p))|^2, 
\end{equation*}
where $L_{ref}$ and $L_{test}$ are the luminances of the reference and test image respectively. We tested with other perceptually-based metrics such as SSIM but did not observe significant differences with the simple Log-Luminance-MSE metric. 

We rendered ORB and ORB-GLOSSY scenes, shown in Figures~\ref{fig:Orb} and~\ref{fig:orb-gloss}, which contain a light source and an orb placed inside a glossy sphere with an occluder above the orb model. These two versions of a similar scene feature 19 and 41 dimensional light paths respectively due to different material and rendering parameters. We used 529 spp for most samplers and 512 spp for QMC methods. Finally, the PAVILION scene, in Figure~\ref{fig:pav}, contains high frequency textures, various materials and complex geometry, resulting in 43 dimensional light paths. We use 3,969 spp for samplers that support it, and 4,096 spp for QMC methods. Reference scenes are rendered with 10,000 Random spp. 

For the ORB scene, our method KDT2Dpad performs comparably good in both metrics, matching, and, in the case of RGBMSE, surpassing the performance of state-of-the-art QMC methods. In ORB-GLOSS, KDT2Dpad exhibits the least error in both metrics for the shadow region considered. Sobol and Jitter2Dpad methods exhibit some structured artifacts. In the PAVILION scene, Sobol performs best in terms of RGBMSE, and KDT2Dpad is better than Jitter2Dpad and Halton. All samplers perform similarly with respect to LLMSE with Random sampling being consistently the worst, followed by Jitter2Dpad.

We tested the intricate combinations of high-dimensional light paths, gloss, textures and defocus by rendering the ORB-GLOSS-DOF scene (Fig.~\ref{fig:orb-gloss-fov} and Fig.~\ref{fig:orb-gloss-dof-insets})  which is identical to ORB-GLOSS but with a wide aperture (shallow depth of field). We observed that Halton performs well, as expected, and Sobol exhibits tell-tale structured artifacts. Our KDT sampler performs well in areas with complex light paths (high-dimensional paths, gloss, texture and depth of field). Surprisingly, jittered sampling performs best with respect to multi-bounce paths (like on the side of the cuboidal occluder). 

\paragraph{Error plots} The box plots accompanying rendered results show the mean (dashed horizontal line), median (solid line), quantiles around the median (shaded box) and the upper and lower fences (end points of whiskers). The mean value is printed above each sampler.

%\paragraph{Supplemental material} For all scenes, we have selected arbitrary portions of the image for the figures. The supplementary material contains an html viewer where reviewers may move the mouse over the reference image to examine all regions. An additional scene -- a standard CORNELL-BOX scene -- is included in the supplemental material, rendered with 12 dimensional light paths and 121 samples per pixel (128 spp for QMC methods). In the CORNELL-BOX scene all samplers exhibit similar performance regardless of the metric, with Sobol having the lowest mean Log-Luminance-MSE, and KDT2Dpad having the lowest mean RGB-MSE. 

\begin{figure*}[h]
\centering
\begin{tabular}{@{}c@{ }c@{ }c@{ }c@{}}
\includegraphics[scale=0.14]{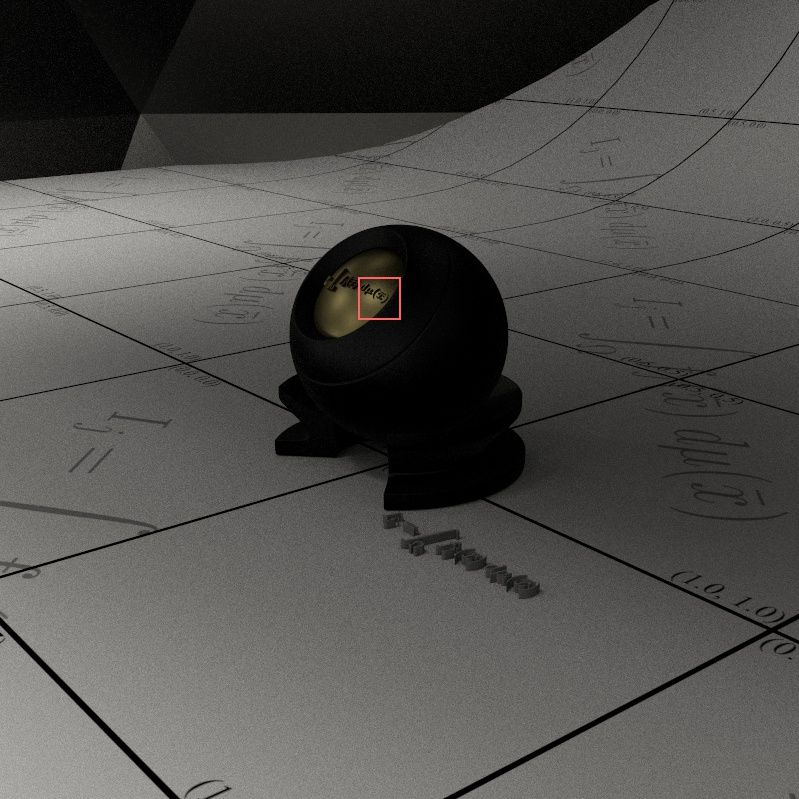} &
\includegraphics[scale=0.21]{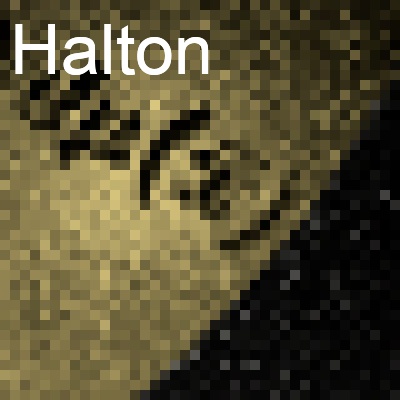}&
\includegraphics[scale=0.21]{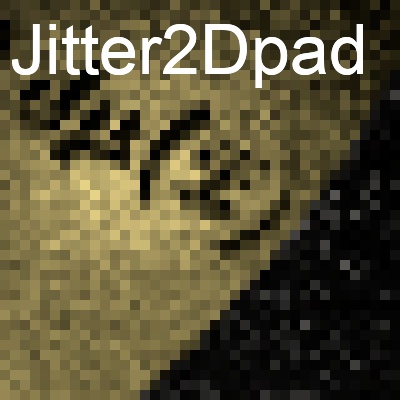}&
\includegraphics[scale=0.21]{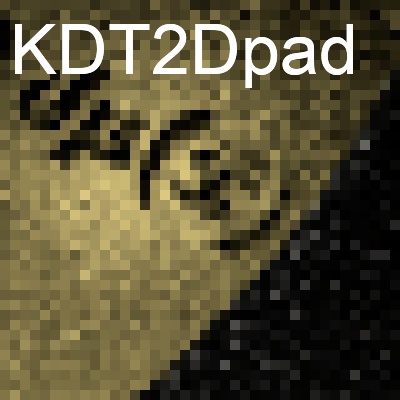}\\
\includegraphics[scale=0.21]{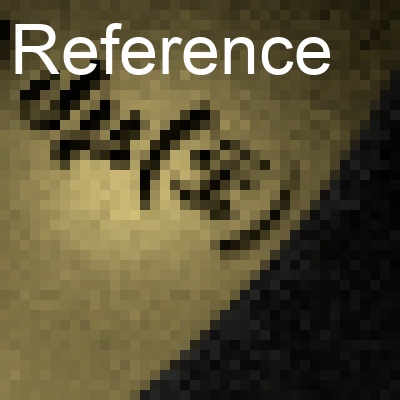} &
\includegraphics[scale=0.21]{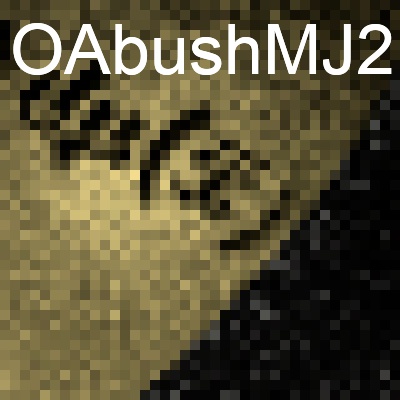}&
\includegraphics[scale=0.21]{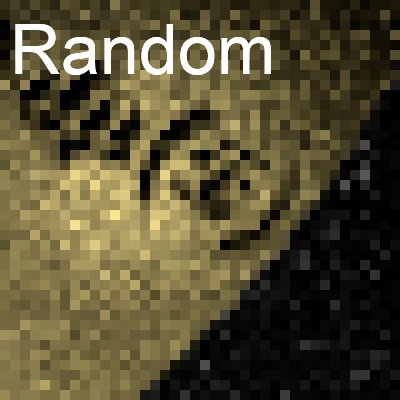}&
\includegraphics[scale=0.21]{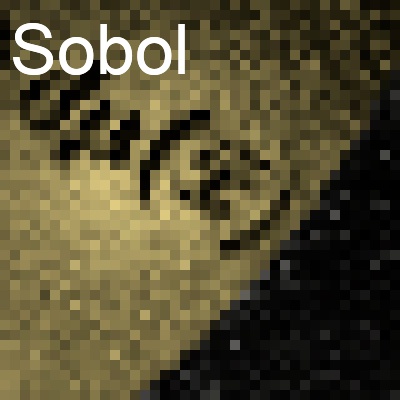}
\end{tabular}
\begin{tabular}{@{}c@{  }c@{}}
\includegraphics[clip,scale=0.38]{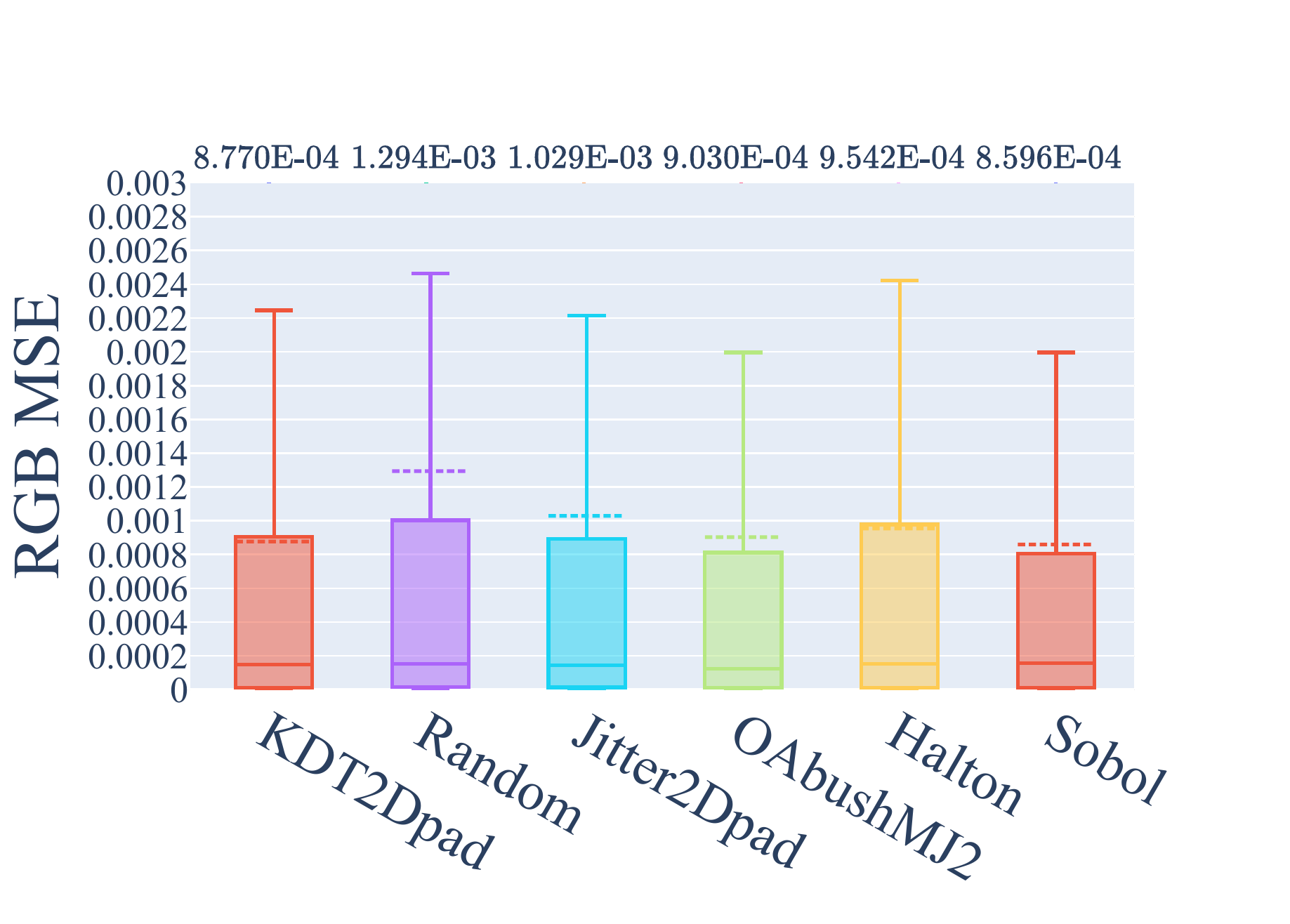}&
\includegraphics[clip,scale=0.38]{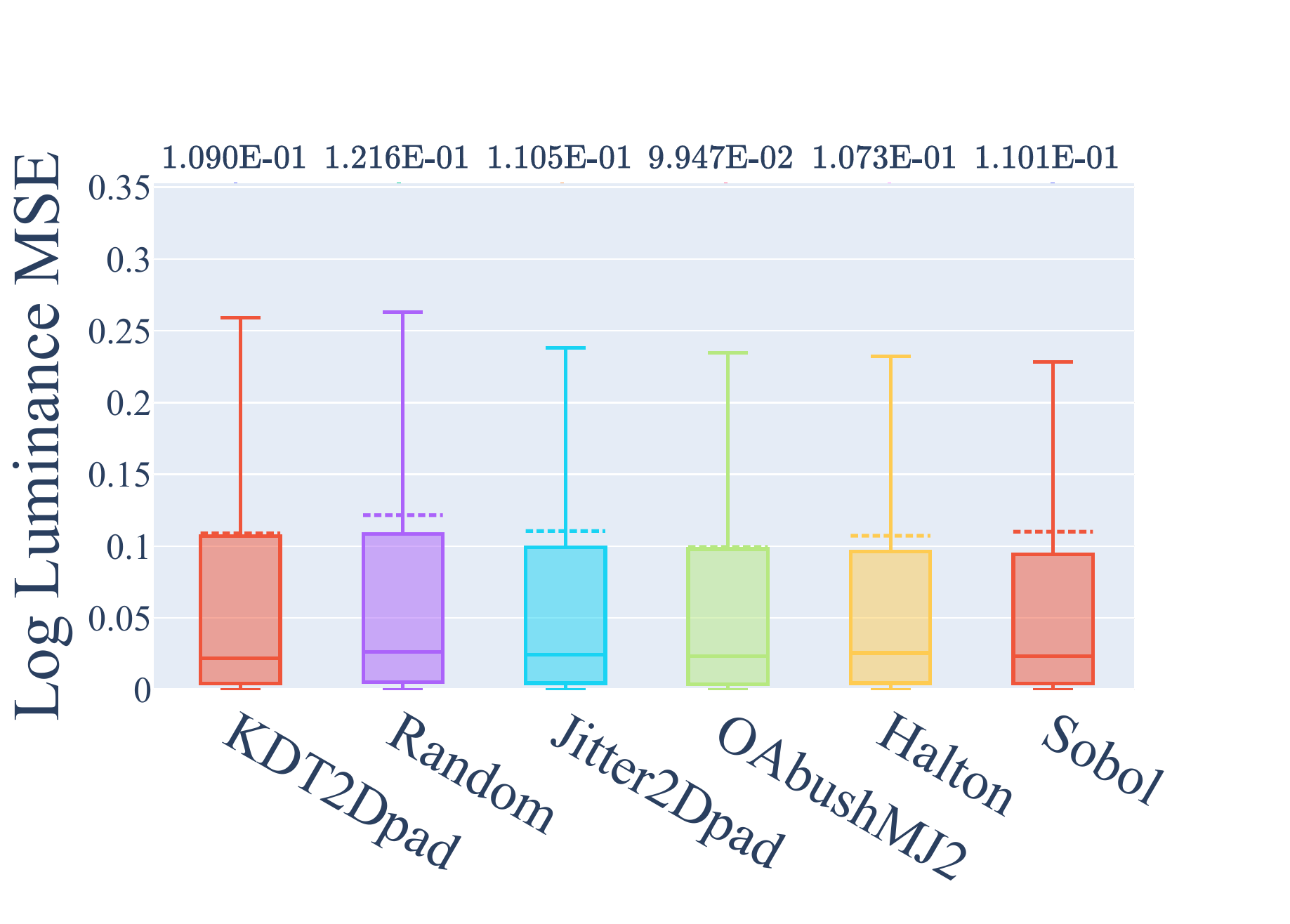}
\end{tabular}
\caption{Our sampler (KDT2DPad) performs comparably to state of the art techniques in this 19 dimensional ORB scene. The region of interest is highlighted in the reference image (top left).Plots comparing $L_2$-norm mean squared error of the RGB values (RGB-MSE) (bottom left) and the log Luminance mean squared error (Log-Luminance-MSE) (bottom right) of the pixels within the region of interest are shown, along with enlarged versions of the highlighted region for each sampler to aid visual inspection (bottom row). Note: we have chosen a representative crop within the image.}% The supplementary material contains an HTML tool that can be used to inspect (and compare) all parts of the image using the mouse.
\label{fig:Orb}
\end{figure*}

\begin{figure*}[h]
\centering
\begin{tabular}{@{}c@{ }c@{}}
\includegraphics[scale=0.14, trim=0.0cm 13cm 0.0cm 0.0cm]{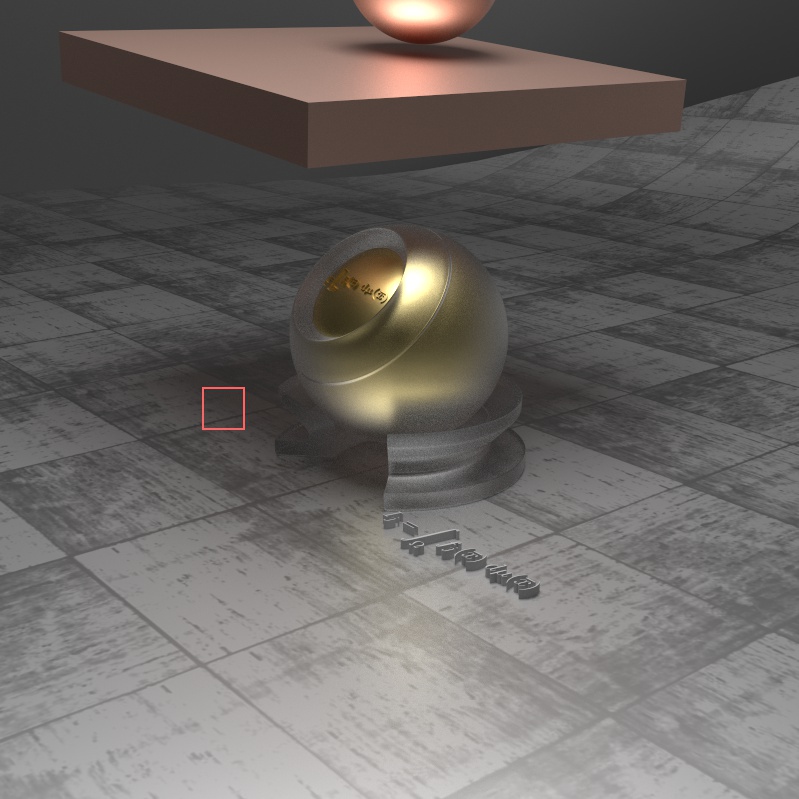}
&\begin{tabular}{@{ }c@{ }c@{ }c@{}}
\includegraphics[scale=0.21]{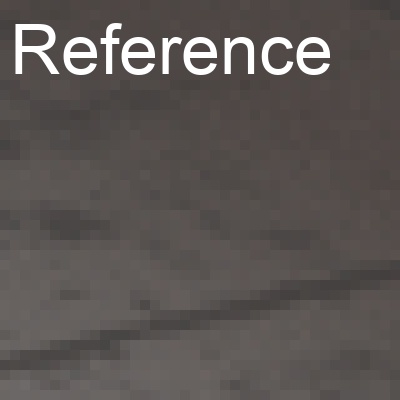} &
\includegraphics[scale=0.21]{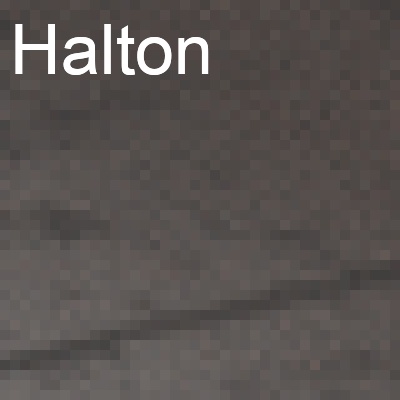}&
\includegraphics[scale=0.21]{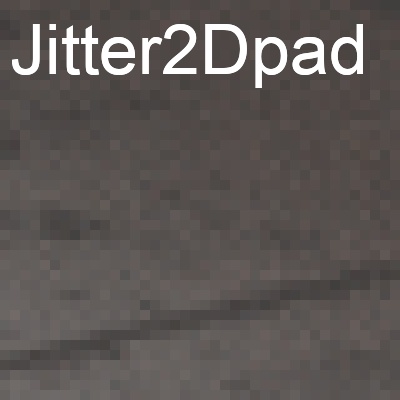}\\
\includegraphics[scale=0.21]{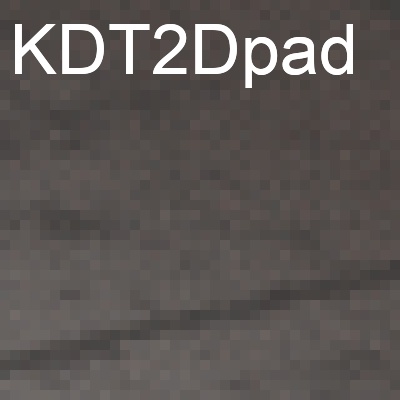}&
\includegraphics[scale=0.21]{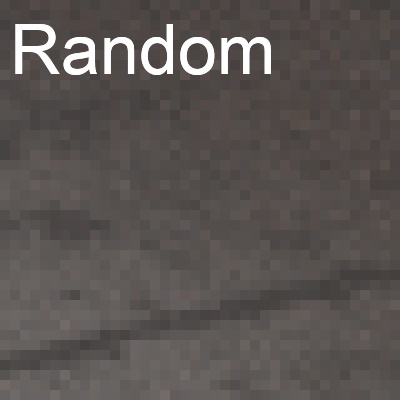}&
\includegraphics[scale=0.21]{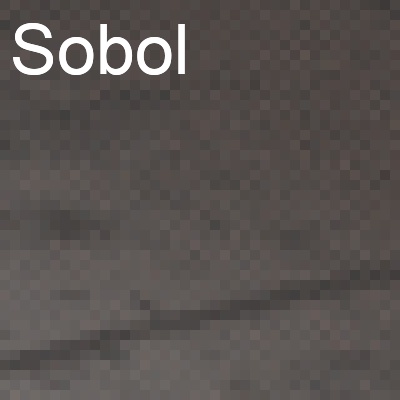}
\end{tabular}
\end{tabular}
\begin{tabular}{@{}c@{  }c@{}}
\includegraphics[clip,scale=0.38]{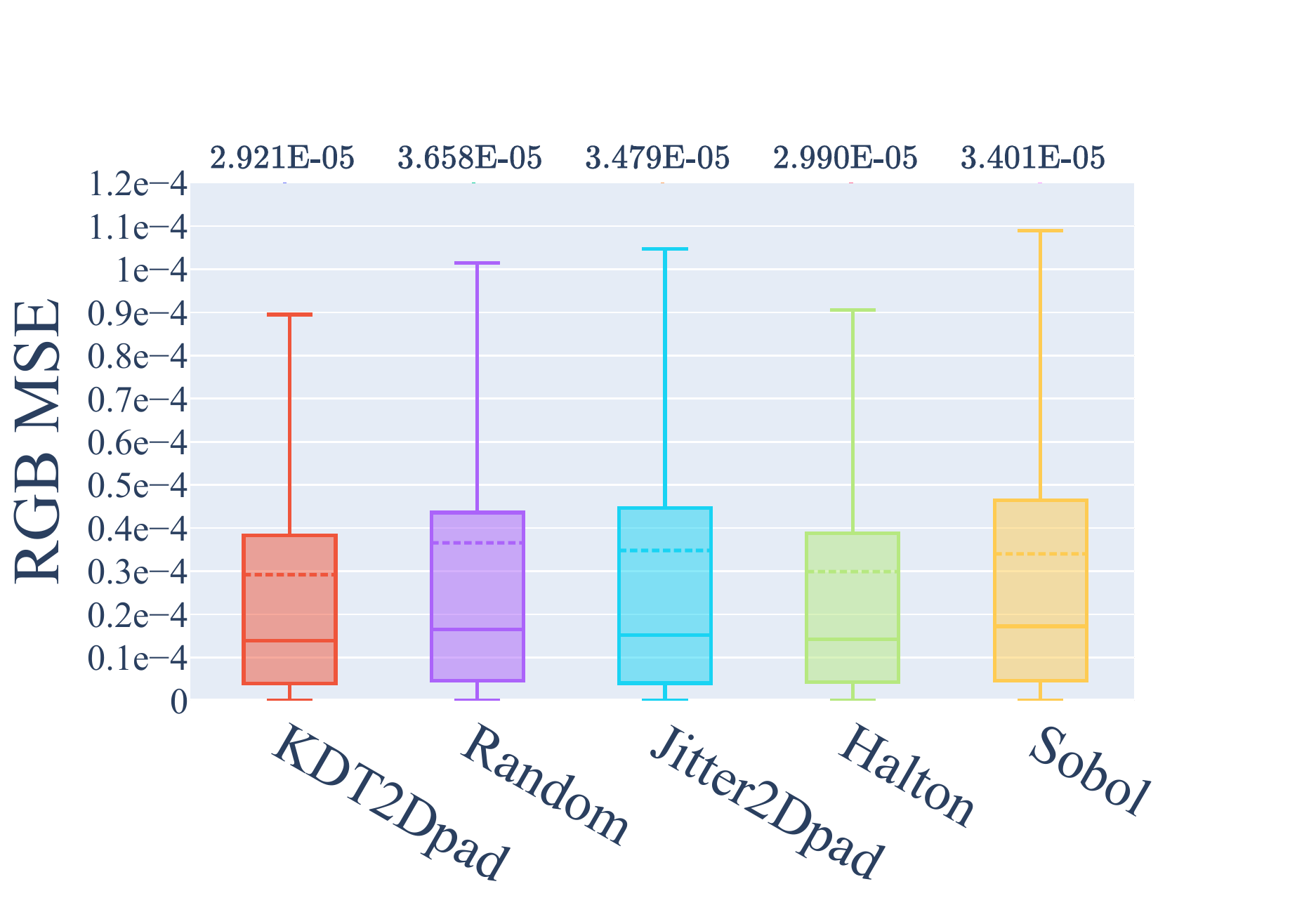}&
\includegraphics[clip,scale=0.38]{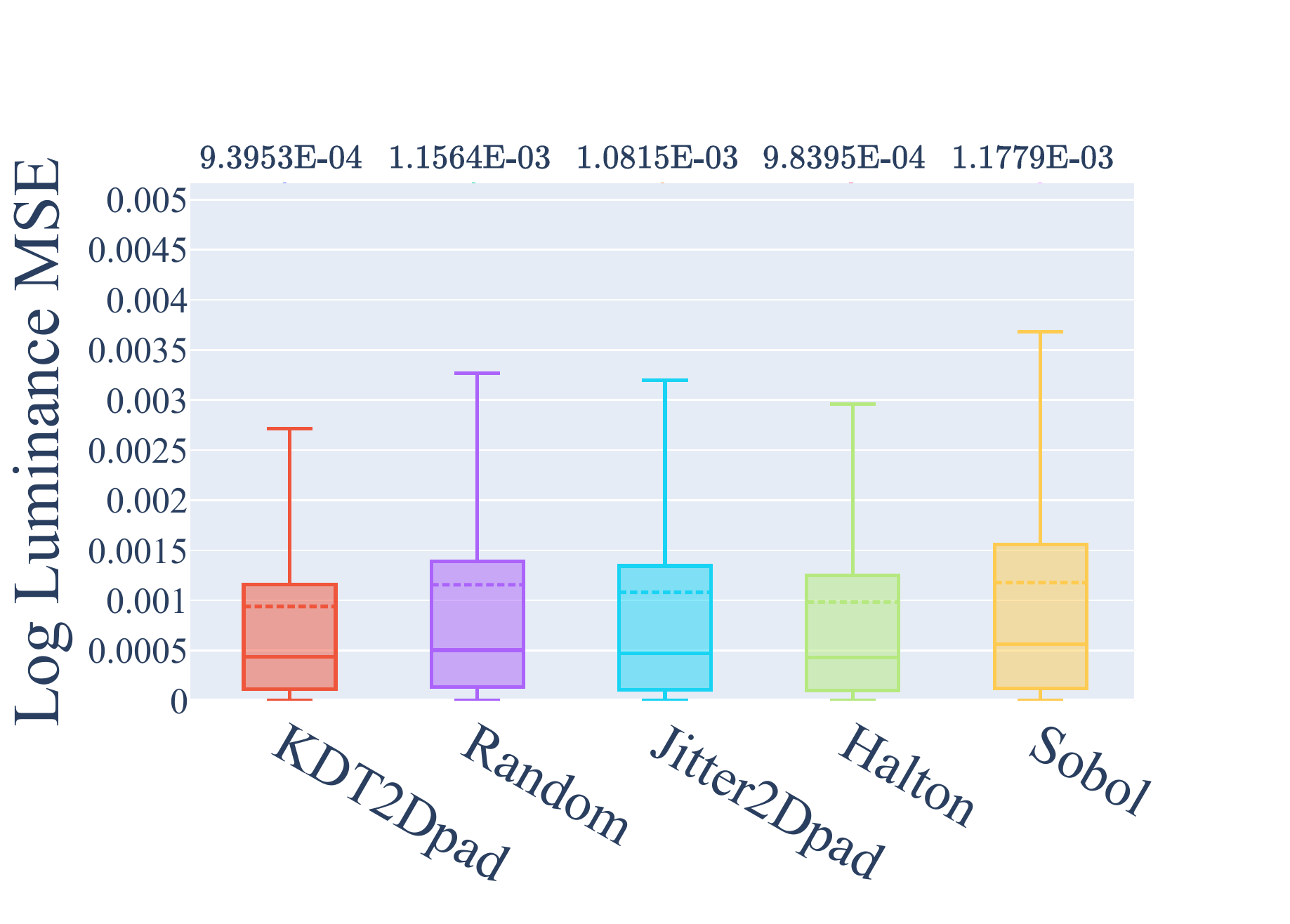}%
\end{tabular}
\caption{KDT2Dpad clearly outperforms all other considered samplers for the indicated shadow region of this 41 dimensional ORB-GLOSS scene. The region of interest is highlighted in the reference image (left).Plots comparing $L_2$-norm mean squared error of the RGB values (RGB-MSE) (bottom left) and the log Luminance mean squared error (Log-Luminance-MSE) (bottom right) of the pixels within the region of interest are shown, along with enlarged versions of the highlighted region for each sampler to aid visual inspection (top right).}% Please refer to supplementary material for inspection of other regions of the image.
\label{fig:orb-gloss}
\end{figure*}

\begin{figure*}[h]
\centering
\begin{tabular}{@{}c@{}}
\includegraphics[scale=0.15]{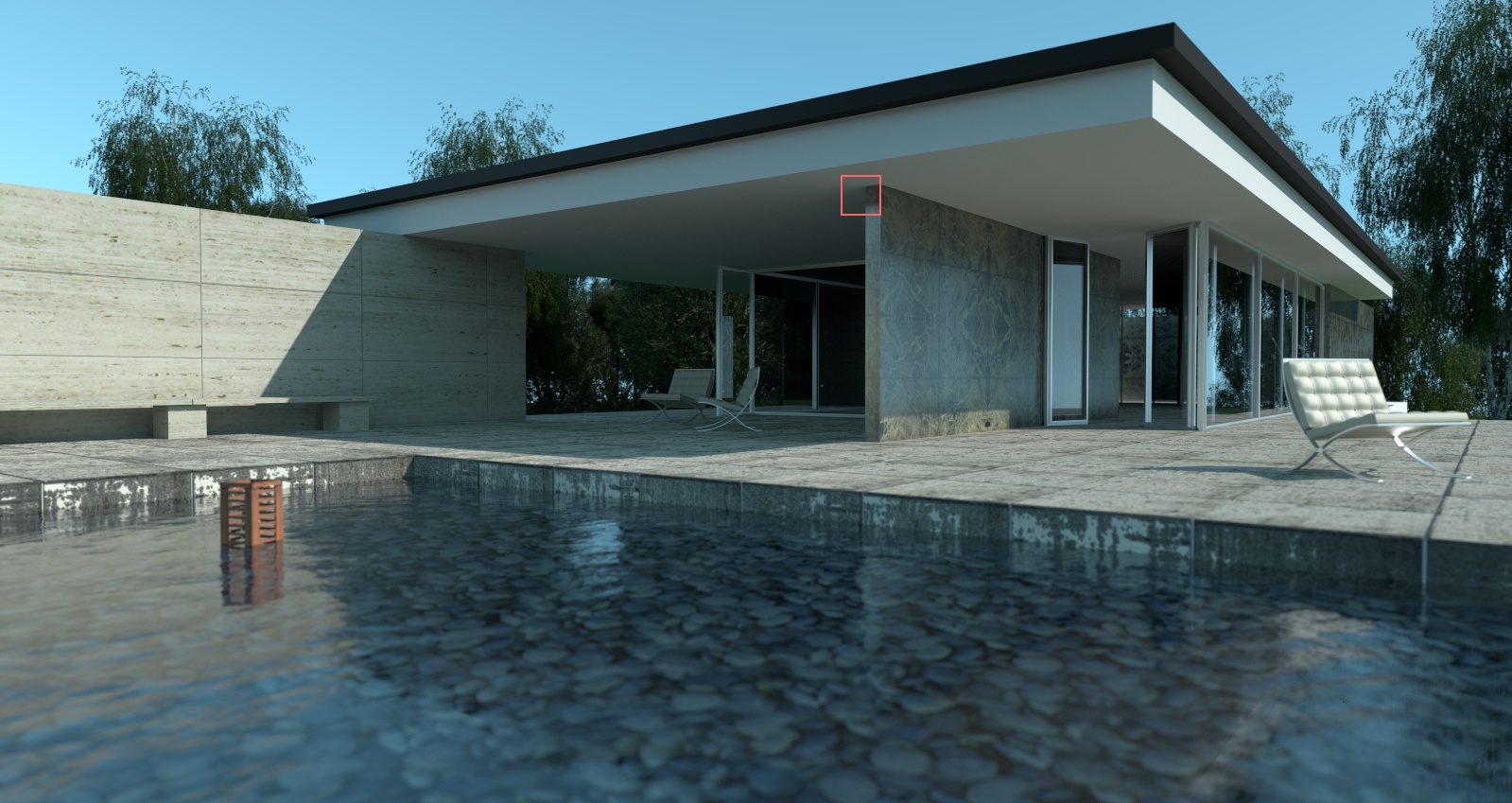}\\
\begin{tabular}{@{ }c@{ }c@{ }c@{}}
\includegraphics[scale=0.21]{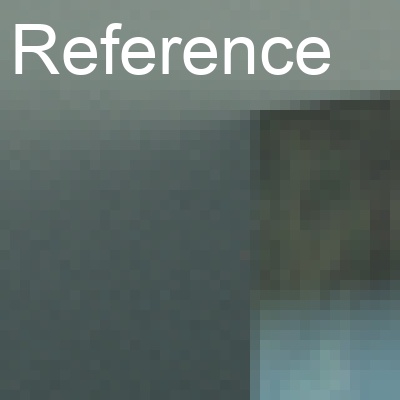} &
\includegraphics[scale=0.21]{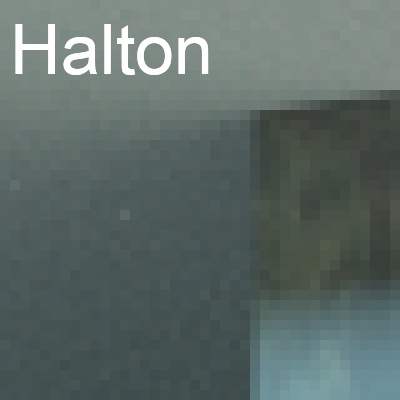}&
\includegraphics[scale=0.21]{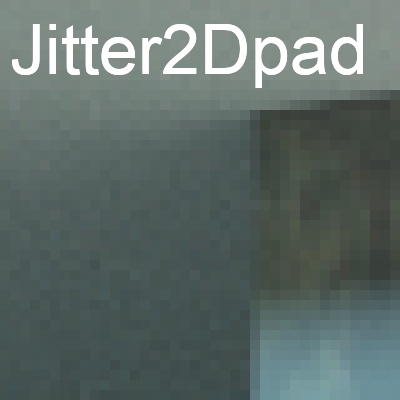}\\
\includegraphics[scale=0.21]{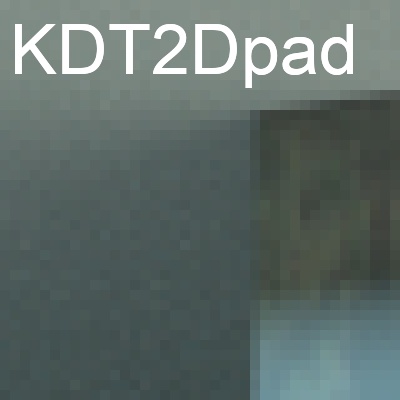}&
\includegraphics[scale=0.21]{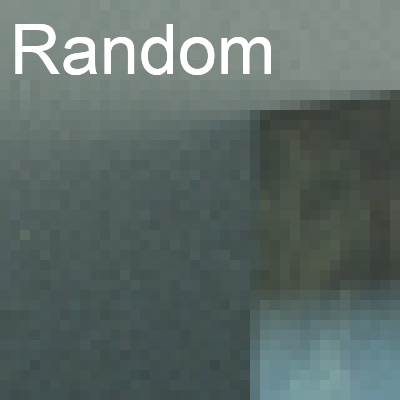}&
\includegraphics[scale=0.21]{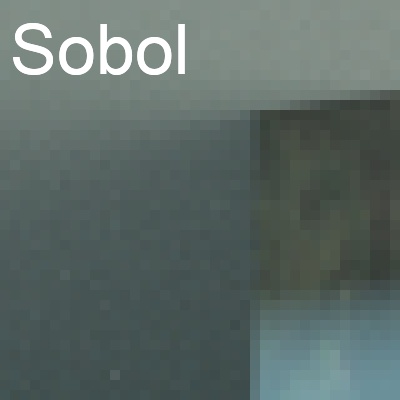}
\end{tabular}
\end{tabular}
\begin{tabular}{@{}c@{  }c@{}}
\includegraphics[clip,scale=0.38]{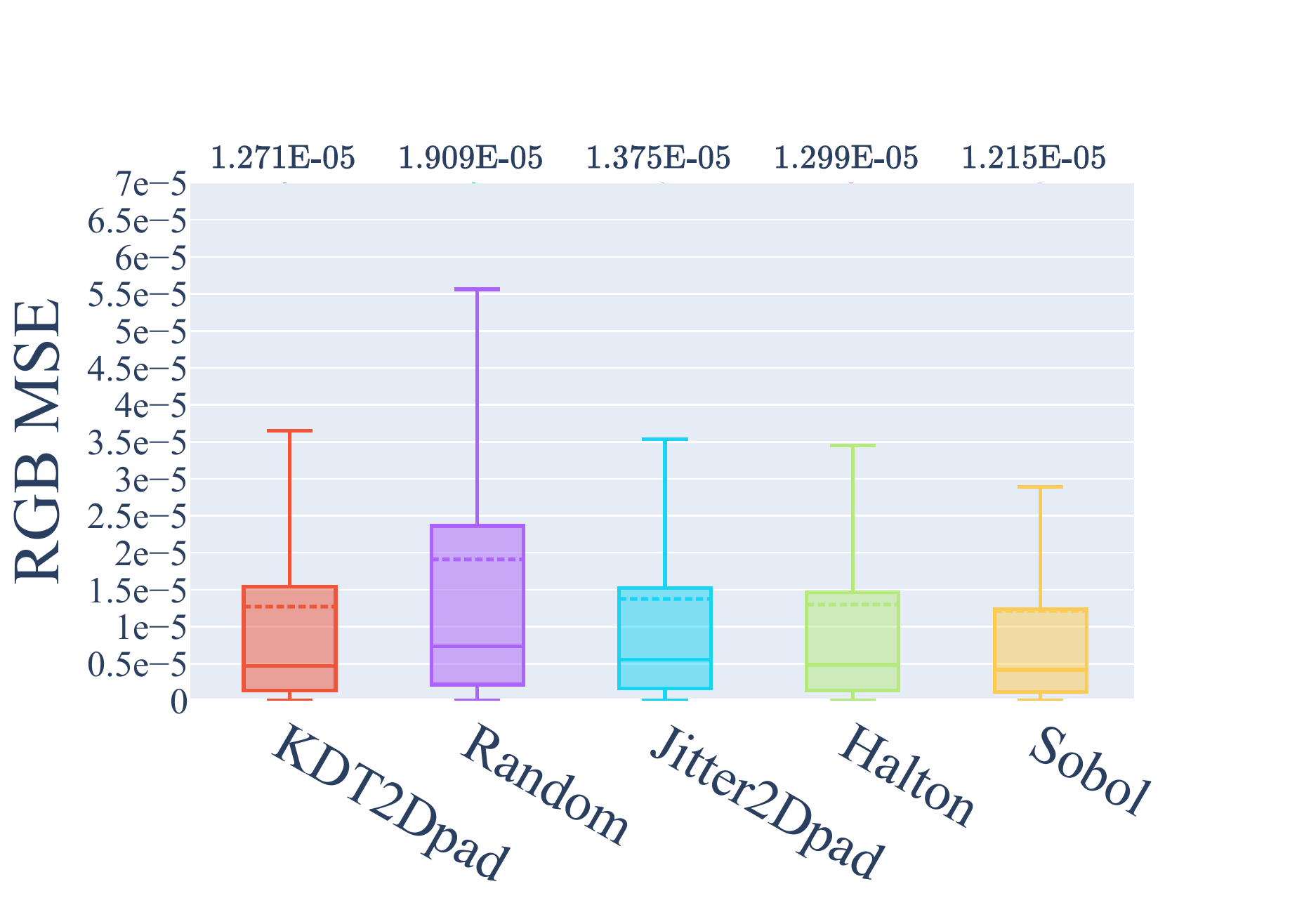}&
\includegraphics[clip,scale=0.38]{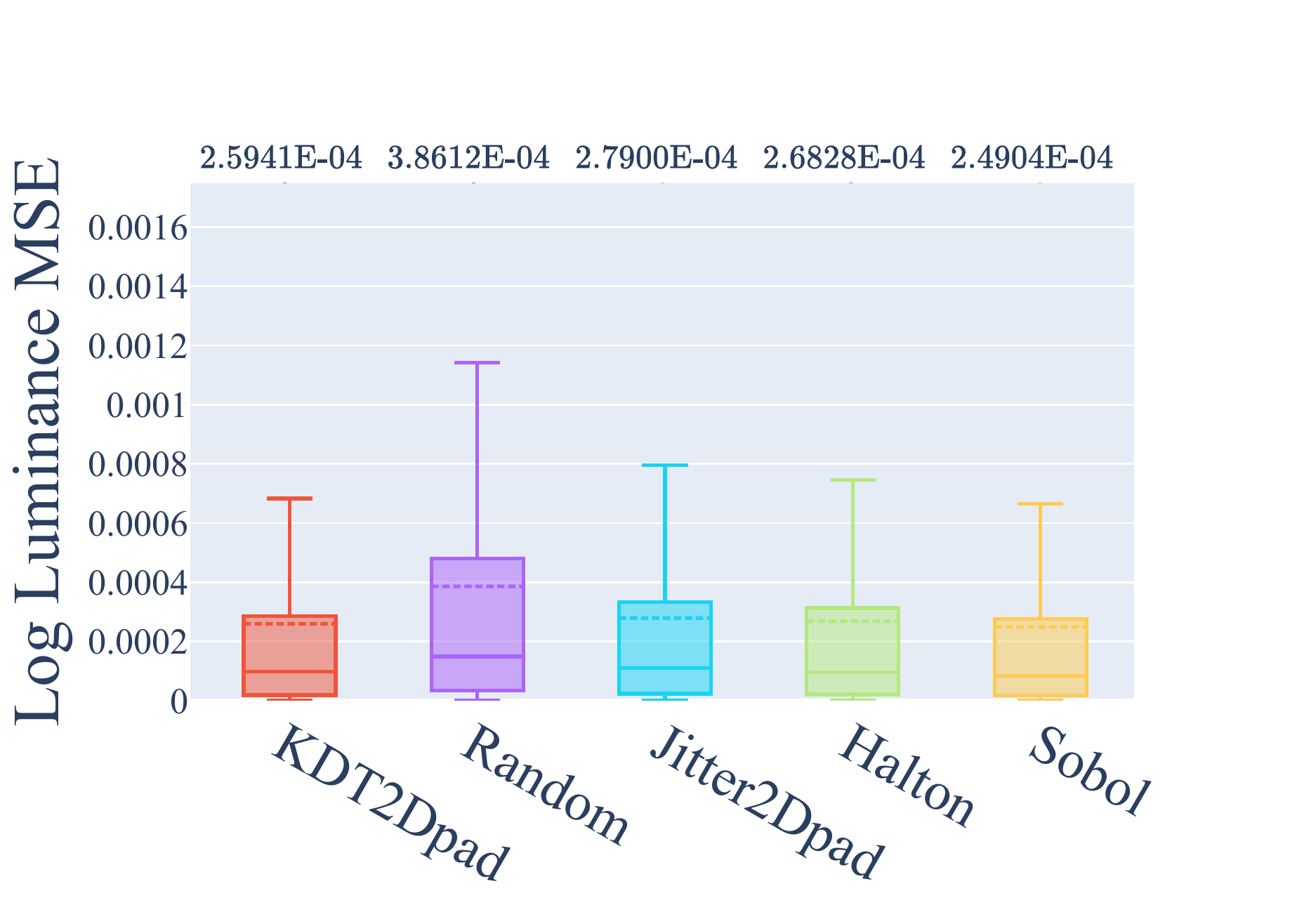}
\end{tabular}
\caption{QMC methods clearly outperform all others in terms of RGB-MSE for the shaded textured region of the PAVILION scene. Regarding Log-Luminance-MSE all metrics exhibit similar performance, with the worst being Random sampling followed by Jitter2Dpad. The region of interest is highlighted in the reference image (top).Plots comparing $L_2$-norm mean squared error of the RGB values (RGB-MSE) (bottom left) and the log Luminance mean squared error (Log-Luminance-MSE) (bottom right) of the pixels within the region of interest are shown, along with enlarged versions of the highlighted region for each sampler to aid visual inspection (middle).}% Please refer to supplementary material for inspection of other regions of the image.
\label{fig:pav}
\end{figure*}

\begin{figure*}[h]
   \centering
\begin{tabular}{c}
\includegraphics[scale=0.15]{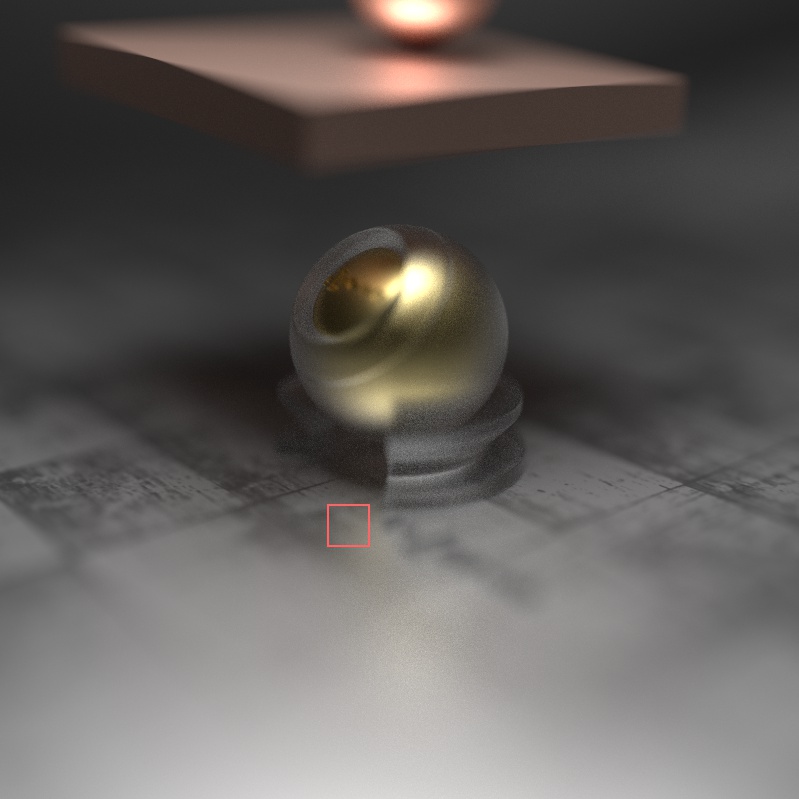}\\
\begin{tabular}{@{}c@{  }c@{}}
\includegraphics[clip,scale=0.35]{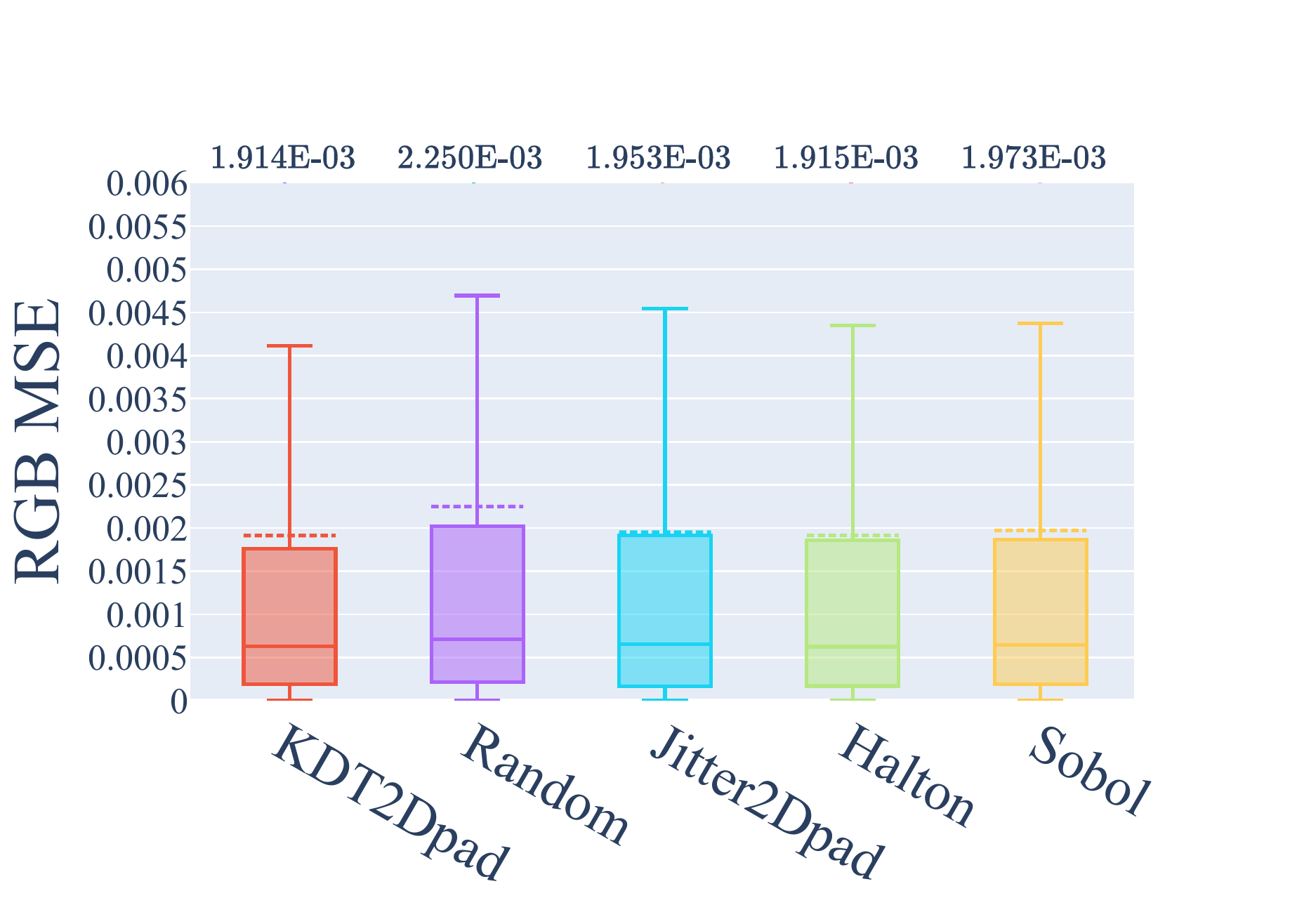}&
\includegraphics[clip,scale=0.35]{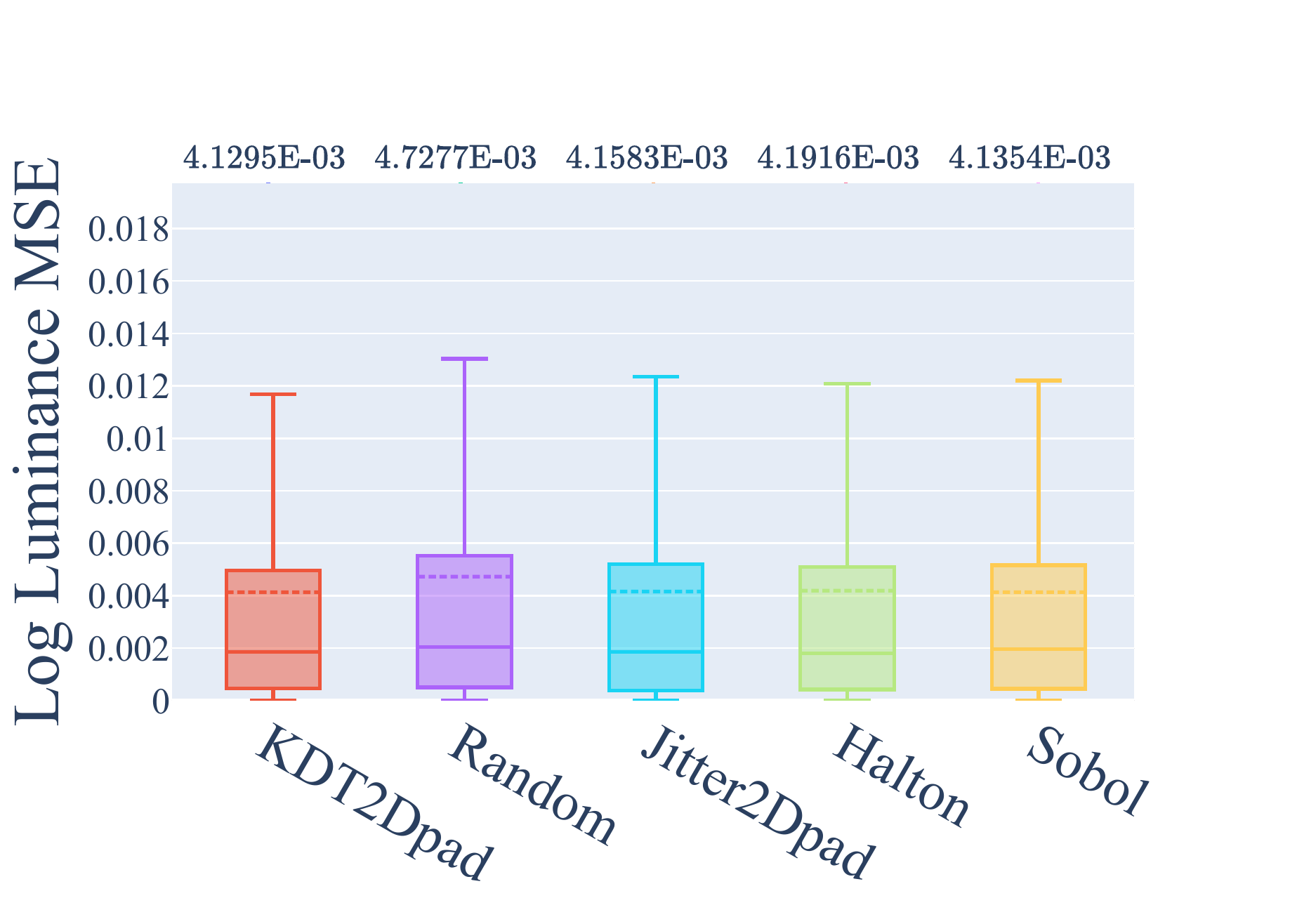}
\end{tabular}
\end{tabular}
\caption{The ORB-GLOSS-DOF scene is identical to ORB-GLOSS but rendered with a shallow depth of field (wide aperture). This scene is interesting because it shows the interplay between high-frequency, high-dimensional light transport due to multiple bounces within the glossy bounding sphere and blur due to defocus. Errors in one region (256spp) are shown here and additional zoomed insets  are shown in Figure~\ref{fig:orb-gloss-dof-insets}. }
\label{fig:orb-gloss-fov}
\end{figure*}

\newcommand\zoomscale{.14}
\begin{figure*}[h]
   \centering

\begin{tabular}{@{}c@{\;}c@{\;}c@{\;}c@{\;}c@{\;}c@{\;}c@{}}
\begin{turn}{90} 
128 spp
\end{turn} &
\includegraphics[width=\zoomscale\linewidth]{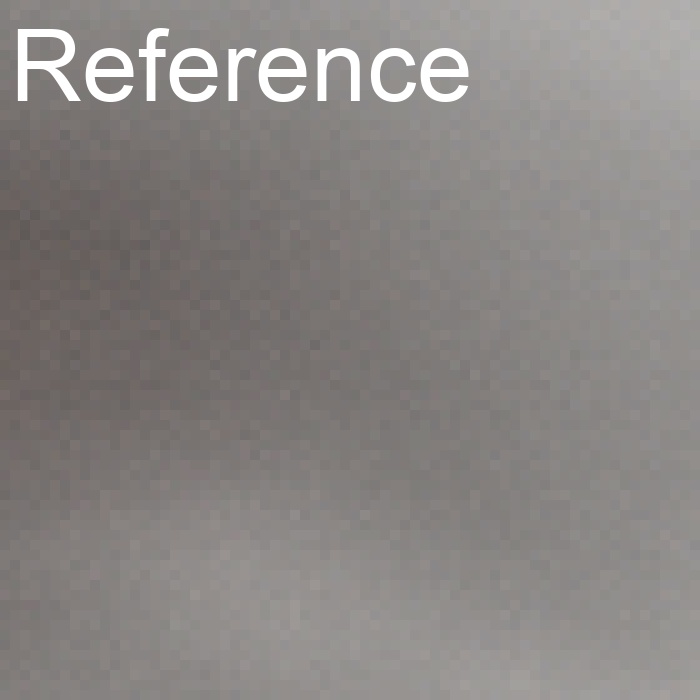}&
\includegraphics[width=\zoomscale\linewidth]{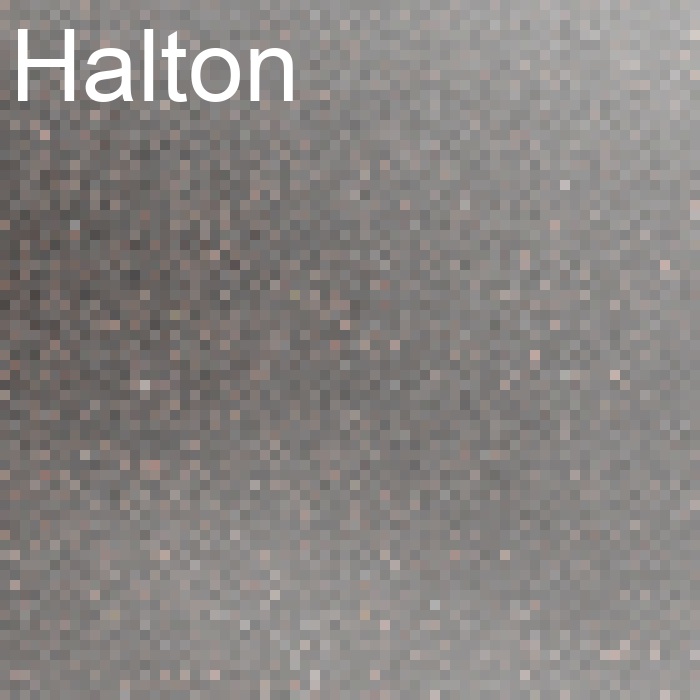}&
\includegraphics[width=\zoomscale\linewidth]{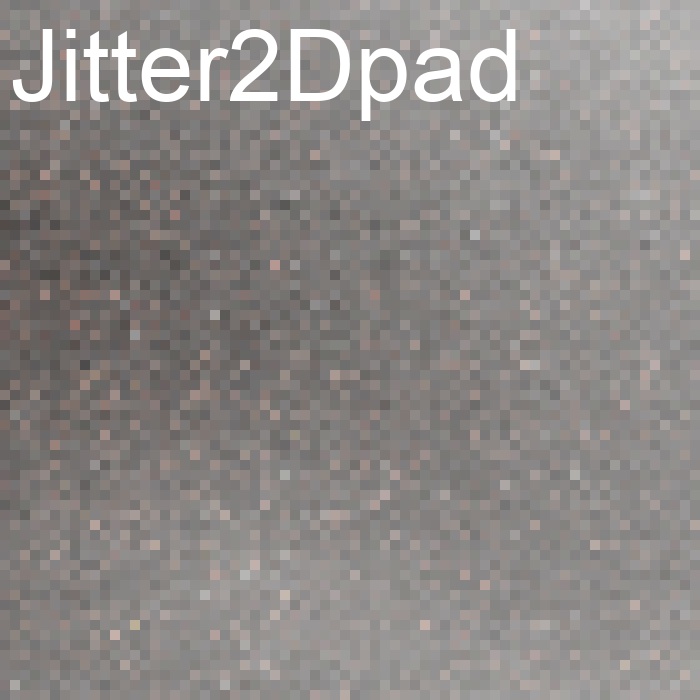}&
\includegraphics[width=\zoomscale\linewidth]{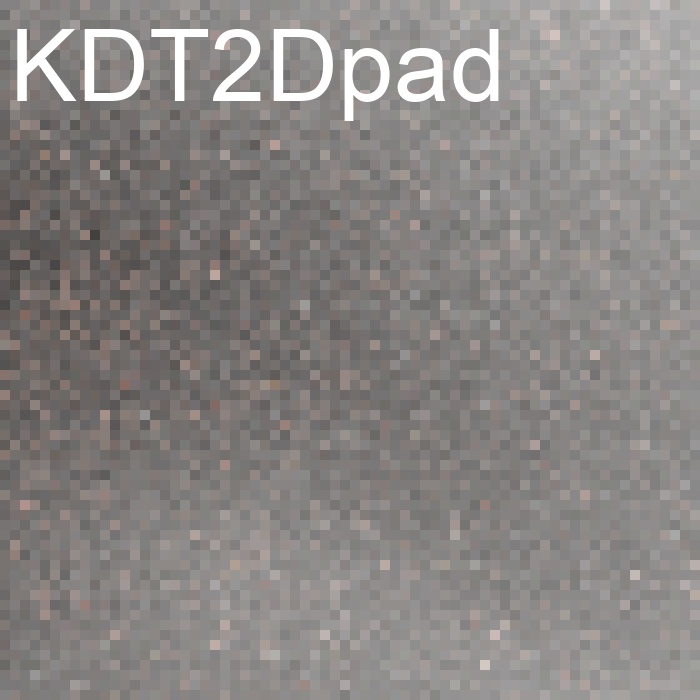}&
\includegraphics[width=\zoomscale\linewidth]{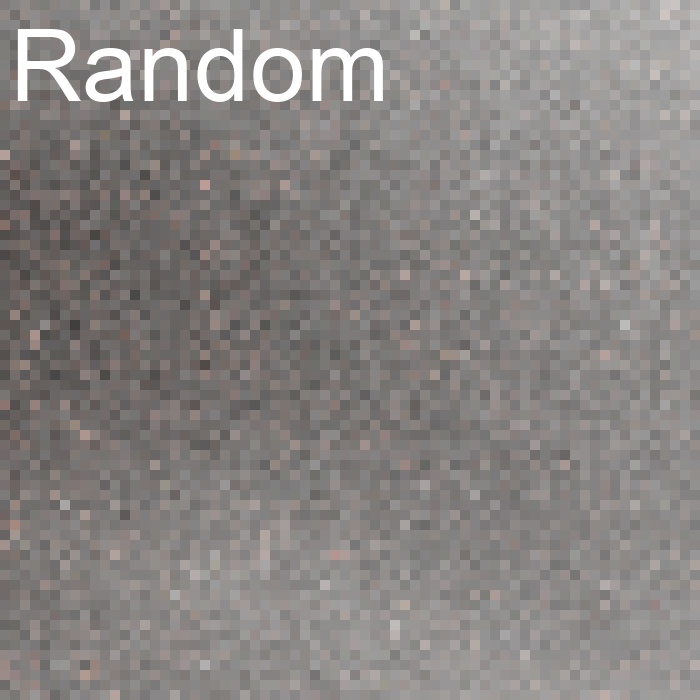}&
\includegraphics[width=\zoomscale\linewidth]{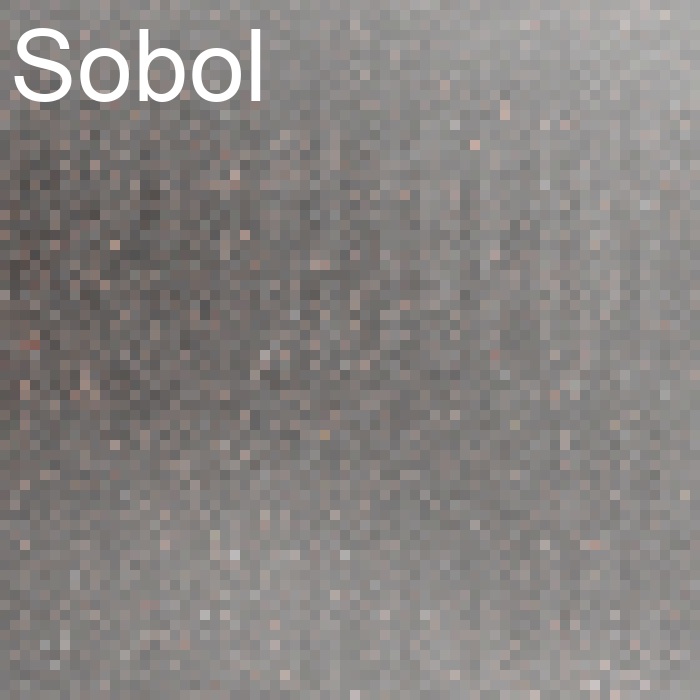} \\
\begin{turn}{90} 
256 spp
\end{turn} &
&
\includegraphics[width=\zoomscale\linewidth]{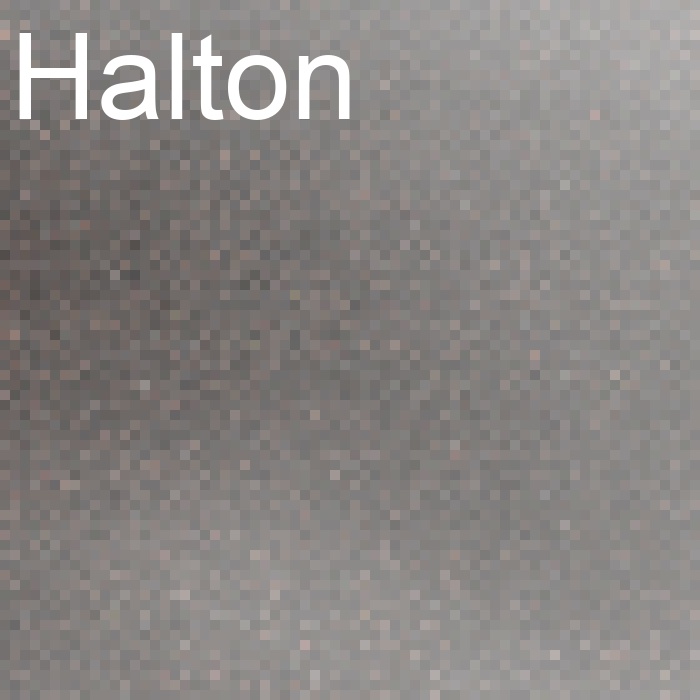}&
\includegraphics[width=\zoomscale\linewidth]{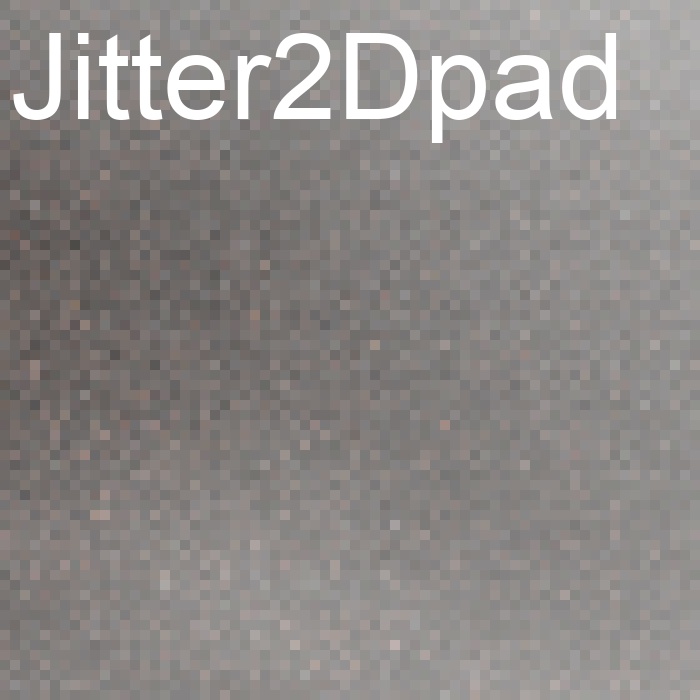}&
\includegraphics[width=\zoomscale\linewidth]{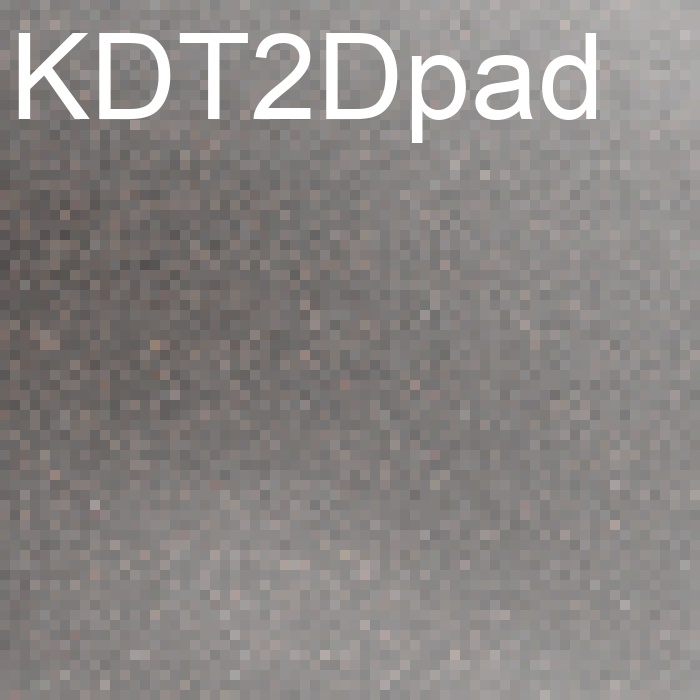}&
\includegraphics[width=\zoomscale\linewidth]{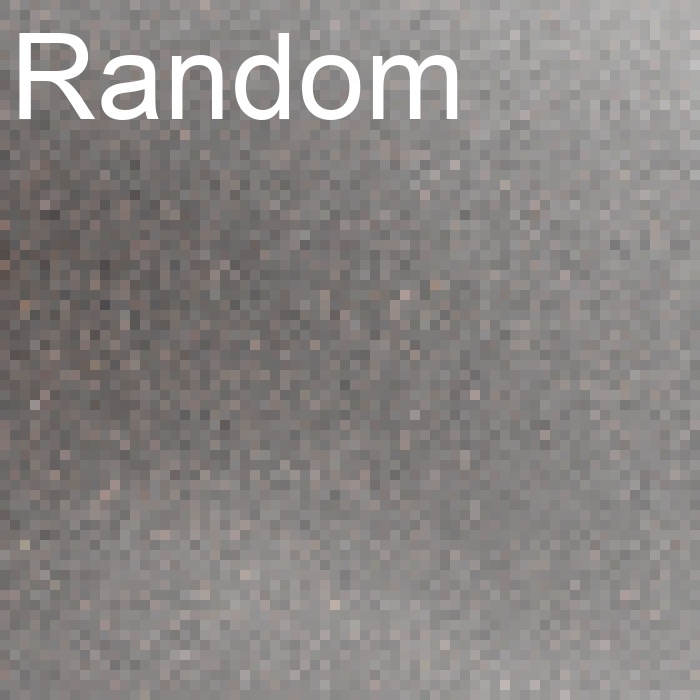}&
\includegraphics[width=\zoomscale\linewidth]{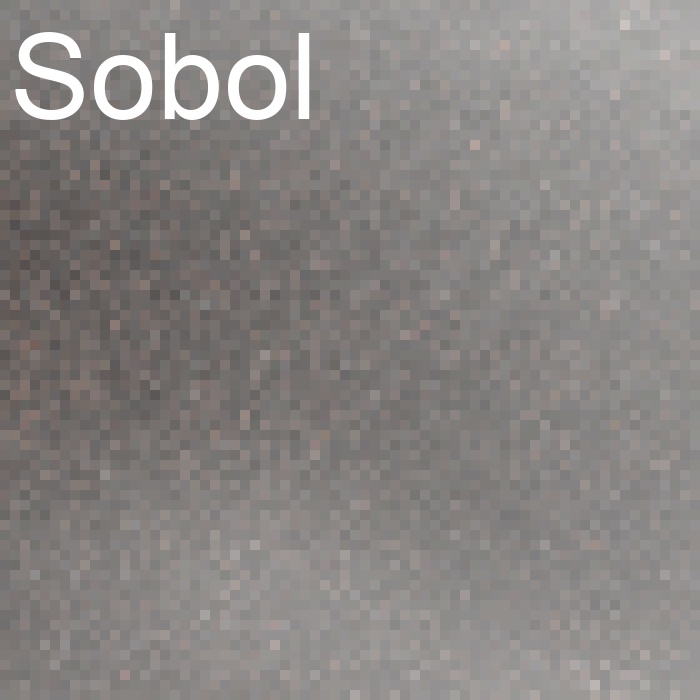} \\
\begin{turn}{90} 
128 spp
\end{turn} &
\includegraphics[width=\zoomscale\linewidth]{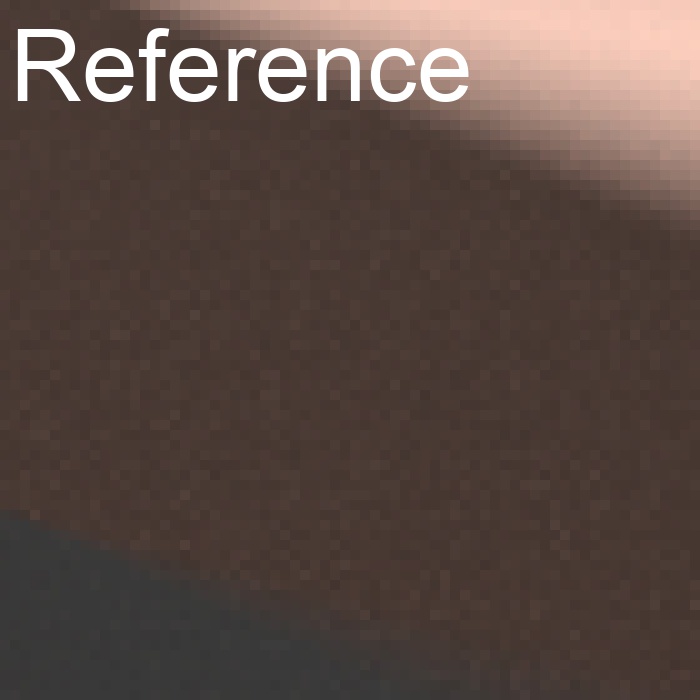}&
\includegraphics[width=\zoomscale\linewidth]{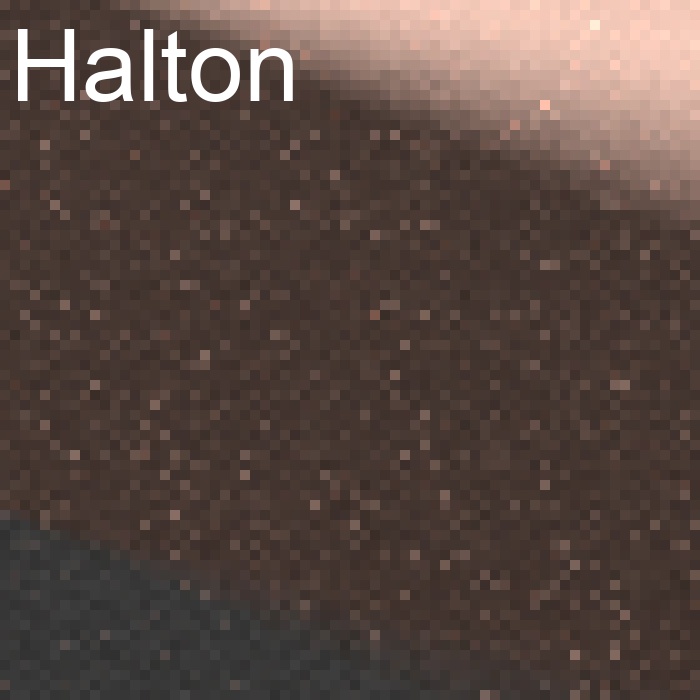}&
\includegraphics[width=\zoomscale\linewidth]{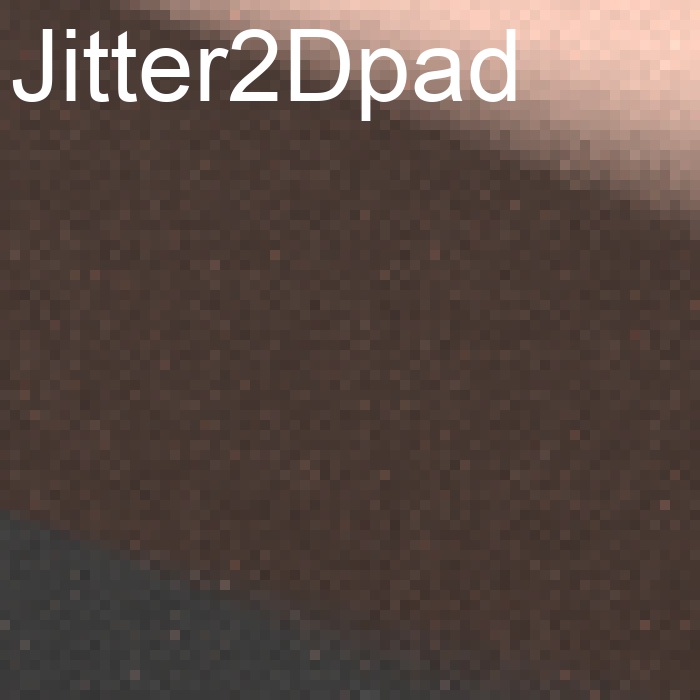}&
\includegraphics[width=\zoomscale\linewidth]{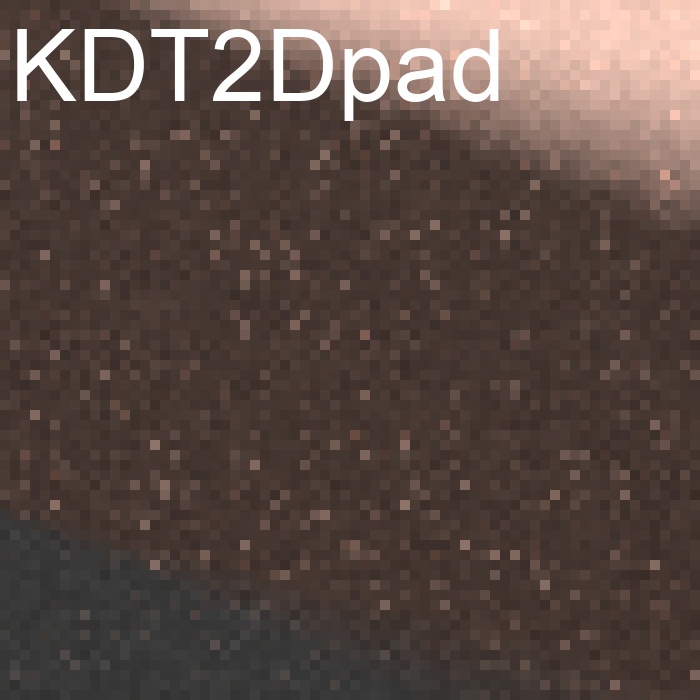}&
\includegraphics[width=\zoomscale\linewidth]{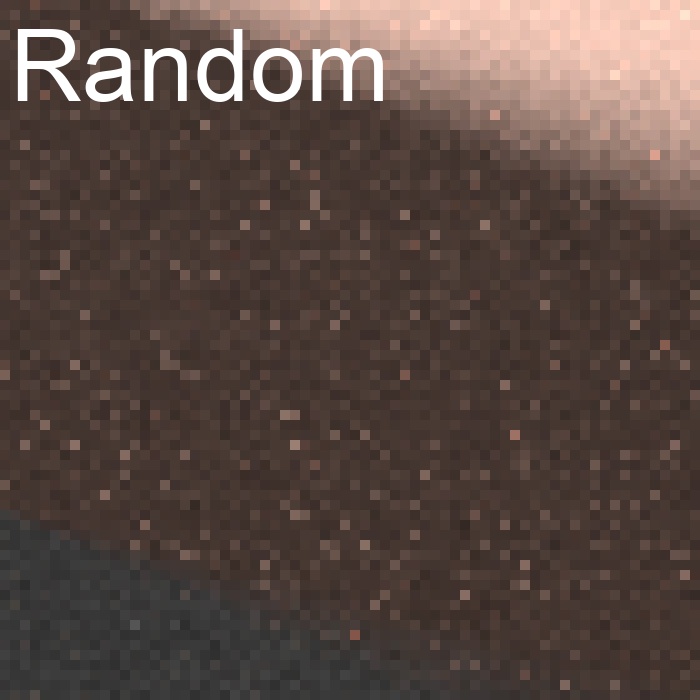}&
\includegraphics[width=\zoomscale\linewidth]{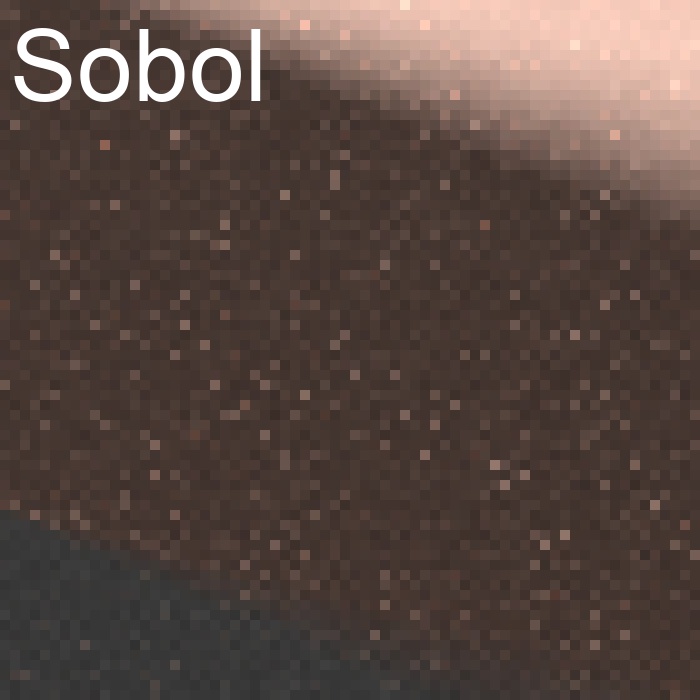} \\
\begin{turn}{90} 
256 spp
\end{turn} &
&
\includegraphics[width=\zoomscale\linewidth]{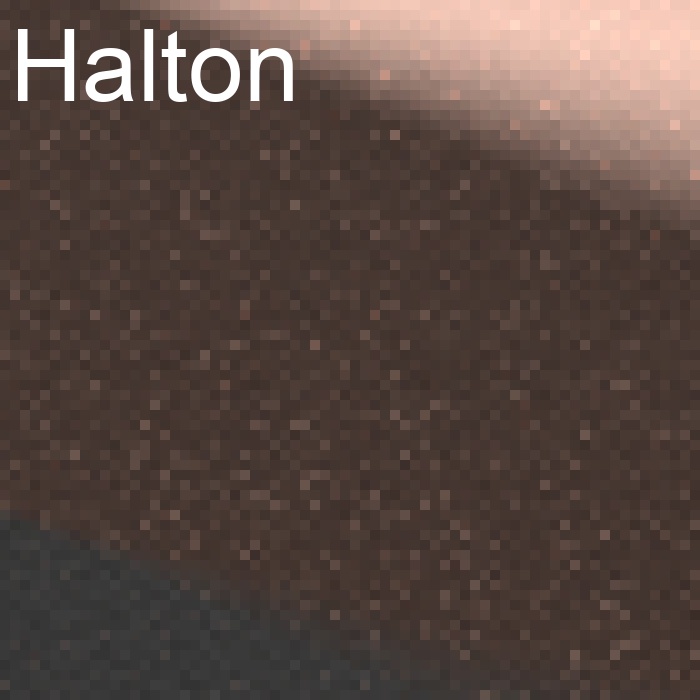}&
\includegraphics[width=\zoomscale\linewidth]{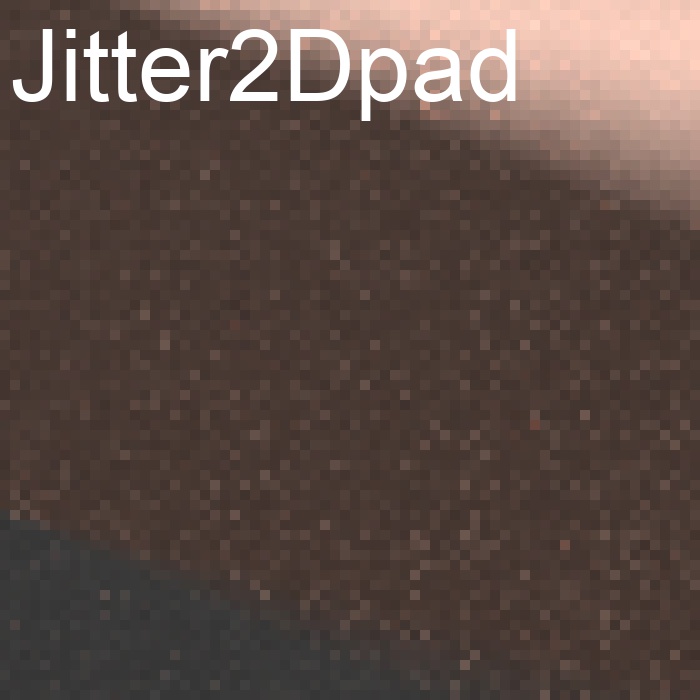}&
\includegraphics[width=\zoomscale\linewidth]{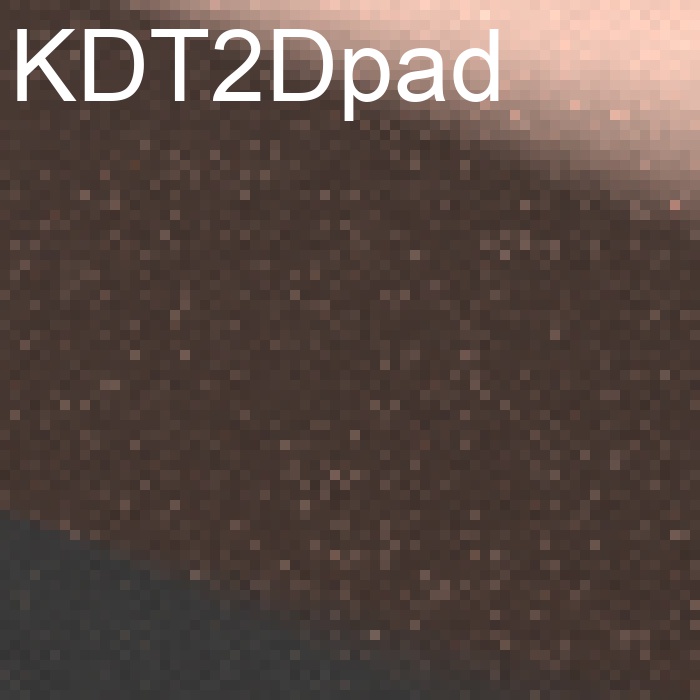}&
\includegraphics[width=\zoomscale\linewidth]{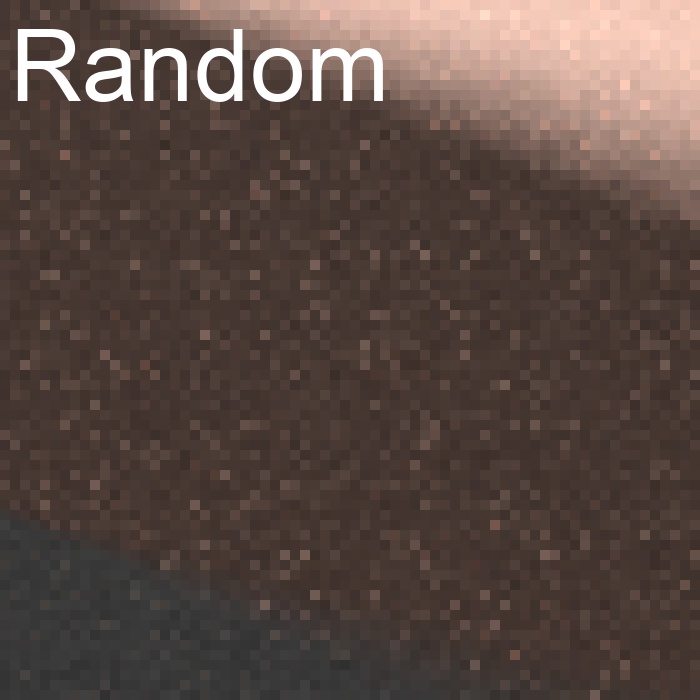}&
\includegraphics[width=\zoomscale\linewidth]{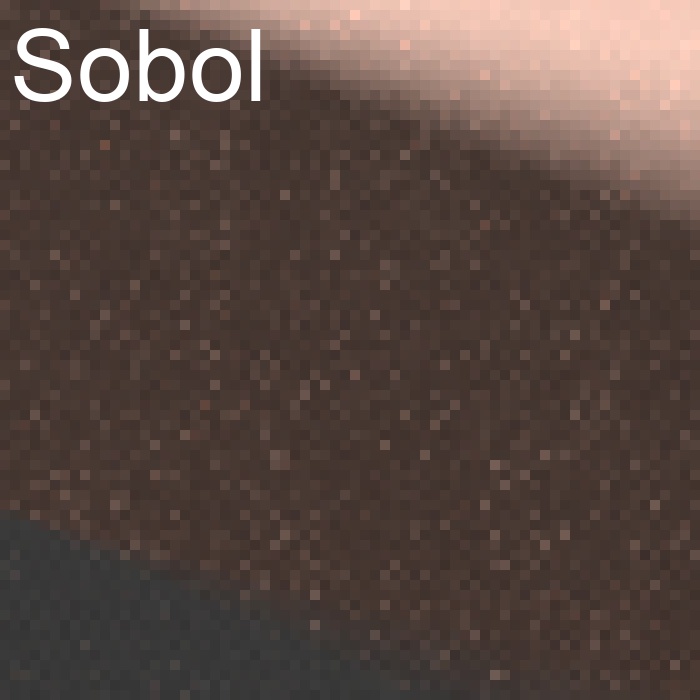} \\
\begin{turn}{90} 
128 spp
\end{turn} &
\includegraphics[width=\zoomscale\linewidth]{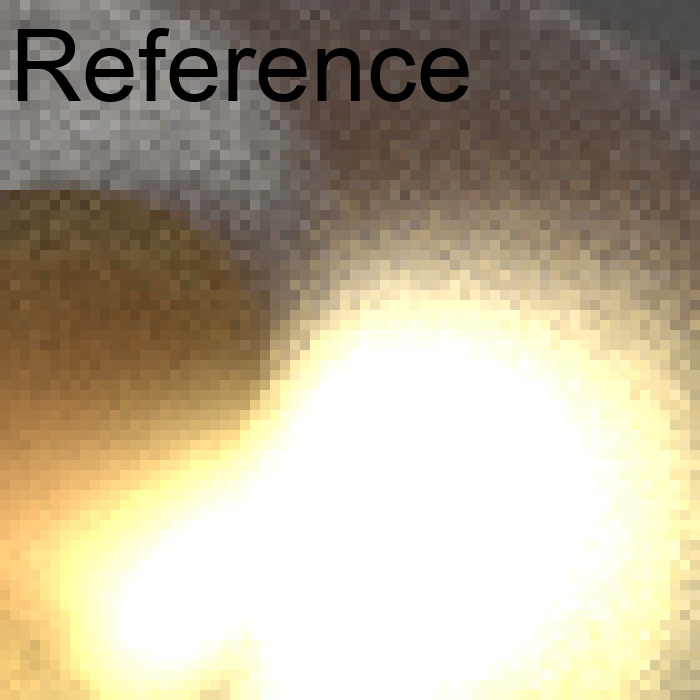}&
\includegraphics[width=\zoomscale\linewidth]{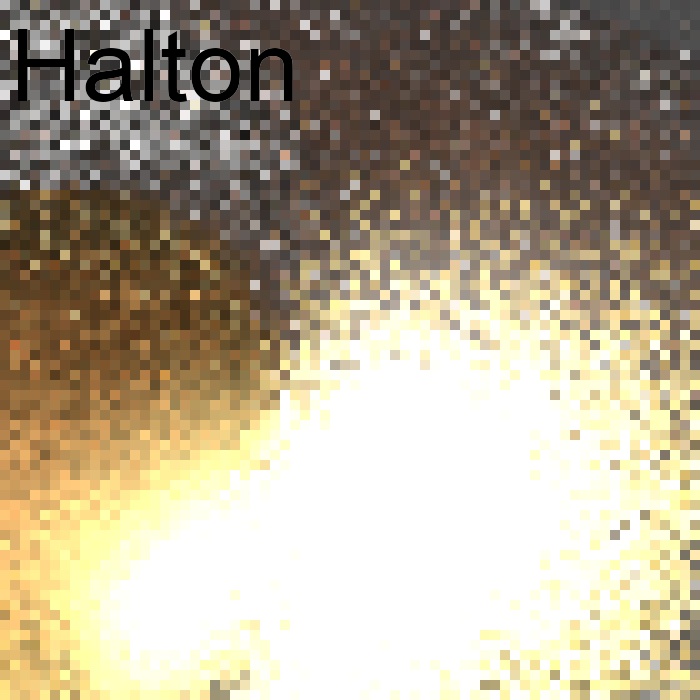}&
\includegraphics[width=\zoomscale\linewidth]{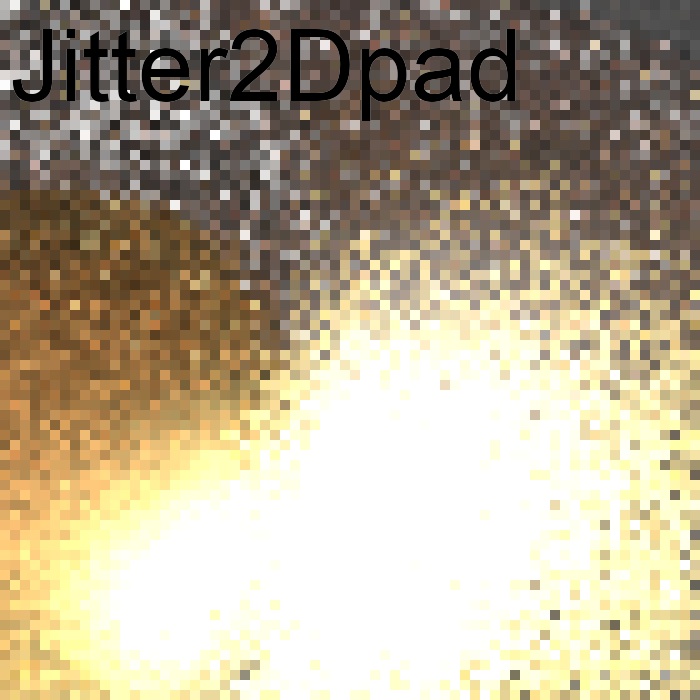}&
\includegraphics[width=\zoomscale\linewidth]{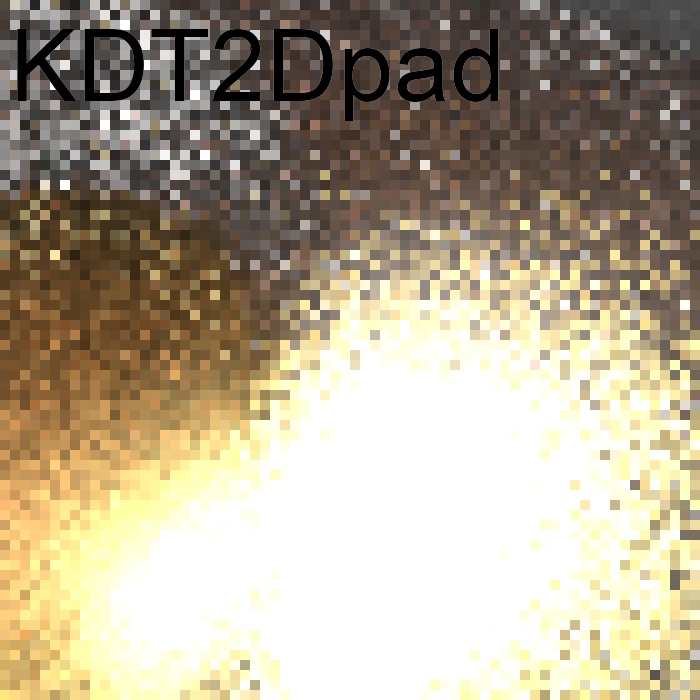}&
\includegraphics[width=\zoomscale\linewidth]{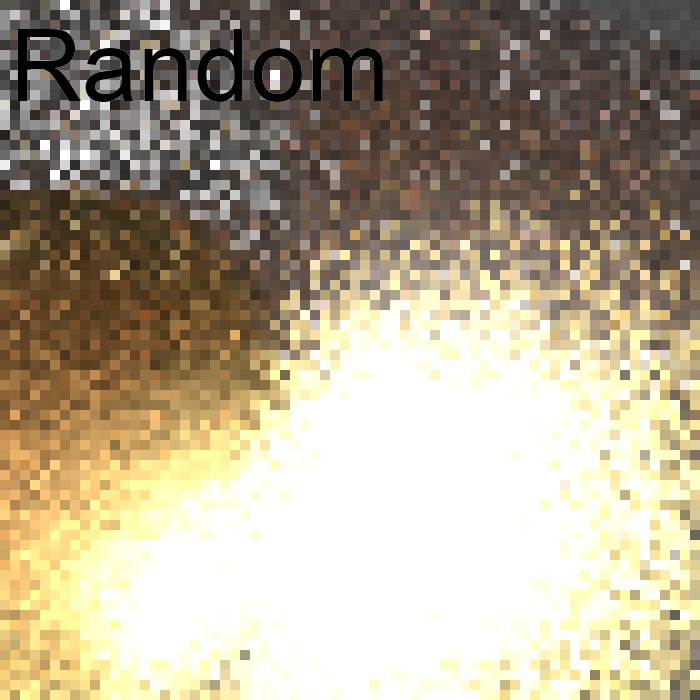}&
\includegraphics[width=\zoomscale\linewidth]{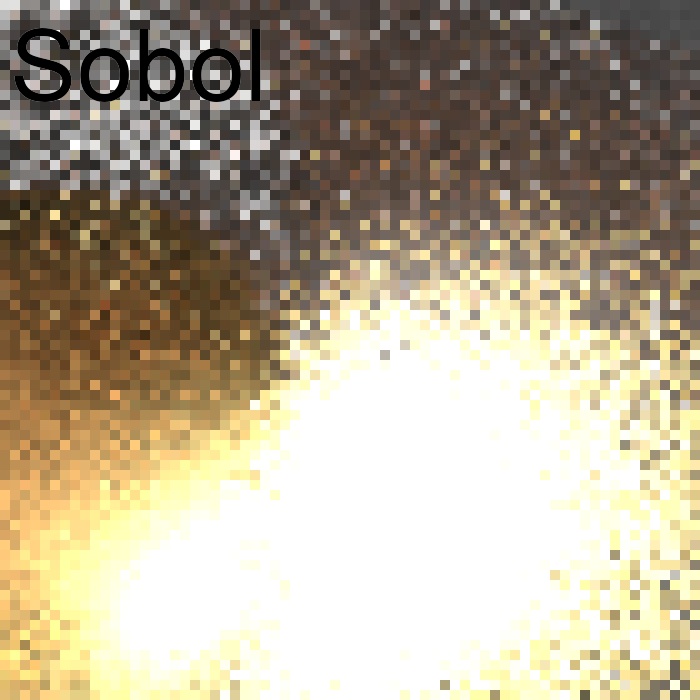} \\
\begin{turn}{90} 
256 spp
\end{turn} &
&
\includegraphics[width=\zoomscale\linewidth]{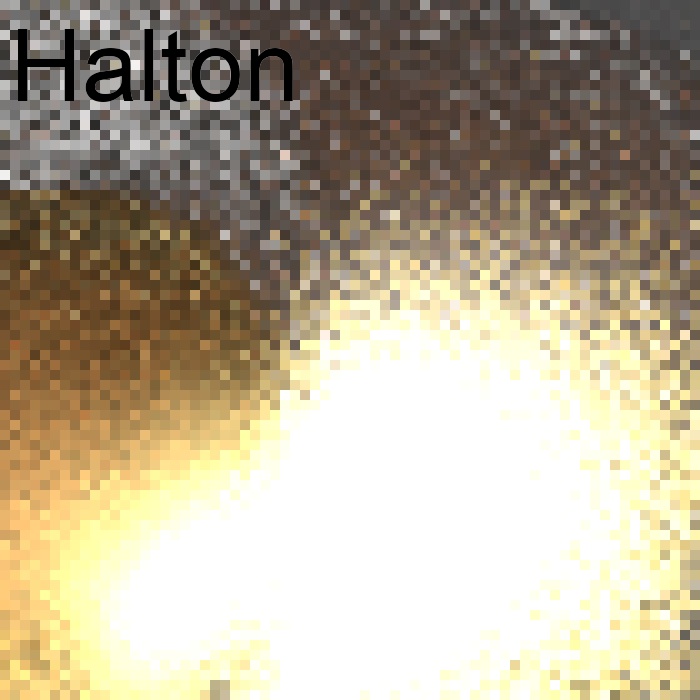}&
\includegraphics[width=\zoomscale\linewidth]{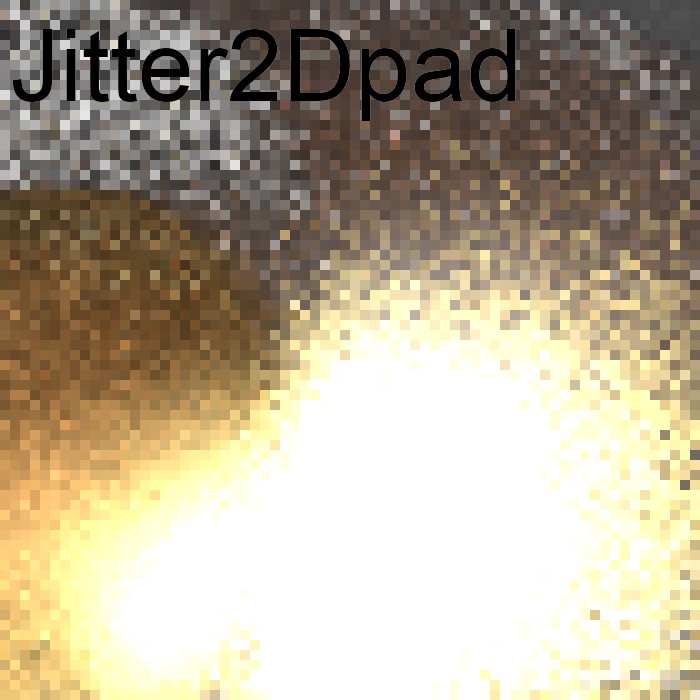}&
\includegraphics[width=\zoomscale\linewidth]{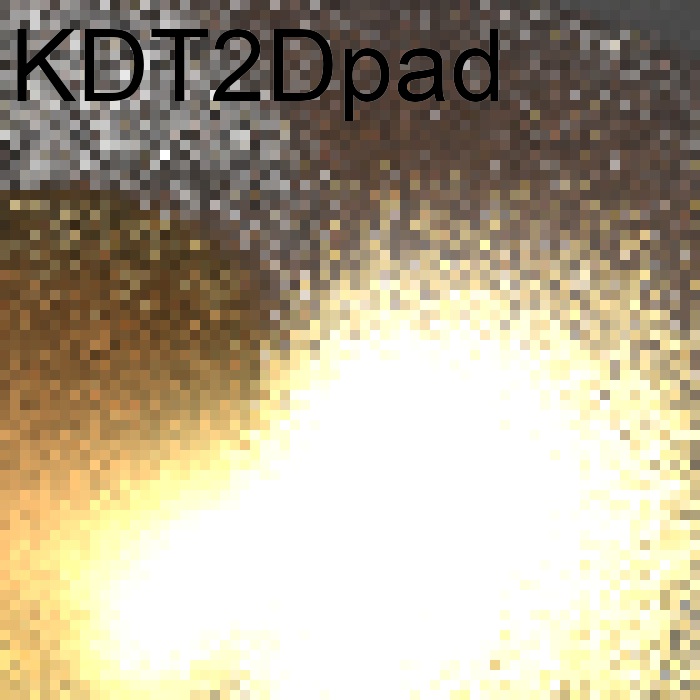}&
\includegraphics[width=\zoomscale\linewidth]{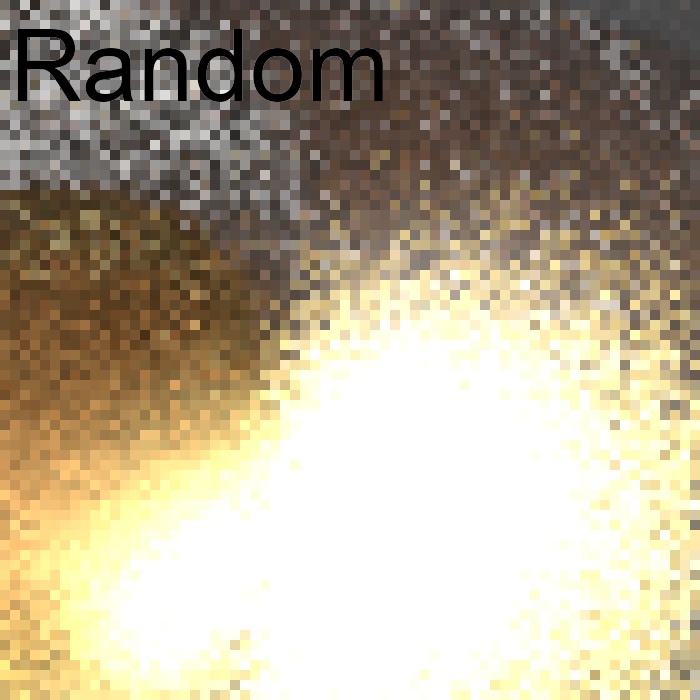}&
\includegraphics[width=\zoomscale\linewidth]{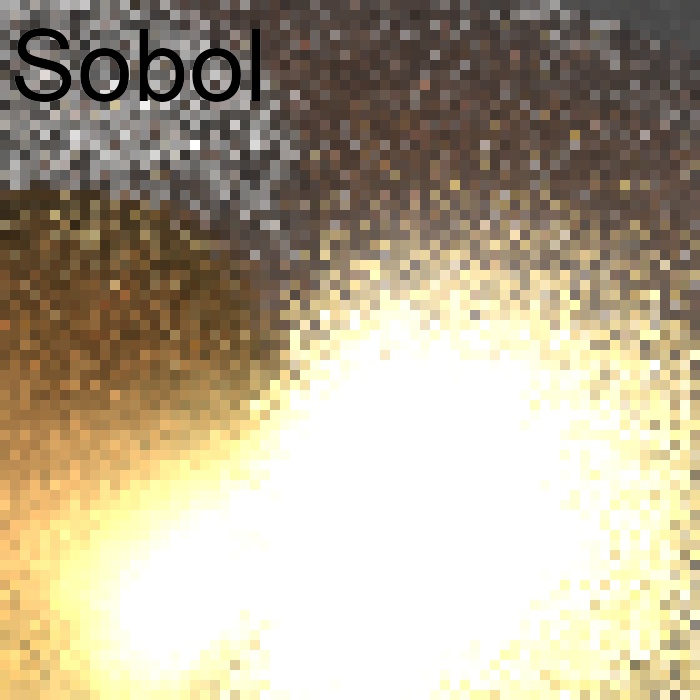} \\

\begin{turn}{90} 
128 spp
\end{turn} &
\includegraphics[width=\zoomscale\linewidth]{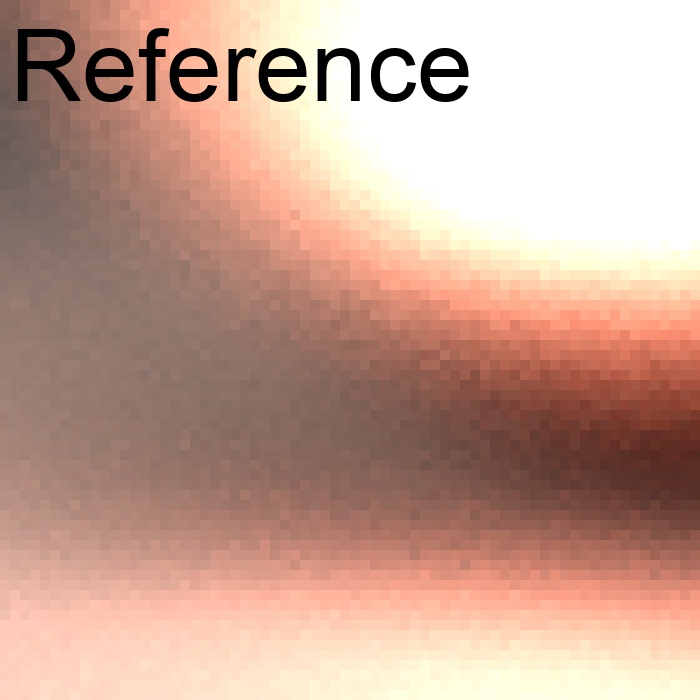}&
\includegraphics[width=\zoomscale\linewidth]{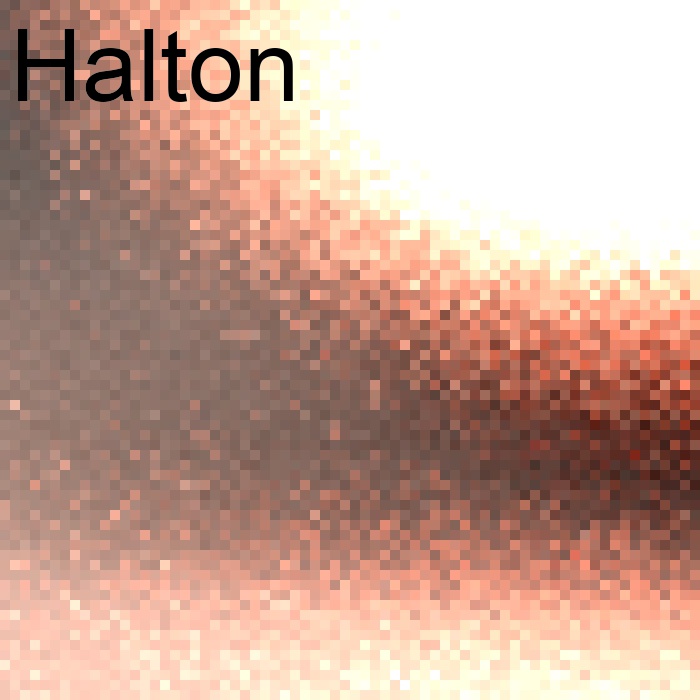}&
\includegraphics[width=\zoomscale\linewidth]{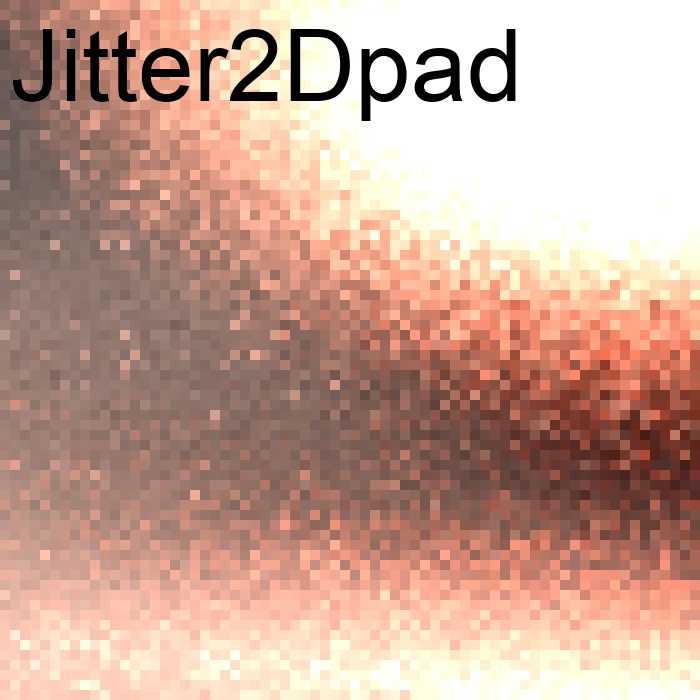}&
\includegraphics[width=\zoomscale\linewidth]{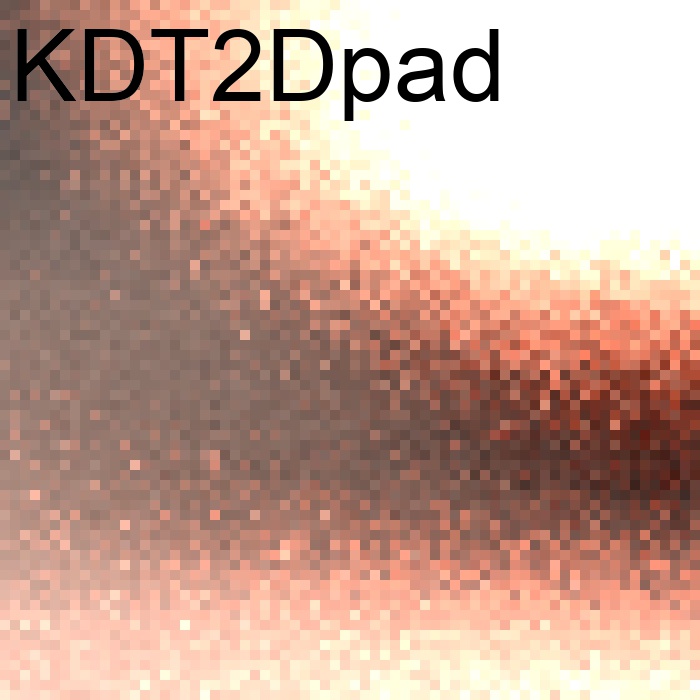}&
\includegraphics[width=\zoomscale\linewidth]{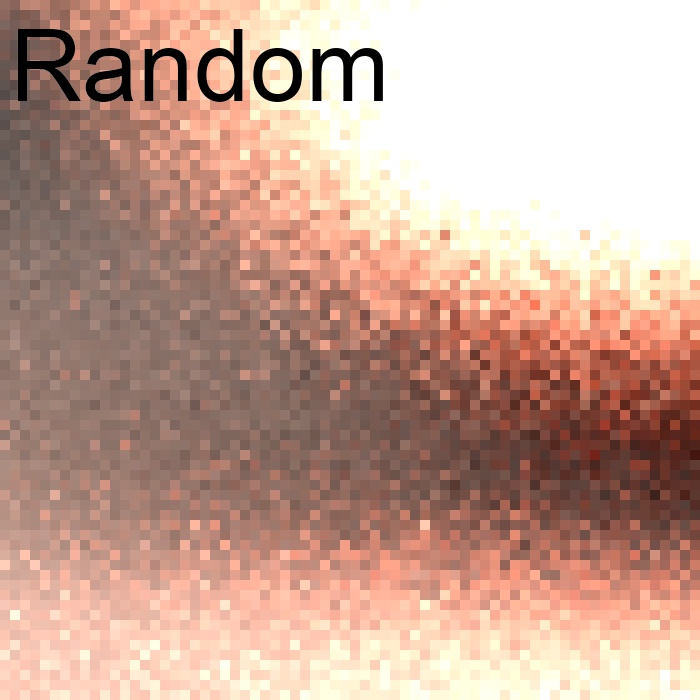}&
\includegraphics[width=\zoomscale\linewidth]{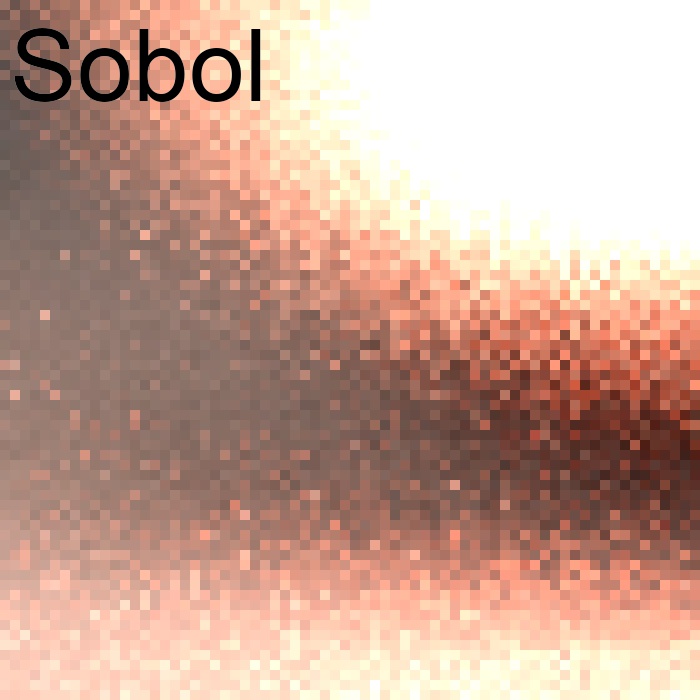} \\
\begin{turn}{90} 
256 spp
\end{turn} &
&
\includegraphics[width=\zoomscale\linewidth]{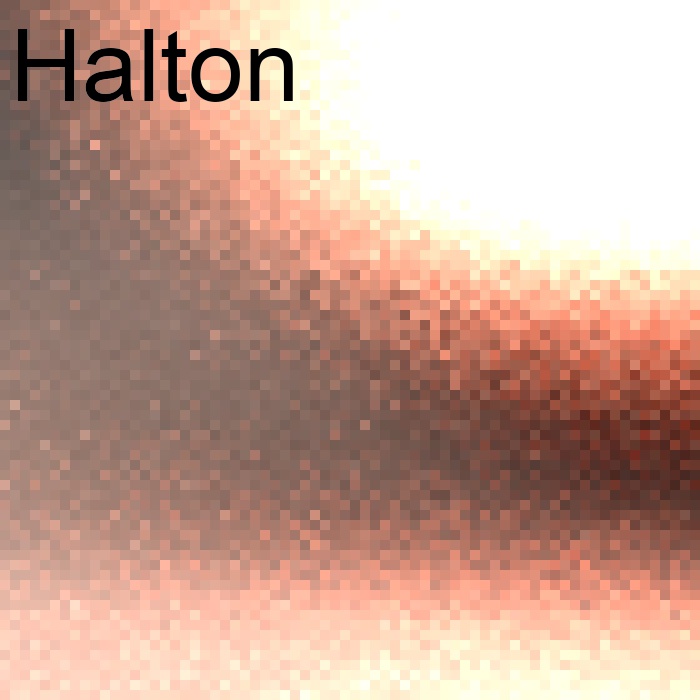}&
\includegraphics[width=\zoomscale\linewidth]{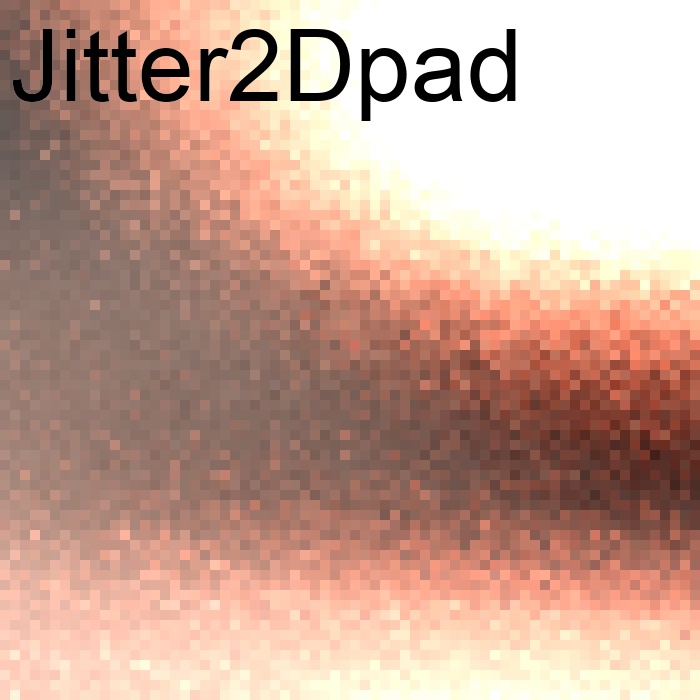}&
\includegraphics[width=\zoomscale\linewidth]{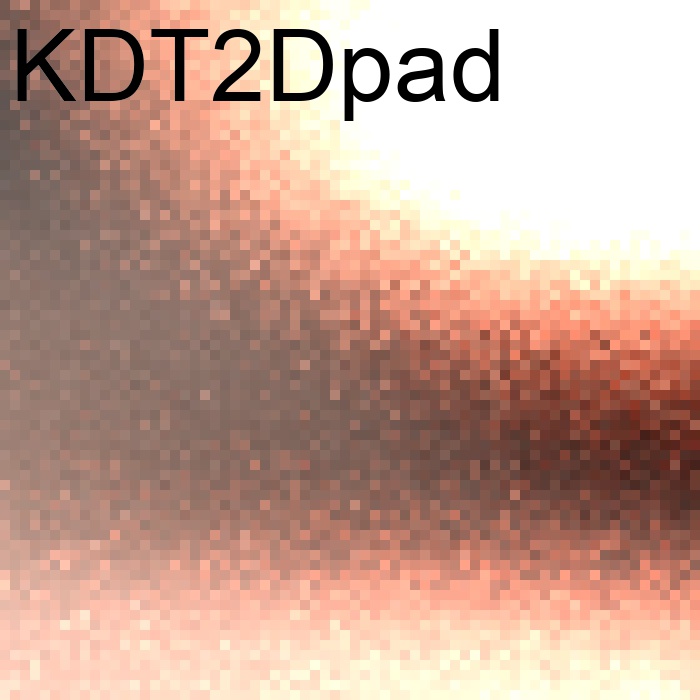}&
\includegraphics[width=\zoomscale\linewidth]{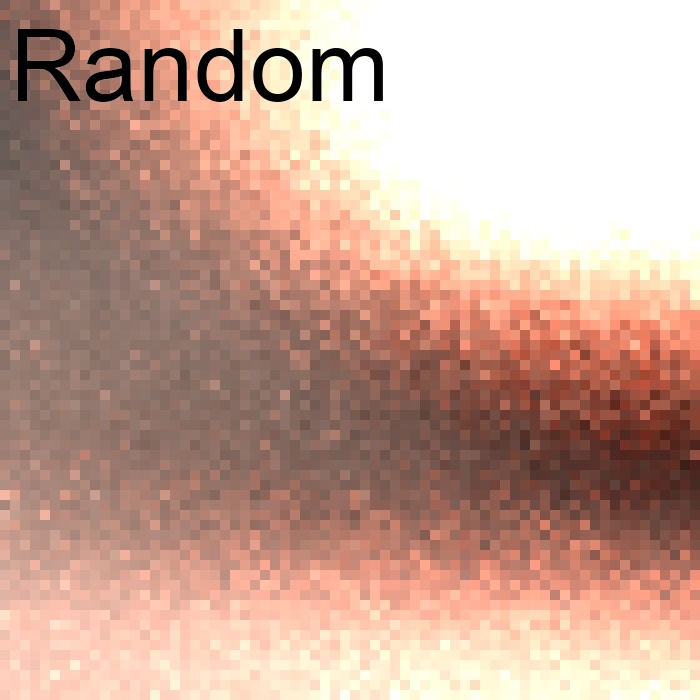}&
\includegraphics[width=\zoomscale\linewidth]{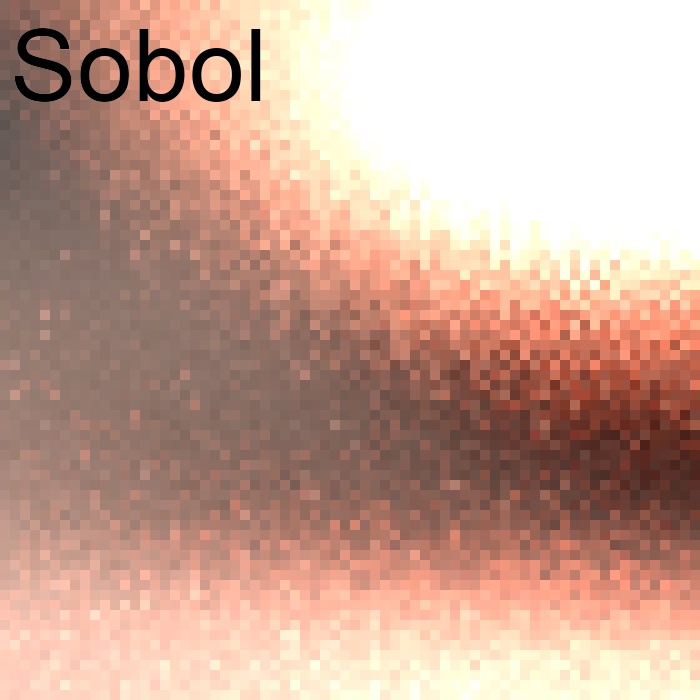} \\
\end{tabular}
\caption{\label{fig:orb-gloss-dof-insets} The figure shows several insets from the ORB-GLOSS-DOF scene shown in Figure~\ref{fig:orb-gloss-fov}. }
\end{figure*}

%___________________________________________________

\subsection{Discussion}
\paragraph{No parameters to tune}
One advantage of our method over samplers such as Halton and Sobol is that it is parameter-free. The output of the sampler only depends on $n$ and $d$. We experimented with several variants such as replacing the binary tree structure with a ternary structure, combining the use of Halton sampling in image space with KDT in path space, etc. None of these variants provided a notable gain in performance.

\paragraph{Generalization of jittered sampling to arbitrary $n$}
Our method essentially provides a generalization of the lattice-based \emph{jittered} sampling strategy to arbitrary $n$, thus breaking the lattice symmetry and partly amending the curse of dimensionality from which the latter suffers. Nevertheless, as the domain dimensionality increases the sample count should proportionally increase as to avoid performance degradation of the method to that of random sampling for the latter $d-\log_2 n$ subspace projections, depending on the dimension partitioning order.

\paragraph{Limitation - requires large $n$}
Although our stratification does not impose constraints on the number of samples, the benefit due to the stratification vanishes as $n \ll d$.~e.g.~if $n=1$, our method is equivalent to random sampling regardless of $d$. For large $d$, unless a large number of samples are drawn, the strata induced by the kd-tree tend to be large thereby weakening the effect of stratification. The full power of our method is exploited when estimating discontinuous integrands in high dimensions with an enormous computational budget. However, the ``padded'' version of our sampler overcomes this limitation by sacrificing the ability to sample in high dimensional spaces. 

\paragraph{Sample sets, not sample sequences}
It should be noted that jittered kd-tree stratification results in \emph{sample sets}, not sample sequences, ostensibly deeming it inappropriate for adaptive sampling. Nevertheless, the recursive nature of our method enables extensions able to generate sample sequences, and allows for adaptive sampling capabilities with minor modifications to the original method, which are left for future work. Furthermore, adaptive dimension partitioning strategies based on our kd-tree method seem attractive for cases where the integrand is $k$-dimensional additive, $k<d$.

\paragraph{Theoretical upper bound}
Theorem~\ref{thm:discr} provides a worst-case upper bound for the star-discrepancy of our method. We conjecture that the expected star-discrepancy satisfies the same asymptotic bounds as jittered sampling~\cite{pausinger2016discrepancy}, as observed by the epirical discrepancy computation of Fig.~\ref{fig:emp_discr}.

%-------------------------------------------------------------------------
\section{Conclusions \& future work}
We have presented a novel stratification method for sampling the $d$-dimensional hypercube, with a theoretical upper bound on its $L_{\infty}$ star-discrepancy. Our sampling algorithm is simple, parallelisable and we have presented comprehensive qualitative and qualitative comparisons to demonstrate that it performs comparably with state-of-the-art sampling methods on analytical tests as well as complex scenes. We believe that this work will inspire future work on how the kd-tree strata might be interleaved or how this scheme can be combined with existing algorithms (as in the case of multijitter).

%-------------------------------------------------------------------------

\small
\bibliographystyle{jcgt}
\bibliography{AK_DD_KS_JITTERED_K-D_TREE_STRATIFICATION}

%-------------------------------------------------------------------------

%\section*{Index of Supplemental Materials}
%
%The supplementary material contains an html viewer for thorough examination of the MSE performance (RGB-MSE and Log-Luminance-MSE) of different samplers over user-selected regions of the scenes provided. Specifically:
%\begin{itemize}
%    \item \texttt{visualize\_HDR.html} provides a web interface for the real-time visualization of the RGB-MSE and Log-Luminance-MSE metrics, along with enlarged versions of the user-selected region of interest for each of the samplers compared, for easy visual examination. The accompanying \texttt{Readme.pdf} provides instructions of use.
%    \item \texttt{orb/, orb-gloss/, orb-gloss-dof/, pavilion/} contain \texttt{.hdr} renderings for the ORB, ORB-GLOSS, ORB-GLOSS-DOF, and PAVILION scenes, respectively, with one rendering per sampler compared, plus a reference rendering. See main body for spp details.
%\end{itemize}

\section*{Author Contact Information}

\hspace{-2mm}\begin{tabular}{p{0.5\textwidth}p{0.5\textwidth}}
Alexandros D. Keros \newline
University of Edinburgh \newline
Informatics Forum, 10 Crichton St. \newline
Ediburgh, EH8 9AB, UK \newline
\href{mailto:a.d.keros@sms.ed.ac.uk}{a.d.keros@sms.ed.ac.uk}\newline
&
Divakaran Divakaran \newline
University of Edinburgh \newline
Informatics Forum, 10 Crichton St. \newline
Ediburgh, EH8 9AB, UK \newline
\href{mailto:ddivakar@staffmail.ed.ac.uk}{ddivakar@staffmail.ed.ac.uk}\newline
\\
Kartic Subr \newline
University of Edinburgh \newline
Informatics Forum, 10 Crichton St. \newline
Ediburgh, EH8 9AB, UK \newline
\href{mailto:K.Subr@ed.ac.uk}{K.Subr@ed.ac.uk}\newline
%\href{http://homepages.inf.ed.ac.uk/ksubr/}
\end{tabular}

%\afterdoc

\end{document}